\documentclass[11pt]{article}

\usepackage[margin=1in]{geometry}
\usepackage{amsmath,amsfonts,amssymb,amsthm}
\usepackage{array}
\usepackage{bbm}
\usepackage{bm}
\usepackage{booktabs}
\usepackage{caption}
\usepackage{color}
\usepackage{dsfont}
\usepackage[shortlabels]{enumitem}
\usepackage{fancyvrb}
\usepackage{graphicx}
\usepackage{hhline}
\usepackage[ruled]{algorithm}
\usepackage[noend]{algpseudocode}
\usepackage{listings}
\usepackage{mathrsfs}
\usepackage{mathtools}
\usepackage{multirow}
\usepackage{natbib}
\usepackage{pgfplots}
\usepackage{rotating}
\usepackage{soul}
\usepackage{subfigure}
\usepackage{tikz}
\usepackage{url}
\usepackage{xspace}

\pgfplotsset{compat=1.18}

\newtheorem{theorem}{Theorem}
\newtheorem{lemma}{Lemma}
\newtheorem{proposition}{Proposition}

\newtheorem{assumption}{Assumption}
\newtheorem{definition}{Definition}

\newtheorem{remark}{Remark}

\definecolor{strcolor}{rgb}{0.6, 0.2, 0.6}
\definecolor{commentcolor}{rgb}{0.3125, 0.5, 0.3125}
\definecolor{keycol}{rgb}{0, 0, 1}

\newcommand {\bea}{\begin{eqnarray}}
	\newcommand {\eea}{\end{eqnarray}}

\DeclareMathOperator{\argmax}{\arg\max}

\newcommand \ourfrmwrk{OSSA}
\newcommand \oursimplerfrmwk{FOSSA}

\newcommand \truemean{\mathbf{M}_{0}}

\def\blot{\quad \mbox{$\vcenter{ \vbox{ \hrule height.4pt
				\hbox{\vrule width.4pt height.9ex \kern.9ex \vrule width.4pt}
				\hrule height.4pt}}$}}

 \bibpunct[, ]{(}{)}{,}{a}{}{,}

\title{One-Shot Screening of Simulated Systems for Acceptability}
\author{Jinbo Zhao \and David J. Eckman}
\date{\normalsize Wm Michael Barnes '64 Department of Industrial and Systems Engineering\\
  Texas A\&M University}

\begin{document}

\maketitle

\begin{abstract}
We introduce a general-purpose framework for designing screening procedures for problems featuring a finite set of simulated systems, a.k.a.\ ranking-and-selection problems.
The framework offers a novel perspective on screening in which decisions to retain or eliminate systems are based on confidence regions for the unknown problem instance rather than comparisons of estimated performances.
Specifically, a system is retained if it has acceptable performance under some plausible configuration of response vectors contained in the confidence region.
This perspective facilitates the design of procedures that guarantee to return either all acceptable systems, or each acceptable system, with high probability and accommodates many well-studied definitions of acceptability, including feasibility with respect to stochastic constraints and optimality with respect to one or more objectives.
We further study a subclass of the framework that yields simple and computationally efficient screening procedures that often have lower-order time complexity than existing methods and naturally supports parallelization without loss of screening power.
We demonstrate the effectiveness and efficiency of the procedures through numerical experiments.
\end{abstract}


\section{Introduction}

In the classical ranking-and-selection (R\&S) problem, a decision maker chooses from among a finite number of simulated systems (i.e., designs) and has sufficient computational resources to simulate all of them to some degree. 
Screening (a.k.a.\ subset selection) refers to the process of removing from consideration systems having poor estimated performance, so that the decision maker can focus on those with more acceptable performance. 
Screening procedures serve an important role supporting the use of simulation optimization for decision making under uncertainty.
For instance, screening procedures can quickly and cheaply pare down the number of candidate systems before running a selection procedure that ultimately recommends a single system as the best.
Or, the decision maker may examine the subset of systems returned by a screening procedure and choose one from it based on some secondary qualitative criteria.
In other situations, such as when the decision maker will be making a sequence of similar decisions over time, the returned subset can serve as an assortment of purportedly good systems from which the decision maker can choose again and again, depending on the scenario.

Screening procedures are readily available in commercial simulation software like Simio\textsuperscript{\texttrademark} and Arena\textsuperscript{\texttrademark} \citep{smith2018simio} and widely used in practice. Among the most common ways practitioners use screening procedures is to ``clean up'' after performing an exploratory simulation experiment, as a way to reduce the number of contending systems. We refer to this setting in which simulation data is collected in a single stage and subsequently post-processed as \emph{one-shot} screening. Alternatively, some screening procedures adaptively eliminate (screen out) systems over time as additional simulation data become available. The Extended Screen-to-the-Best procedure of \cite{boesel2003using} and the bisection Parallel Adaptive Survivor Selection (bi-PASS) procedure of \cite{pei2022parallel} are arguably the state-of-the-art one-shot and adaptive screening procedures, respectively.
Both procedures are implemented in Simio and designed for screening out systems whose performance---defined as the expectation of a scalar simulation output---is suboptimal. Screening procedures for other definitions of ``acceptable performance'' are relatively nonexistent.

This paper introduces a general-purpose framework, called \emph{one-shot screening for acceptability} (\ourfrmwrk), in which screening procedures can be designed with ease to deliver guarantees on returning systems with acceptable performance with high probability.
A strength of the \ourfrmwrk\ framework is its ability to accommodate unequal sample sizes across systems, dependent simulation outputs and different definitions of acceptability, including those based on multi-valued system performance.
We believe the name \ourfrmwrk---Latin for ``bones''---is especially fitting, as the generality of the framework comes from reducing screening to its primitives: defining acceptability and quantifying uncertainty. 
A procedure designed under the \ourfrmwrk\ framework screens each system by checking whether that system is plausibly acceptable given the uncertainty associated with the estimated response vectors of all systems.
More precisely, when screening a system, the procedure constructs a confidence region for the response vectors of all systems and returns that system if the confidence region contains a configuration in which that system is acceptable.
\ourfrmwrk\ procedures involve checking the intersection of sets of \emph{configurations} of response vectors and therefore operate at a different level of abstraction from most screening procedures.
More powerful procedures can be designed within the framework by tailoring the confidence regions to the given definition of acceptability.
A well-known special case within this framework is the Extended Screen-to-the-Best procedure of \cite{boesel2003using}.
\ourfrmwrk\ procedures can also be applied to non-simulation data, provided the requisite confidence region can be constructed.

The \ourfrmwrk\ framework resembles the plausible screening approach of \cite{eckman2022plausible} in that both involve checking how plausible it is that a solution is acceptable in light of the uncertainty surrounding the true configuration.
The main differences are that the \ourfrmwrk\ framework entails simulating all systems under consideration---thus, unsimulated systems cannot be screened out---and makes no assumption (or use) of functional information relating the responses of systems, making it applicable to the analysis of black-box simulation models.
Other relationships between screening and plausible screening have been studied and some plausible-screening-inspired screening procedures proposed \citep{eckman2020revisiting}.
Unlike plausible screening, the \ourfrmwrk\ framework places less emphasis on the choice of the discrepancy, which measures deviations from the estimated configuration, and more on the geometry of the confidence region dictated by the discrepancy and how it can be exploited for more efficient and powerful screening.

We identify and delve into a subclass of the framework that facilitates the design of simple, yet computationally efficient, screening procedures for a wide array of R\&S problems.
Specifically, we investigate procedures that simulate systems independently and construct marginal confidence regions for the response vectors of each system.
We refer to this subclass as \emph{fragmented} \ourfrmwrk\ (\oursimplerfrmwk) because of how the overall allowed error is split across systems.
We find that the time complexity of the procedures resulting from these design choices tends to be of a lower order than existing methods.
For instance, the \oursimplerfrmwk\ procedure for optimizing a single response has a linear dependence on the number of systems versus the quadratic dependence of the Extended Screen-to-the-Best procedure, which makes pairwise comparisons.
For optimization with stochastic constraints and multi-objective optimization, the complexity of the \oursimplerfrmwk\ procedures is linear and (at worst) sub-quadratic in the number of systems, respectively. 
The procedures also naturally lend themselves to parallelization via a straightforward divide-and-conquer scheme with no loss of screening power. 

We view the chief contributions of this paper being the following:
\begin{itemize}
    \item We present a design framework for one-shot screening procedures that encompasses problems concerning optimization and feasibility and provides a template for designing new procedures for yet unstudied problems.
    \item We introduce an innovative design perspective centered around confidence regions for the unknown problem instance
    and demonstrate how the construction and geometry of these confidence regions can influence the screening power of fixed-confidence screening procedures.
    \item We investigate a subclass within the framework that yields powerful and easy-to-implement screening procedures, including ones that are the first known one-shot screening procedures for optimization with stochastic constraints and multi-objective optimization.
\end{itemize}

The rest of this paper is organized as follows:
We review relevant literature in Section~\ref{sec:literature}.
In Section~\ref{sec:ossa}, we introduce relevant notation and guarantees and present the \ourfrmwrk\ framework.
We then introduce the FOSSA subclass of procedures, which exploit the geometry of a simplifying confidence region, and discuss extensions of the framework to the adaptive sampling setting.
In Section~\ref{sec:two examples}, we detail \oursimplerfrmwk\ procedures for simulation-optimization problems of broad interest, specifically, optimization subject to stochastic constraints and multi-objective optimization. (A separate discussion of the well-studied case of optimizing a single response appears in the electronic companion.)
In Section~\ref{sec:num_exp}, we apply the algorithms to a simulation study of a multi-class queueing system subject to blocking.
We conclude in Section~\ref{sec:conclusion} and point out directions for future research.

\section{Background}
\label{sec:literature}

The approach of returning a subset of promising systems originated as an alternative to selection procedures designed under the indifference-zone formulation of \cite{bechhofer1954single}, in which a single system is selected and there is assumed to be a clear best system \citep{gupta1965some,gupta1985subset}.
Screening and selection procedures were initially developed for problems featuring a small number of systems that could be evaluated only by performing expensive physical experiments. They were later co-opted for use in computer simulation experiments to solve optimization problems with thousands or even millions of systems \citep{Ni2017,avci2024simulation}.
Screening and selection procedures also arise in non-simulation contexts such $A/B/n$ testing \citep{Russac2021} and the use of large language models (LLMs) as proxies for human evaluators \citep{li2025efficientbudgetallocationlargescale}.

\subsection{One-Shot Versus Adaptive Screening}

Screening procedures can be adaptive or non-adaptive, depending on how they allocate simulation effort (i.e., replications) across systems. 
Non-adaptive (one-shot) screening procedures take a fixed, possibly unequal, number of replications from each system and determine which systems to eliminate based on summary statistics.
Adaptive screening procedures, on the other hand, obtain samples from systems either in stages \citep{dudewicz:twostage, rinott:twostage} or sequentially \citep{pei2022parallel} with sample sizes being based on the data collected.
As a result, adaptive screening procedures can direct more simulation effort to systems with promising, yet uncertain, performance.
Although one-shot procedures may have less screening power than adaptive procedures, they can be easily parallelized and exemplify the exploratory role screening can serve in situations where the decision maker may have already obtained data from a preliminary experiment.
Although our focus in this paper is not on how the data are collected but on how the returned subset of systems is formed using said data, the OSSA framework is compatible with adaptive sampling schemes; see Section~\ref{sec:the Envelope CI} in the electronic companion for more details.

The majority of the literature on screening has studied the problem from a frequentist perspective, where the systems' true responses are fixed, but unknown, and the provided statistical guarantees are with respect to repeated runs of the data-collection and subset-construction processes.
For example, the well-known probability of correct selection (PCS) guarantee implies that the system with the best (scalar) response will be returned with at least some user-specified probability, regardless of the true configuration.
\cite{EckmanHenderson2021} and \cite{zhao23screening} provide an overview and discussion of other fixed-confidence frequentist guarantees for screening. 
We adopt the frequentist treatment in this paper.
For background on the alternative Bayesian perspective of R\&S problems, see \cite{miescke:bayesdesigns99} and \cite{chen2015ranking}.

\subsection{R\&S with a Single Response}

Most conventional R\&S procedures focus on the setting in which a decision maker is interested in systems' \textit{expected} outputs. Other problem variations have also been studied, such as comparing systems based on quantiles \citep{batur2010quantile, peng2021efficient, shin2022practical} or the variances \citep{batur2010mean} of their output distributions. 
Our proposed framework can accommodate most risk measures, including the expected value, variance, quantile, tail probability, and conditional value at risk, provided valid finite-sample or asymptotic confidence regions can be constructed. 
In each case, the same machinery is applied and analogous statistical guarantees are delivered.

Early research on screening procedures was primarily restricted to problems in which each system has a single response, given by an expected value. Within the context of optimization, a common goal has been to return the system with the best response with high probability \citep{boesel2003using}. Other goals related to optimization involving alternative definitions of \emph{acceptable} performance have been studied; prominent examples include returning systems with near-optimal responses \citep{lam1986new,sullivan1989restricted,zhao23screening}, top-$m$ responses  \citep{koenig1985procedure,chen2008efficient}, or low expected opportunity cost \citep{gao2015efficient}. A related problem in this area is feasibility determination, i.e., identifying systems whose responses are better than some user-specified threshold \citep{szechtman2008new,solow2021novel}.
Goals of optimization and feasibility can also be combined. For example, \cite{yan2012efficient} and \cite{jia2013efficient} study a setting in which systems have a deterministic primary response and a secondary response associated with a simulation output. The goals studied in these papers are to return the set of systems with the $m$ best primary responses among either those having good enough secondary responses (i.e., better than a threshold) or top-$g$ secondary responses, respectively.

\subsection{R\&S with Multiple Responses}

Many real-world problems tackled using simulation involve multiple responses of interest \citep{butler2001multiple}.
For example, supply-chain planners may care about transportation and holding costs and order fill rates when evaluating inventory control policies.
Screening procedures exist for feasibility determination problems with multiple thresholds on performance \citep{baturkim2010, gao2017efficient}.
A problem formulation that combines optimization and feasibility is the optimization of a single primary responses subject to one or more so-called stochastic constraints on secondary responses.
Selection procedures that deliver finite-sample fixed-confidence guarantees for this problem include those of \cite{andradottir2010fully,healey2014selection}, and \cite{hong2015chance}.
These selection procedures treat the feasibility and optimality aspects of the problem separately, splitting the allowable probability of making an error between the two. They feature a first stage in which clearly infeasible systems are screened out, followed by a second stage in which a conventional selection procedure selects from among the survivors.
Any dependence between primary and secondary responses is either ignored or avoided by the use of the Bonferroni inequality.
\cite{lee2012approximate} and \cite{pasupathy2014stochastically} study ways to allocate simulation effort for this class of problems with the goal of maximizing the probability of correct selection.
Another approach involves using bootstrapping to construct a subset whose estimated probability of containing the best system is sufficiently high \citep{currie2021practical}.
In this paper, we introduce what we believe to be the first screening procedure with finite-sample guarantees on returning a set containing the (constrained) optimal solution with high probability while jointly handling the uncertainty about the primary and secondary responses.

Another problem featuring multiple responses is multi-objective simulation optimization.
Unlike in optimization with stochastic constraints, the decision maker is interested in finding a set of systems, namely, the set of non-dominated or efficient solutions \citep{hunter2019introduction}. 
Among selection procedures for multi-objective optimization, the typical goal is to return the \emph{exact} set of efficient solutions with high probability \citep{wang2017sequential,andradottir2021pareto} or to allocate simulation effort to maximize this probability
\citep{lee2010finding,li2018optimal,applegate2020multi}.
Here too, to the best of our knowledge, no screening procedures exist for the multi-objective R\&S problem that either guarantee to return a set containing all efficient systems or guarantee that each efficient system is returned with high probability.
We present two new screening procedures that fill this gap.

\section{One-Shot Screening for Acceptability}
\label{sec:ossa}

In this section, we introduce the problem of screening for acceptability and the OSSA/FOSSA frameworks in their fullest abstraction, with concrete procedures given later in Sections~\ref{sec:two examples} and \ref{sec:single_resp}.
We consider $k$ systems, each corresponding to a specific parameter setting of a simulation model.
The performance of each system is described by $d$ responses.
For System $i$, $i=1,2,\dots,k$, let the true (but unknown) response vector be denoted by $\boldsymbol{\theta}_{i}=(\theta_{i1},\theta_{i2},\dots,\theta_{id})^\intercal\in\mathds{R}^{d}$.
We assume that $n_i$ replications of System $i$ were simulated and produced independent and identically distributed (i.i.d.)\ output vectors $\boldsymbol{X}_{i1}, \boldsymbol{X}_{i2},\dots,\boldsymbol{X}_{in_i}$, where $\boldsymbol{X}_{i\ell} \in \mathds{R}^s$ for $\ell = 1, 2, \ldots, n_i$.
Each component of $\boldsymbol{\theta}_{i}$ is a statistical functional of one or more components of $\boldsymbol{X}_{i1}$, such as an expectation, variance, or quantile.
Although a system's true response vector is unknown, it can be estimated using these simulation outputs.
We refer to the matrix 
$\truemean=(\boldsymbol{\theta}_{1},\boldsymbol{\theta}_{2},\dots,\boldsymbol{\theta}_{k})^{\intercal} \in \mathds{R}^{k\times d}$
as the true \emph{configuration} of response vectors and let $\widehat{\mathbf{M}}_0 = (\widehat{\boldsymbol{\theta}}_{1},\widehat{\boldsymbol{\theta}}_{2},\dots,\widehat{\boldsymbol{\theta}}_{k})^{\intercal} \in \mathds{R}^{k\times d}$ denote its estimate.

For a given configuration $\truemean$, let $\mathcal{A}(\truemean) \subseteq \{1,2, \dots,k\}$ denote the set of indices of systems whose responses are \emph{acceptable} to the decision maker, where the acceptability mapping $\mathcal{A} \colon \mathds{R}^{k \times d} \mapsto \mathscr{P}(\{1, 2, \ldots, k\})$ is known and $\mathscr{P}(\{1, 2, \ldots, k\})$ denotes the set of all possible subsets of systems.
In this section, we keep the definition of acceptability general to highlight the flexibility of the framework. 
In later sections, we present procedures tailored to specific forms of acceptability.
Common definitions of acceptability include, but are not limited to,
optimality with respect to a single response; 
optimality with respect to a single response subject to individual constraints on other responses;
and Pareto optimality, i.e., being non-dominated with respect to all responses.

Because the true configuration $\truemean$ is unknown, $\mathcal{A}(\truemean)$ cannot be identified with certainty given a finite sampling budget.
To tackle this challenge, one-shot screening procedures take as input the data $\{\boldsymbol{X}_{i\ell}\colon  i=1,2,\dots,k \text{ and } \ell=1,2,\dots,n_{i}\}$ and return a subset $\mathcal{S}\subseteq \{1,2, \dots,k\}$ accompanied by a certain statistical guarantee. 
We consider two types of guarantees:
\begin{definition} 
A subset $\mathcal{S}$ delivers the \emph{system-wise probability of acceptable selection (system-wise PAS) guarantee} if for any $\truemean \in \mathds{R}^{k\times d}$, $\mathrm{P}(i \in \mathcal{S})\geq 1-\alpha$ for all $i\in \mathcal{A}(\truemean)$.
\end{definition}
\begin{definition} 
A subset $\mathcal{S}$ delivers the \emph{set-wise probability of acceptable selection (set-wise PAS) guarantee} if for any $\truemean \in \mathds{R}^{k\times d}$, $\mathrm{P}(\mathcal{A}(\truemean) \subseteq \mathcal{S})\geq 1-\alpha$.
\end{definition}
In both guarantees, $1-\alpha \in (0,1)$ represents the decision maker's  desired level of confidence. 
The system-wise PAS guarantee ensures that \emph{each} acceptable system has a high (marginal) probability of being returned, whereas the set-wise PAS guarantee ensures that, with high probability, \emph{all} acceptable systems will be returned. 
The set-wise PAS guarantee implies the system-wise PAS guarantee.
When there are multiple acceptable systems, such as when acceptability is defined as Pareto optimality, decision makers may prefer the set-wise PAS guarantee.
Both the system-wise and set-wise PAS guarantees are trivially delivered by returning all systems, i.e., choosing $\mathcal{S} = \{1, 2, \ldots, k\}$, but a smaller subset is desirable.
We treat the problem under the frequentist framework, hence the probabilities in both guarantees are with respect to repeated experiments on the fixed distributions of $\boldsymbol{X}_{i1}$ for all $i=1,2,\dots,k$. More precisely, the probability statements pertain to the random returned subset $\mathcal{S}$, not $\mathcal{A}(\truemean)$, which is deterministic but unknown. 
Additional frequentist guarantees and their interrelationships are detailed in \cite{zhao23screening}.

\subsection{The OSSA Framework}
We introduce a design framework for one-shot screening procedures delivering either the system-wise or set-wise PAS guarantee.
The procedures do so by returning systems deemed plausibly acceptable, meaning that within a confidence region for the true configuration, $\truemean$, there exists a plausible configuration for which that system would be acceptable.
Put differently, when screening a given system, the procedures attempt to manipulate the estimated configuration $\widehat{\mathbf{M}}_0$ to make that system appear acceptable, but these perturbations are restricted based on the uncertainty associated with the unknown configuration and the decision maker's desired confidence level.\

For a given definition of acceptability, let $\mathds{A}_{i} \subseteq \mathds{R}^{k\times d}$ be the set of configurations for which System $i$ is acceptable,
i.e., $\mathds{A}_{i}=\{\mathbf{M}\in \mathds{R}^{k\times d}\colon i \in \mathcal{A}(\mathbf{M})\}$, where $\mathbf{M}=(\boldsymbol{m}_{1},\boldsymbol{m}_{2},\dots,\boldsymbol{m}_{k})^{\intercal}$ represents a generic configuration.
We call $\mathds{A}_{i}$ the \emph{acceptable region} of System $i$, which is fixed and known.
As an example, when acceptability is defined as optimality with respect to a single response, where smaller is better,
$\mathds{A}_{i}=\{\mathbf{M}\in \mathds{R}^{k \times 1}\colon m_{i}\leq m_{j}\ \mathrm{for\ all}\ j\neq i\}$. The acceptable regions for the two-system select-the-best case are shown in Figure~\ref{fig:Fig1} in Section~\ref{sec:single_resp} of the electronic companion. Acceptable regions have also been adopted by \cite{wang2026rankingandselectionmultiplecorrectanswers} for the problem of selecting one acceptable system, wherein $\mathds{A}_{i}$ is referred to as the \textit{acceptance set} of System $i$.

OSSA procedures screen each system under consideration to determine which of them should be returned.
When screening System $i$, $i = 1, 2, \ldots, k$, an OSSA procedure works with a $1-\alpha$ confidence region for $\truemean$, denoted by $\mathds{C}_{i} \subseteq \mathds{R}^{k\times d}$, where the value of $1-\alpha$ is that of the desired PAS guarantee.
The construction of $\mathds{C}_{i}$ depends on the distributional assumptions made about the simulation outputs and the guarantee sought; specific examples will be given in Section~\ref{subsetion:section 2.4}.

For either PAS guarantee, an OSSA procedure returns those systems whose acceptable region intersects with the corresponding $1-\alpha$ confidence region; i.e., $\mathcal{S}=\{i\colon \mathds{A}_{i}\cap \mathds{C}_{i}\neq \emptyset\}$.
In Theorems \ref{theorem1} and \ref{theorem2}, we establish that this returned subset delivers the desired guarantee.

\begin{theorem} 
\label{theorem1}
The subset $\mathcal{S}^{system}=\{i\colon \mathds{A}_{i} \cap \mathds{C}_{i} \neq \emptyset \}$ delivers the system-wise PAS guarantee.
\end{theorem}

\noindent \textbf{Proof}: For all $i\in \mathcal{A}(\truemean)$, 
\vspace{-1em}
\begin{align*}  
    \mathrm{P}(i \in \mathcal{S}^{system}) &= \mathrm{P}(\mathds{A}_{i} \cap \mathds{C}_{i} \neq \emptyset)\\
&\geq \mathrm{P}(\truemean \in \mathds{A}_{i} \cap \mathds{C}_{i})\\
&=\mathrm{P}(\truemean \in \mathds{C}_{i})\\
&\geq 1-\alpha,
\end{align*}
where the second equality comes from the fact that $i \in \mathcal{A}(\truemean)$ implies that $\truemean \in \mathds{A}_{i}$, and the second inequality comes from the definition of $\mathds{C}_{i}$ as a $1-\alpha$ confidence region for $\truemean$. $\square$

For the set-wise PAS guarantee, a common confidence region is used for screening all systems, i.e.,   $\mathds{C}_{1}=\mathds{C}_{2}=\dots=\mathds{C}_{k}$; we denote this common confidence region by $\mathds{C}_{u}$.

\begin{theorem}
\label{theorem2}
The subset $\mathcal{S}^{set}=\{i\colon \mathds{A}_{i} \cap \mathds{C}_{u} \neq \emptyset \}$ delivers the set-wise PAS guarantee.
\end{theorem}

\noindent \textbf{Proof}:
\vspace{-1em}
\begin{align*}  
\mathrm{P}&(\mathcal{A}(\truemean) \subseteq \mathcal{S}^{set}) \\
&= \mathrm{P}(\mathds{A}_{i} \cap \mathds{C}_{u} \neq \emptyset \mathrm{\ for\ all\ }i \in \mathcal{A}(\truemean))\\
&\geq\mathrm{P}(\truemean \in \mathds{A}_{i} \cap \mathds{C}_{u}  \mathrm{\ for\ all\ }i \in \mathcal{A}(\truemean))\\
&=\mathrm{P}(\truemean \in \mathds{C}_{u})\\
&\geq 1-\alpha.\
\square
\end{align*} 

\begin{remark}
The framework and proof techniques introduced here are extremely general as we have yet to make any assumptions about the form of $\mathds{A}_{i}$ and $\mathds{C}_{i}$.
In particular, the confidence regions can be constructed in any way as long as they contain $\mathbf{M}_{0}$ with probability exceeding $1-\alpha$.
The proofs of Theorems~\ref{theorem1} and \ref{theorem2} show that there are two sources of conservativeness with how OSSA procedures deliver PAS guarantees.
The first arises from controlling the error over all configurations by bounding $\sup\limits_{\mathbf{M}_{0} \in\mathds{A}_{i}}\mathrm{P}_{\mathbf{M}_{0}}(\mathds{A}_{i} \cap \mathds{C}_{i} = \emptyset)$ from above, without assuming anything about the geometry of $\mathds{A}_{i}$.
The second comes from the potential inexactness of $\mathds{C}_{i}$ as a $1-\alpha$ confidence region for $\truemean$.
We will later see situations when constructing $\mathds{C}_{i}$ where there is a trade-off between  having more exact confidence and having a geometry more closely aligned with that of $\mathds{A}_{i}$.
\end{remark}

Constructing finite-sample confidence regions often necessitates making assumptions about the distribution of simulation outputs.
Alternatively, bootstrapping or empirical likelihood methods can be applied with very mild assumptions on the distribution to construct regions with asymptotic confidence at or above $1-\alpha$ \citep{efron1994introduction,Owen1988empirical}.
The OSSA framework can be modified to use such regions so that the returned subset delivers either an asymptotic system-wise or set-wise PAS guarantee by altering only the last inequality in the proofs of Theorems~\ref{theorem1} and \ref{theorem2}.

The system-wise and set-wise PAS guarantees control the probability of incorrectly screening out acceptable systems given fixed sample sizes.
Alternatively, one may be interested in the probability that an unacceptable system incorrectly survives screening and how this probability decreases as sample sizes increase.
We say that a screening procedure is \textit{consistent} if it can asymptotically screen out any (non-borderline) unacceptable system.
This concept is made more rigorous in the following definition.
\begin{definition}
\label{def:consistency}
A subset $\mathcal{S}$ is \textit{consistent} if for any $\truemean \in \mathds{R}^{k\times d}$, $\mathrm{P}(i \in \mathcal{S})\rightarrow 0$ as $\min_{i=1,2,\dots,k}n_i\rightarrow \infty$ for all $i$ such that $\truemean \notin \mathrm{cl}(\mathds{A}_i)$, where $\mathrm{cl}(\mathds{A}_i)$ denotes the closure of $\mathds{A}_i$.
\end{definition}
The closure appearing in Definition~\ref{def:consistency} is necessary because for configurations in $\mathrm{cl}(\mathds{A}_i) \setminus \mathds{A}_i$, it is generally impossible to determine whether $\truemean$ belongs to $\mathds{A}_i$ with certainty as sample sizes approach infinity, even when a consistent estimator of $\truemean$ is available.

It is difficult to claim that OSSA procedures are consistent without considering specific definitions of acceptability or making stronger assumptions on the confidence regions they employ. Keeping with our goal of presenting OSSA procedures in their fullest generality, we elect not to do so here but instead offer some insights into when they can be presumed to be consistent.
Roughly speaking, OSSA procedures will be consistent if for any unacceptable system, $i$, the confidence region $\mathds{C}_i$ ``shrinks away from'' the acceptable region $\mathds{A}_i$ as sample sizes increase, making the two sets less likely to intersect.

\subsection{Fragmented OSSA (FOSSA) }
\label{subsection:2.3}
As we have seen, the OSSA framework involves checking the intersection of the acceptable region $\mathds{A}_{i}$ and the confidence region $\mathds{C}_{i}$ for each system $i=1,2,\dots,k$.
Detecting the intersection of two sets is a fundamental problem in computational geometry \citep{toth2017handbook}.
If both $\mathds{A}_{i}$ and $\mathds{C}_{i}$ are polyhedra, detecting their intersection is of $O(\log|\mathds{A}_{i}|\log|\mathds{C}_{i}|)$ complexity, where $|\cdot|$ is the total number of faces of a polyhedron \citep{dobkin1190determining,barba2014optimal}.
More generally, when both $\mathds{A}_{i}$ and $\mathds{C}_{i}$ are convex, their intersection detection can be treated as a convex optimization problem and solved numerically \citep{wang2020distance}.  
However, to the best of our knowledge, there is no general solution to detect the intersection of two sets when one of them is non-convex.
In some of the problem settings we consider, $\mathds{A}_{i}$ is unavoidably non-convex, whereas $\mathds{C}_{i}$ can be designed to have a geometry that simplifies the intersection-detection task.
Nevertheless, in these situations, constructing the returned subset can be made computationally feasible by exploiting the particular geometry of $\mathds{A}_{i}$; examples will be given in Section~\ref{sec:two examples}.

We introduce a subclass of the \ourfrmwrk\ framework in which checking the non-emptiness of $\mathds{A}_{i}\cap\mathds{C}_{i}$ is in many cases straightforward.
To unify the upcoming results for different problem settings, we adopt a definition of acceptability based on the \emph{feasibility} of a system's response vector and its \emph{preferability} when compared to those of other systems.
In other words, acceptable systems are those that have feasible response vectors and are regarded as no worse than any other feasible system, as laid out in Assumption~\ref{assump: Acceptability}.

\begin{assumption}
For a given System $i$, $i = 1, 2, \ldots, k$, the acceptable region is defined as
\label{assump: Acceptability}
\begin{align*}
\mathds{A}_{i}=\{\mathbf{M}\in \mathds{R}^{k\times d}\colon &\boldsymbol{m}_{i}\in\mathds{F} \text{ and } \\
&\boldsymbol{m}_{i}\nsucc \boldsymbol{m}_{j} \text{ for all } j\neq i \},
\end{align*}
where 
\begin{itemize}
    \item $\mathds{F} \subseteq \mathds{R}^{d}$ is a closed set;
    \item $\boldsymbol{m}_{i}\prec \boldsymbol{m}_{j}$ and $\boldsymbol{m}_{i}\nprec \boldsymbol{m}_{j}$ signify that the decision maker does or does not prefer $\boldsymbol{m}_{i}$ to $\boldsymbol{m}_{j}$, respectively;
    \item the preference relationship is transitive, i.e., if $\boldsymbol{m}_{i}\prec \boldsymbol{m}_{j}$ and $\boldsymbol{m}_{j}\prec \boldsymbol{m}_{l}$, then $\boldsymbol{m}_{i}\prec \boldsymbol{m}_{l}$;
    \item any feasible system is preferred to any infeasible system, i.e., $\boldsymbol{m}_{i}\prec \boldsymbol{m}_{j}$ for all $i \in \mathcal{F}(\mathbf{M})$ and $j \in \mathcal{F}(\mathbf{M})^c$ where $\mathcal{F}(\mathbf{M})=\{i\colon \boldsymbol{m}_{i}\in\mathds{F}\}$ denotes the set of feasible systems; and
    \item all infeasible systems are equally preferable, i.e., $\boldsymbol{m}_{i}\nprec \boldsymbol{m}_{j}$ and $\boldsymbol{m}_{i}\nsucc \boldsymbol{m}_{j}$ for all $i,j\in\mathcal{F}(\mathbf{M})^{c}$.
\end{itemize}
\end{assumption}

Assumption~\ref{assump: Acceptability} encompasses virtually all forms of acceptability appearing in the R\&S literature, including single-response optimality, constrained optimality and Pareto optimality.
For example, the problem of minimizing a single response corresponds to the case where $\mathds{F}=\mathds{R}$ and systems with smaller responses are preferred.
Meanwhile, the problem of feasibility determination with respect to a single response corresponds to the case where $\mathds{F}=\{m \in \mathds{R} \colon m\leq \mu^{\dagger}\}$ for some user-specified $\mu^{\dagger}\in\mathds{R}$ and all feasible systems are equally preferable.

We proceed with a few other technical assumptions involving the definition of acceptability and its interaction with the confidence regions.
Let $\boldsymbol{m}_{i} \preccurlyeq \boldsymbol{m}_{j}$ denote that either $\boldsymbol{m}_{i}\prec \boldsymbol{m}_{j}$ or $\boldsymbol{m}_{i}= \boldsymbol{m}_{j}$, where $\boldsymbol{m}_{i}= \boldsymbol{m}_{j}$ means that the two vectors are identical.
For any vector $\boldsymbol{m}\in\mathds{R}^{d}$, we denote the set of response vectors that are better (or worse) than or equal to $\boldsymbol{m}$ by $\mathds{B}(\boldsymbol{m})=\{\boldsymbol{m}'\in\mathds{R}^{d}\colon \boldsymbol{m}'\preccurlyeq \boldsymbol{m} \}$ and $\mathds{W}(\boldsymbol{m})=\{\boldsymbol{m}'\in\mathds{R}^{d}\colon \boldsymbol{m}'\succcurlyeq \boldsymbol{m} \}$, respectively.

\begin{assumption}
\label{assump: existence}
For the definition of acceptability given in Assumption~\ref{assump: Acceptability}, on any closed and bounded set $\mathbb{S}\subseteq\mathds{R}^{d}$,
\begin{equation*}
\mathds{B}^{*}(\mathbb{S})\equiv\{\boldsymbol{m}\in\mathbb{S}\colon \mathds{B}(\boldsymbol{m})\cap\mathbb{S}=\{\boldsymbol{m}\}\}\neq\emptyset,    
\end{equation*}
and
\begin{equation*}
\mathds{W}^{*}(\mathbb{S})\equiv\{\boldsymbol{m}\in\mathbb{S}\colon \mathds{W}(\boldsymbol{m})\cap\mathbb{S}=\{\boldsymbol{m}\}\}\neq\emptyset.
\end{equation*}
\end{assumption}

We adopt Assumption \ref{assump: existence} to exclude cases where there is no generic response vector that is either the most or least preferable over a closed and bounded set.
We have yet to encounter a practical situation where Assumption \ref{assump: existence} is violated.

The defining characteristic of the upcoming subclass of procedures is the construction of the confidence region for the configuration as a Cartesian product of $k$ confidence regions, one for each system's response vector. Mathematically, we mean that $\mathds{C}_{i}=\bigtimes\limits_{j=1}^{k}\mathds{C}_{ij}$ where  $\mathds{C}_{ij}\subseteq\mathds{R}^{d}$ is a confidence region for $\boldsymbol{\theta}_{j}$, $j=1,2,\dots,k$, and $\mathds{C}_{ii}$ is the confidence region for the response vector of the system being screened. (Likewise, for the set-wise PAS guarantee, we will have $\mathds{C}_u = \bigtimes\limits_{j=1}^k\mathds{C}_{uj}$.)

\begin{assumption}
\label{assump: closed & bounded}
For any $i=1,2,\dots,k$ and any $\boldsymbol{m}\in\mathds{C}_{ii}$, $\mathds{B}(\boldsymbol{m})\cap\mathds{C}_{ii}$ is closed and bounded; 
and for all $j\neq i$ and any $\boldsymbol{m}\in\mathds{C}_{ij}$, $\mathds{W}(\boldsymbol{m})\cap\mathds{C}_{ij}$ is closed and bounded.
\end{assumption}
Assumption \ref{assump: closed & bounded} ensures that the confidence region $\mathds{C}_{i}$ does not include any configurations in which the system being screened can be made feasible and infinitely preferable while all other systems can be made infeasible or infinitely unpreferable.
For example, when minimizing a single response, a system is considered infinitely preferable if its response value is negative infinity and infinitely unpreferable if its response value is positive infinity.
This property is desirable because if $\mathds{C}_{i}$ otherwise included such extreme configurations, System $i$ would always be returned.

With Assumptions~\ref{assump: Acceptability}--\ref{assump: closed & bounded} on hand, we introduce a subclass of the \ourfrmwrk\ framework called fragmented OSSA (\oursimplerfrmwk).
The term \textit{fragmented} refers to how the uncertainty about the configuration $\truemean$ is controlled by separately controlling the uncertainty about each system's response vector $\boldsymbol{\theta}_{j}$---in other words, the choice to construct $\mathds{C}_i$ (or $\mathds{C}_u$) as a Cartesian product. As we will see shortly, Assumptions~\ref{assump: Acceptability}--\ref{assump: closed & bounded} will enable FOSSA procedures to check the non-emptiness of $\mathds{A}_{i}\cap\mathds{C}_{i}$ more directly. 
In particular, when screening a System $i$, \oursimplerfrmwk\ procedures work with a certain subset $\mathbb{M}^{*}_{i}\subseteq \mathds{C}_{i}$, where
\begin{align*}
\mathbb{M}^{*}_{i}=\{\mathbf{M}\in&\mathds{C}_{i}\colon \boldsymbol{m}_{i} \in \mathds{F} 
\text{ and }  \nexists \mathbf{M}'\in
\mathds{C}_{i}\setminus\{\mathbf{M}\} 
 \text{ s.t. } \\
 &\boldsymbol{m}'_{i} \preccurlyeq  \boldsymbol{m}_{i} \text{ and }  \boldsymbol{m}'_{j} \succcurlyeq \boldsymbol{m}_{j} \text{ for all }j\neq i\}. 
\end{align*}
In words, $\mathbb{M}^{*}_{i}$ is the set of configurations within $\mathds{C}_{i}$ that are most favorable for System $i$, namely, those for which there is no other configuration in $\mathds{C}_{i}$ for which System $i$ is as or more acceptable and the other systems are as or less acceptable.

In our setup, if System $i$ is plausibly feasible, $\mathbb{M}^{*}_{i}$ can be expressed as a Cartesian product.
\begin{proposition} 
\label{Proposition: Best and Worst }
If Assumptions~\ref{assump: Acceptability}--\ref{assump: closed & bounded} hold and $\mathds{B}^{*}(\mathds{C}_{ii})\cap\mathds{F}\neq\emptyset$, 
\begin{align*}
\mathbb{M}^{*}_{i} &= 
\bigtimes\limits_{j=1}^{k}\mathbb{M}^{*}_{ij} \\
&\equiv \bigtimes\limits_{j=1}^{i-1}\mathds{W}^{*}(\mathds{C}_{ij})\times \mathds{B}^{*}(\mathds{C}_{ii})\times \bigtimes\limits_{j=i+1}^{k}\mathds{W}^{*}(\mathds{C}_{ij}),
\end{align*}
where $\mathds{B}^{*}(\cdot)$ and $\mathds{W}^{*}(\cdot)$ are as defined in Assumption~\ref{assump: existence}.
\end{proposition}

The reason for considering the special subset $\mathbb{M}^{*}_i$ is that when screening System $i$, it suffices for FOSSA procedures to check the non-emptiness of $\mathds{A}_{i} \cap \mathbb{M}^{*}_{i}$. 
\begin{theorem} 
\label{theorem: FOSSA}
If Assumptions~\ref{assump: Acceptability}--\ref{assump: closed & bounded} hold, $\mathds{A}_{i} \cap \mathds{C}_{i}\neq\emptyset$ if and only if $\mathds{A}_{i} \cap \mathbb{M}^{*}_{i}\neq\emptyset$.
\end{theorem}
As a result of Theorem \ref{theorem: FOSSA}, \oursimplerfrmwk\ procedures return the subset $\mathcal{S}=\{i\colon\mathds{A}_{i}\cap\mathbb{M}^{*}_{i}\neq\emptyset\}$.
In many cases, checking the non-emptiness of $\mathds{A}_{i}\cap\mathbb{M}^{*}_{i}$ is easier than checking that of $\mathds{A}_{i}\cap\mathds{C}_{i}$, such as when $\mathbb{M}_{i}^{*}$ is a singleton or the configurations in $\mathbb{M}_{i}^{*}$ share a consistent preference regarding System $i$.
In these cases, the intersection-detection problem reduces to verifying whether a single configuration\ is in $\mathds{A}_{i}$.
In Figure~\ref{fig:Fig1} of Section~\ref{sec:single_resp}, we illustrate $\mathds{A}_i$ and $\mathbb{M}_{i}^{*}$ in the classical setting where the system with the smallest response is acceptable.
In Sections \ref{sec:constrained_opt} and \ref{sec:pareto}, we present \oursimplerfrmwk\ procedures designed for two commonly used definitions of acceptability.

A further consequence of Theorem~\ref{theorem: FOSSA} is that FOSSA procedures will be consistent if $\mathbb{M}^{*}_{i}$ converges to some point or set outside of $\mathds{A}_{i}$ for those systems that can be screened out asymptotically.
In cases where the confidence region $\mathds{C}_{i}$ is built around a consistent estimator $\widehat{\mathbf{M}}_0$, this will generally occur since $\truemean \notin \mathrm{cl}(\mathds{A}_{i})$ for those systems.
It is our belief that the vast majority of FOSSA procedures users are likely to work with will be consistent.
This conjecture is partially supported by Assumption~\ref{assump: closed & bounded}, which ensures that the confidence regions used for screening unacceptable systems should be bounded in a direction that would otherwise allow that system to be made infinitely preferable, and hence always returned.
Moreover, the FOSSA procedures tested in the numerical experiments can all be shown to be consistent.

We conclude this subsection with a brief discussion of how FOSSA procedures can be parallelized within a master-worker computing framework. 
Consider a divide-and-conquer scheme in which on each worker processor, FOSSA is applied locally on an assigned subset of systems and the surviving systems from each worker (or, more precisely, their data) are then returned to the master processor for a final application of the same FOSSA procedure.
The resulting final subset of systems is the same as that returned for a single-processor implementation of the FOSSA procedure.
This result is summarized in Theorem~\ref{theorem: Divide and Conquer} below; the proof appears in the electronic companion.
\begin{theorem}
\label{theorem: Divide and Conquer}
If Assumptions~\ref{assump: Acceptability}--\ref{assump: closed & bounded} hold and the confidence regions $\mathds{C}_{ij}$ for all $i,j=1,2,\dots,k$ are held fixed, FOSSA procedures can be parallelized using a divide-and-conquer scheme without compromising their screening power.
\end{theorem}

In cases where checking the non-emptiness of $\mathds{A}_{i}\cap\mathds{C}_{i}$ is time-intensive and its complexity grows with $k$, a divide-and-conquer approach to screening can offer significant time savings relative to a centralized approach.
Moreover, this approach can also reduce communication and memory overhead in cases where the simulation runs are distributed among workers, because workers can immediately perform screening without waiting for other workers to share their results, and data from screened out systems can be discarded rather than reported back to the master.

\subsection{Confidence Regions}
\label{subsetion:section 2.4}
Before giving examples of FOSSA procedures, we briefly preview the confidence regions that the various procedures will employ.

When systems are simulated independently, constructing Cartesian product confidence regions $ \mathds{C}_{i}$ with an overall confidence level of $1-\alpha$ is straightforward.
Specifically, if $\mathds{C}_{i1}, \mathds{C}_{i2},\dots, \mathds{C}_{ik}$ each have a confidence level of $(1-\alpha)^{1/k}$, then
\begin{align*} 
\mathrm{P}\left(\truemean \in \mathds{C}_{i}\right)
&=\mathrm{P}\left(\truemean \in \bigtimes\limits_{j=1}^{k}\mathds{C}_{ij}\right) \\
&=\prod\limits_{j=1}^{k}\mathrm{P}\left(\boldsymbol{\theta}_{j}\in\mathds{C}_{ij}\right) \\
&\geq\left(1-\alpha\right)^{(1/k)k} \\
&=1-\alpha,
\end{align*} 
where the independence assumption is applied in the second equality.
Even when systems are not simulated independently, the Bonferroni correction can be used to additively split the error $\alpha$ across each system, though the resulting \oursimplerfrmwk\ procedures would be slightly more conservative.

\begin{remark}
For certain R\&S problems and procedures, inducing dependence in the outputs across systems by using common random numbers (CRN) can prove advantageous. 
However, using CRN usually requires that all systems receive a common number of replications.
Preliminary numerical experiments suggest that applying the FOSSA procedures outlined in this paper on outputs generated via CRN will not undermine the promised PAS guarantee.
Furthermore, applying CRN tends to improve the performance of the worst returned system, but also increases the variance of the size of the returned subset.
We leave the development of OSSA procedures that exploit CRN with confidence regions designed for dependent sampling for future research.
\end{remark}

\begin{figure}[tb]
    \centering
    \includegraphics[height=4.6cm]{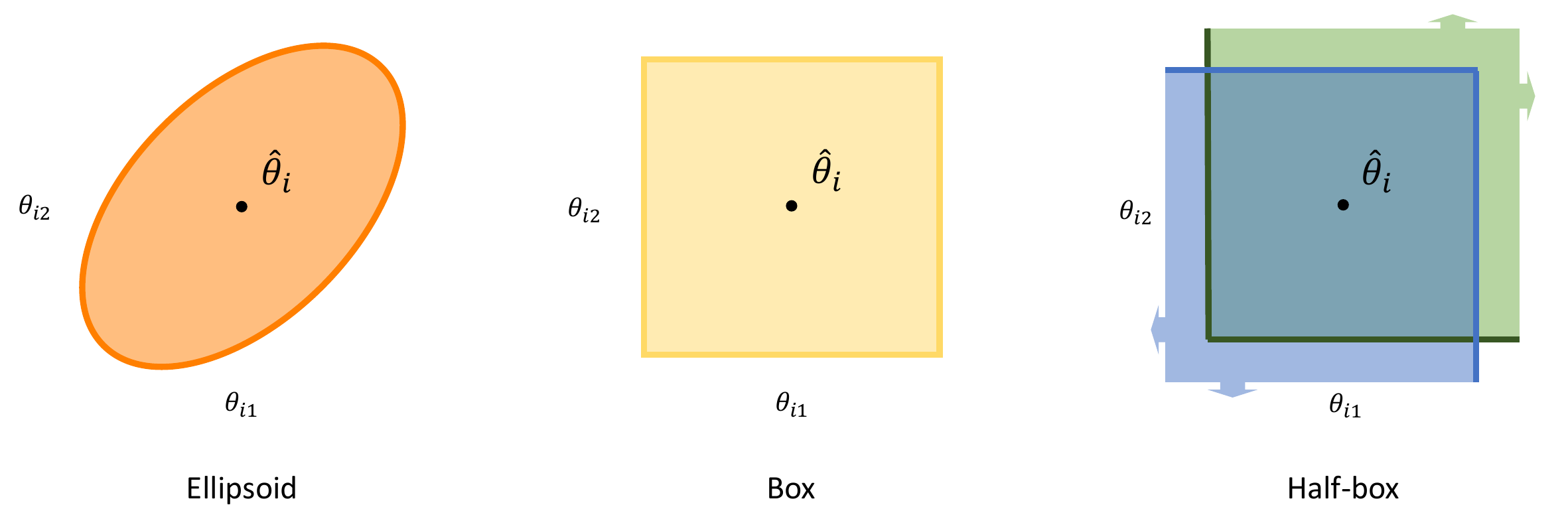}
    
    \caption{Three types of confidence regions for each system's response vector when $d=2$. The Half-Box type consists of two distinct confidence regions: one for the system being screened (blue) and the other for the other systems (green).}
    \label{fig: confidence regions}
\end{figure}

We now discuss three common geometries of $\mathds{C}_{ij}$ for $i, j=1,2,\dots,k$, namely, an ellipsoid, a box, and a half-box in $\mathds{R}^d$.
These regions are shown in Figure~\ref{fig: confidence regions} for the $d=2$ setting and described in more detail in the EC.
An ellipsoid-shaped confidence region can be constructed when, for example, $\boldsymbol{\theta}_j$ is the expected value of outputs following an elliptical distribution (e.g., multivariate normal distribution and  multivariate $t$-distribution) \citep{frahm2004generalized}, or when the maximum likelihood estimator for $\boldsymbol{\theta}_j$ is asymptotically normally distributed and unbiased. Alternatively, a box-shaped (hyperrectangle) confidence region can be constructed as the Cartesian product of $d$ two-sided confidence intervals corresponding to each component of $\boldsymbol{\theta}_j$. Let $l_{ijr}^{\mathrm{box}}$ and $u_{ijr}^{\mathrm{box}}$ denote the associated lower and upper confidence bounds, respectively, for $\theta_{jr}$, $r = 1, 2, \ldots, d$.
For notational convenience, we assume that the same confidence intervals are used regardless of which system is being screened, i.e., $l_{jr}^{\mathrm{box}}=l_{1jr}^{\mathrm{box}}=l_{2jr}^{\mathrm{box}}=\dots=l_{kjr}^{\mathrm{box}}$ and $u_{jr}^{\mathrm{box}}=u_{1jr}^{\mathrm{box}}=u_{2jr}^{\mathrm{box}}=\dots=u_{kjr}^{\mathrm{box}}$ for all $j=1,2,\dots,k$ and $r=1,2,\dots,d$.
The overall confidence level of the box-shaped confidence region can be controlled using probability inequalities, such as the Bonferroni or \v{S}id\'{a}k correction.

By construction, the ellipsoid-shaped and box-shaped confidence regions are not oriented differently depending on which system is being screened.
Consequently, their geometries are better suited to delivering the stronger set-wise PAS guarantee.
To deliver the system-wise PAS guarantee, we turn to half-box confidence regions, which, in contrast to box-shaped regions, are the Cartesian product of $d$ \textit{one-sided} confidence intervals for each component of $\boldsymbol{\theta}_j$. 
The screening power of FOSSA procedures can be enhanced by orienting the half-boxes depending on which system is being screened.
More specifically, the confidence region for $\boldsymbol{\theta}_{i}$ can be oriented to exclude overly unfavorable responses, whereas the confidence regions for $\boldsymbol{\theta}_{j}$, $j \neq i$ can be oriented to exclude overly favorable responses.
For example, when smaller responses are preferred,
we orient the half-boxes such that the confidence intervals are bounded from below for System $i$ and bounded from above for all the other Systems $j \neq i$.
Formally, the half-box confidence region for $\boldsymbol{\theta}_i$ is defined as
$$\mathds{C}_{ii}^{\mathrm{hb}}=\{\boldsymbol{m}\in\mathds{R}^{d}\colon m_r\geq l_{iir}^{\mathrm{hb}} \text{ for all }r=1,2,\dots,d\},$$
and for each competing System $j\neq i$, the confidence region for their response is
$$\mathds{C}_{ij}^{\mathrm{hb}}=\{\boldsymbol{m}\in\mathds{R}^{d}\colon m_r\leq u_{ijr}^{\mathrm{hb}} \text{ for all }r=1,2,\dots,d\},$$
where $l_{iir}^{\mathrm{hb}}$ and $u_{ijr}^{\mathrm{hb}}$ denote the one-sided confidence bounds used to form the half-boxes.

Our numerical experiments indicate that using a simple Bonferroni correction across responses---as would usually be done for the box and half-box constructions---achieves comparable screening power to using ellipsoid confidence regions.
The reason is that, although the box and half-box confidence regions control the error more conservatively, their geometries are better aligned with the acceptable regions for definitions of acceptability that involve optimizing individual responses, e.g,. optimization with multiple objectives or stochastic constraints.
Therefore, depending on the desired PAS guarantee, we recommend using the box or half-box confidence regions over the ellipsoid confidence region, as they lead to a simpler intersection-detection problem.

\begin{remark}
It may be appealing to leverage more sophisticated approaches for controlling the family-wise error rate when screening for acceptability. For instance, the set-wise PAS guarantee can be viewed as a test of multiple hypotheses, namely, whether each System $i = 1, 2, \ldots, k$ is acceptable or not. However, testing each hypothesis $\{i \in \mathcal{A}(\truemean)\}$ generally necessitates controlling the allowable error across systems, especially for definitions of acceptability that feature preference relationships, i.e., those captured by the conditions $\boldsymbol{m}_{i}\nsucc \boldsymbol{m}_{j} \text{ for all } j\neq i$ in Assumption~\ref{assump: Acceptability}.
As a result, applying standard approaches from multiple hypothesis testing, such as Holm's step-down method \citep{holm1979simple}, would lead to a redundant splitting of the error across systems: once across hypotheses and again within each hypothesis.
In contrast, OSSA and FOSSA procedures deliver the set-wise PAS guarantee by controlling the error across systems once, through the use of a common confidence region $\mathds{C}_{u}$, and do not appear to be easily improved upon in this respect.
\end{remark}

\section{Examples of FOSSA Procedures}
\label{sec:two examples}
In this section, we describe  FOSSA procedures that employ box and half-box confidence regions under two definitions of acceptability that arise when decision makers are concerned with multiple aspects of system performance: optimization with stochastic constraints and Pareto optimality.
The Cartesian product structure of the box and half-box confidence regions allows the intersection-detection problem to be solved analytically, resulting in procedures whose time complexities match those of identifying acceptable systems when the true configuration is known.
Theoretical results related to intersection detection for confidence regions with general geometry, as well as an illustrative procedure for the ellipsoid confidence region, are provided in the EC.

\subsection{Optimization with Stochastic Constraints}
\label{sec:constrained_opt}

This subsection focuses on the case where acceptability is defined as optimality subject to one or more stochastic constraints.
Without loss of generality, we designate the first response as the primary response and the remaining responses as secondary responses.
System $i$ is considered \emph{feasible} if its response vector $\boldsymbol{\theta}_{i}$ is in some polyhedron $\mathds{F}=\{\boldsymbol{m} \in \mathds{R}^{d} \colon m_{r} \leq \theta_{r}^{\dagger} \ \mathrm{for} \ r=2,3,\dots,d\}$, where the vector $\boldsymbol{\theta}^{\dagger} \equiv (\theta_{2}^{\dagger},\theta_{3}^{\dagger},\dots,\theta_{d}^{\dagger})\in \mathds{R}^{d-1}$ is specified by the decision maker.

Among all feasible systems, those that have the lowest primary response are called \emph{constrained optimal} systems.  
The acceptable set is
$\mathcal{A}(\truemean)=\{i \colon 
i\in \mathcal{F}(\mathbf{M}_0)
\mathrm{\ and\ } 
\theta_{i1}\leq\theta_{j1} \mathrm{\ for \ all \ } j\in\mathcal{F}(\mathbf{M}_0) \} $ where $\mathcal{F}(\mathbf{M}_0) = \{i \colon \boldsymbol{\theta}_i \in \mathds{F}\}$, and the corresponding acceptable regions are
$\mathds{A}_{i}=\{\mathbf{M}\in\mathds{R}^{k\times d} \colon 
\boldsymbol{m}_{i}\in \mathds{F}
\mathrm{\ and\ } 
m_{i1}\leq m_{j1} \mathrm{\ for \ all \ } j  
\mathrm{\ s.t.\ } 
\boldsymbol{m}_{j} \in \mathds{F} \}
\mathrm{\ for \ } i=1,2,\dots,k
$.
For any $d\geq2$ and $k\geq2$, $\mathds{A}_{i}$ is non-convex.
Nevertheless, \oursimplerfrmwk\ procedures designed for this problem setting will be able to check the non-emptiness of $\mathds{A}_{i}\cap\bigtimes\limits_{j=1}^{k}\mathbb{M}^{*}_{ij}$. 

A key insight in this setting is that each system can be succinctly described by its primary response and its feasibility.
Consequently, when screening System $i$, it will be possible to further simplify checking the non-emptiness of $\mathds{A}_{i}\cap\bigtimes\limits_{j=1}^{k}\mathbb{M}^{*}_{ij}$ to checking whether an extended real-value vector, $\tau_i  = (\tau_{i1}, \tau_{i2}, \dots, \tau_{ik})$, describes a configuration in which System $i$ is a constrained optimal system.
Define the scalars
\vspace{-1em}
\begin{align}
\label{optimization problem 1_Constrained Opt}
\tau_{ii} = \begin{cases}
\min m_{1}\ \mathrm{s.t.}\ \boldsymbol{m}\in \mathds{C}_{ii}\cap\mathds{F}&\mathrm{if}\ \mathds{C}_{ii}\cap\mathds{F}\neq\emptyset,\\
\infty &\mathrm{if}\ \mathds{C}_{ii}\cap\mathds{F}=\emptyset;	
\end{cases}
\end{align}
and
\vspace{-1em}
\begin{align}
\label{optimization problem 2_Constrained Opt}
\tau_{ij} = \begin{cases}
\max m_{1}\ \mathrm{s.t.}\ \boldsymbol{m}\in\mathds{C}_{ij} &\mathrm{if}\ \mathds{C}_{ij}\cap\mathds{F}^{c}=\emptyset,\\
\infty &\mathrm{if}\ \mathds{C}_{ij}\cap\mathds{F}^{c}\neq\emptyset.
\end{cases}
\end{align}
for all $j \neq i$.
In words, $\tau_{ii}$ represents the best primary response among feasible response vectors in System $i$'s confidence region. On the other hand, $\tau_{ij}$ represents the worst primary response among feasible response vectors in System $j$'s confidence region. If there are no feasible vectors in $\mathds{C}_{ii}$, we set $\tau_{ii} = \infty$, or if there are any infeasible vectors in $\mathds{C}_{ij}$, we set $\tau_{ij} = \infty$.
These values can be thought of as the primary response of elements in $\mathds{B}^{*}(\mathds{C}_{ii})$ and $\mathds{W}^{*}(\mathds{C}_{ij})$, respectively, with the value of infinity indicating that the sets contains only infeasible vectors.
This interpretation is depicted in Figure \ref{fig:Fig2} for ellipsoid confidence regions, where $\tau_{ii}$ and $\tau_{ij}$ are shown to be projections of $\mathds{B}^{*}(\mathds{C}_{ii})$ and $\mathds{W}^{*}(\mathds{C}_{ij})$ on the dimension associated with the primary response. 

Introducing $\tau_{ii}$ and $\tau_{ij}$ facilitates construction of the subset returned by FOSSA procedures, as made clear in Proposition \ref{Proposition for Constrained Optimality}. Specifically, the non-emptiness of $\mathds{A}_{i}\cap\bigtimes\limits_{j=1}^{k}\mathbb{M}^{*}_{ij}$ can be checked by simply comparing the values in the vector $\tau_i$.

\begin{figure}
    \centering
    \includegraphics[width=0.50\textwidth]{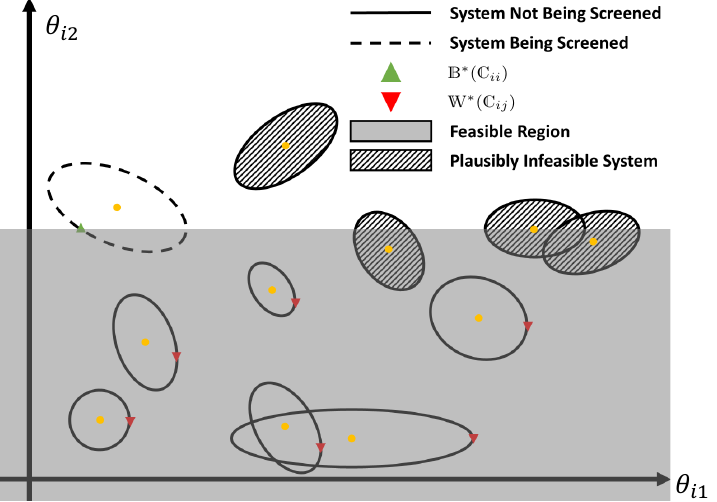}
     \caption{Elliptical confidence regions for each system's response vector in a 11-system problem with one stochastic constraint. The feasible region is shaded gray and the confidence region of the system being screened (System $i$) is indicated with a dashed boundary. For every other system, a hatched confidence region indicates that the corresponding system is plausibly infeasible, i.e., the corresponding $\tau_{ij}$ is $\infty$. For Systems $j \neq i$ that are not plausibly infeasible, a red upside-down triangle indicates $\mathds{W}^{*}(\mathds{C}_{ij})$ and its projection on the $\theta_{i1}$ axis is $\tau_{ij}$. System $i$ is plausibly feasible and $\mathds{B}^{*}(\mathds{C}_{ii})$ is represented by a green triangle. Because the projection of $\mathds{B}^{*}(\mathds{C}_{ii})$ on $\theta_{i1}$ ($\tau_{ii}$) is less than those of all other plausibly feasible systems, System $i$ will be returned.}
    \label{fig:Fig2}
\end{figure}

\begin{proposition}
\label{Proposition for Constrained Optimality}
For the definition of acceptability as constrained optimality,
\begin{align*}
\biggr\{i\colon&\mathds{A}_{i}\cap\bigtimes\limits_{j=1}^{k}\mathbb{M}^{*}_{ij}\neq\emptyset\biggr\} \\
&=\{i\colon\tau_{ii}\neq\infty \text{ and } \tau_{ii}\leq \tau_{ij} \ \mathrm{for\ all\ } j\neq i\}.    
\end{align*}
\end{proposition}

As can be seen from (\ref{optimization problem 1_Constrained Opt}), calculating $\tau_{ii}$ entails solving a feasibility problem, namely, checking whether $\mathds{C}_{ii} \cap \mathds{F} = \emptyset$, and, depending on the result, solving another optimization problem: $\min m_1$ s.t.\ $\boldsymbol{m} \in \mathds{C}_{ii} \cap \mathds{F}$. 
A similar observation can be made about calculating $\tau_{ij}$ as defined in (\ref{optimization problem 2_Constrained Opt}). 
For the choices of confidence regions presented in this paper, the relevant feasibility problems can be solved with minimal effort.
For the box and half-box cases in particular, the optimization problems can be solved analytically.

We will simultaneously address FOSSA procedures that use the box and half-box confidence regions due to their similar geometry. We call these procedures FOSSA Box and FOSSA Half-Box, respectively. Letting $l_{ir}=l_{ir}^{\mathrm{box}}$ and $u_{ir}=u_{ir}^{\mathrm{box}}$ for all $i=1,2,\dots,k$ will deliver the set-wise PAS guarantee, whereas choosing $l_{ir}=l_{ir}^{\mathrm{hb}}$ and $u_{ir}=u_{ir}^{\mathrm{hb}}$ for all $i=1,2,\dots,k$ will deliver the system-wise PAS guarantee. 
In both cases, the confidence region for each system's response vector is an intersection of half-spaces whose defining hyperplanes are parallel to those describing the feasible region $\mathds{F}$.
As a result, checking the non-emptiness of $\mathds{C}_{ii} \cap \mathds{F}$ in (\ref{optimization problem 1_Constrained Opt}) reduces to checking whether each lower confidence bound on $\theta_{ir}$ for $r = 2, 3, \dots, d$ is less than the corresponding threshold $\theta_r^{\dagger}$. Likewise, checking the non-emptiness of $\mathds{C}_{ij} \cap \mathds{F}^c$ in (\ref{optimization problem 2_Constrained Opt}) involves the upper confidence bounds on $\theta_{jr}$ for $r = 2, 3, \dots, d$.
As for the other optimization problems in (\ref{optimization problem 1_Constrained Opt}) and (\ref{optimization problem 2_Constrained Opt}), their optimal values are $l_{i1}$ and $u_{i1}$, respectively.

Algorithm \ref{Constrained-Optimal Procedure(system-wise)} summarizes the two FOSSA procedures.
For both procedures, the returned subset can be constructed in $O(kd)$ operations.

\begin{algorithm}[tb]
\caption{FOSSA Box and FOSSA Half-Box for Constrained Optimality}
\label{Constrained-Optimal Procedure(system-wise)}
\begin{algorithmic}[1]
\Require $l_{ir},u_{ir}$ for $i=1,2,\dots,k$ and $r=1,2,\dots,d$
\State $\mathcal{S}\gets\{1,2,\dots,k\}$
\For{$i=1,2,\dots,k$}
    \For{$r=2,3,\dots,d$}
        \If{$l_{ir}>\theta^{\dagger}_{r}$}
            \State $l_{i1}\gets\infty$
            \State \textbf{break}
        \EndIf
    \EndFor
    \For{$r=2,3,\dots,d$}
        \If{$u_{ir}>\theta^{\dagger}_{r}$}
            \State $u_{i1}\gets\infty$
            \State \textbf{break}
        \EndIf
    \EndFor
\EndFor
\State $\tau\gets\min_{i\in\{1,2,\dots,k\}}u_{i1}$
\For{$i=1,2,\dots,k$}
    \If{$l_{i1}>\tau$}
        \State $\mathcal{S}\gets\mathcal{S}\setminus\{i\}$
    \EndIf
\EndFor
\State \Return $\mathcal{S}$
\end{algorithmic}
\end{algorithm}
   
\subsection{Pareto Optimality}
\label{sec:pareto}

This subsection focuses on the case where acceptability is defined as Pareto optimality, as can arise in situations where system performance is described by multiple, equally important responses.
System $i$ is said to dominate System $j$, denoted as $\boldsymbol{\theta}_{i} \prec_{p} \boldsymbol{\theta}_{j}$, if all of System $i$'s $d$ responses are less than or equal to those of System $j$, and System $i$ performs strictly better than System $j$ in at least one response.
We use $\boldsymbol{\theta}_{i} \nsucc_{p} \boldsymbol{\theta}_{j}$ to denote that System $i$ is not dominated by System $j$.
In this setting, acceptable systems are those that are not dominated by any other systems, hence the acceptable set is $\mathcal{A}(\truemean)=\{i \colon \boldsymbol{\theta}_{i} \nsucc_{p} \boldsymbol{\theta}_{j} \ \textrm{for all} \ j\neq i\}$, the so-called \emph{efficient set}.
Given the potential multiplicity of acceptable systems, the set-wise PAS guarantee may be favored in this setting.
A byproduct of FOSSA procedures that return a subset $\mathcal{S}$ delivering the set-wise PAS guarantee is a $1-\alpha$ confidence region for the so-called Pareto front. The specific form of this confidence region and the proof can be found in the electronic companion.
The acceptable regions are $\mathds{A}_{i}=\{\mathbf{M} \in \mathds{R}^{k\times d} \colon \boldsymbol{m}_{i} \nsucc_{p} \boldsymbol{m}_{j} \ \textrm{for all} \ j\neq i \}$ for $i = 1, 2, \ldots, k$, which for any $d\geq2$ and $k\geq2$ are non-closed non-convex sets. 
We will demonstrate how FOSSA procedures can still efficiently check the non-emptiness of $\mathds{A}_{i}\cap\mathds{C}_{i}$ or $\mathds{A}_{i}\cap\mathbb{M}^{*}_{i}$. 

FOSSA procedures designed for Pareto optimality can use the box confidence regions for the set-wise PAS guarantee or the half-box confidence regions for the system-wise PAS guarantee.
In either case, the non-emptiness of $\mathds{A}_i \cap \mathbb{M}^{*}_{i}$ can be determined straightforwardly because $\mathbb{M}^{*}_{i}$ is a singleton.
The sole configuration in $\mathbb{M}^{*}_{i}$ is that  defined by the lower confidence bounds of all of the responses for System $i$, denoted by $\boldsymbol{l}_i=(l_{i1},l_{i2},\dots,l_{id})$, and the upper confidence bounds of all of the responses of all Systems $j$, $j \neq i$, denoted by $\boldsymbol{u}_j=(u_{j1},u_{j2},\dots,u_{jd})$. Therefore, when screening System $i$, one needs to check only whether or not $\boldsymbol{l}_{i} \nsucc_{p} \boldsymbol{u}_{j} \ \textrm{for all} \ j\neq i$.
(For other confidence region geometries, checking the non-emptiness of $\mathds{A}_i \cap \mathds{C}_i$ is equivalent to checking whether $\mathds{C}_i$ intersects the union of an $O(k^{\lfloor d/2 \rfloor})$ number half-boxes specified by $\mathds{C}_j$ for all $j\neq i$. When $\mathds{C}_i$ is ellipsoidal, this intersection can be determined analytically.)

The FOSSA procedures with the box and half-box confidence regions are outlined in Algorithm~\ref{alg:pareto-box}.
Lines 3-5 in Algorithm \ref{alg:pareto-box} identify those systems for which their vector of lower confidence bounds are dominated by the vector of upper confidence bounds of some other system, in which case the system is not plausibly acceptable and is screened out.
This subtask can be done efficiently using Algorithm 4.1 in \cite{kung1975finding}, which has time complexity $O(k\log k)$ for $d=2$ and  $O(k(\log k)^{d-2})$ for $d\geq 3$. 
The overall time complexity of Algorithm \ref{alg:pareto-box} matches these, which is also the same time complexity as finding the Pareto front among $k$ systems with known responses.
 
\begin{algorithm}[tb]
\caption{FOSSA Box and FOSSA Half-Box for Pareto Optimality}
\label{alg:pareto-box}
\begin{algorithmic}[1]
\Require $\boldsymbol{l}_{i},\boldsymbol{u}_{i}$ for $i=1,2,\dots,k$
\State $\mathcal{S}\gets\{1,2,\dots,k\}$
\For{$i=1,2,\dots,k$}
    \For{$j=1,2,\dots,k$}
        \If{$\boldsymbol{l}_{i}\succ_{p}\boldsymbol{u}_{j}$}
            \State $\mathcal{S}\gets\mathcal{S}\setminus\{i\}$
            \State \textbf{break}
        \EndIf
    \EndFor
\EndFor
\State \Return $\mathcal{S}$
\end{algorithmic}
\end{algorithm}

\section{Numerical Experiments}
\label{sec:num_exp}

In this section, we evaluate and compare FOSSA with contemporary screening procedures on a modified version of the buffer allocation problem from \cite{patsis1997simd}, adapted to a bi-objective setting. This instance is well suited for multi-objective R\&S because it exhibits a well-structured Pareto front, and the simulation outputs from most systems are approximately normally distributed, consistent with standard assumptions in R\&S.
We consider stochastically constrained optimality and Pareto optimality as definitions of acceptability. Because there are no existing screening methods that offer rigorous confidence guarantees in either case, we use the most suitable existing approaches as baselines, making adaptations as needed.
We assess each procedure in terms of its screening power and computational efficiency. By comparing FOSSA procedures employing (half-)box and ellipsoid confidence regions, we find that box confidence regions are generally preferable  to ellipsoids, as they offer comparable screening power but with faster computation and lower computational complexity. Additionally, our results demonstrate that under the normality assumption, using Bonferroni or \v{S}id\'{a}k corrections to build box and half-box confidence regions that account for correlations between simulation outputs is effective for both stochastically constrained optimality and Pareto optimality, owing to the fact that the box and half-box confidence regions conform to the geometry of the acceptable regions better than ellipsoids.

\subsection{Problem Instance and Setup}

We consider a queueing network consisting of ten nodes arranged in three stages, as depicted in Figure~\ref{fig: Buffer Allocation Diagram}.
Each node functions as a single-server queue following a first-in-first-out service discipline.
The network serves two classes of customers, C1 and C2, having exponential and uniform interarrival times, respectively. Irrespective of their class, each customer enters the network at one of the four entry nodes (Nodes 0-3) uniformly at random and passes through all three service stages. 
The network's structure is such that a customer's path is predetermined upon entry. 
Service times for both classes are uniformly distributed at Nodes 0-7 and exponentially distributed at Nodes 8 and 9.
Each node has a finite number of buffer spaces in which customers can wait for service.
A node is considered \emph{full} when its service area and all of its buffer spaces are occupied by customers waiting, being served or blocked. In such cases, upstream customers attempting to transition to a full node are said to be \emph{blocked} and must remain at their current node until space becomes available. 
If an entry node is full, incoming customers at that node are denied entry into the network.
\begin{figure}[tb]
    \centering
    \includegraphics[width=0.5\textwidth]{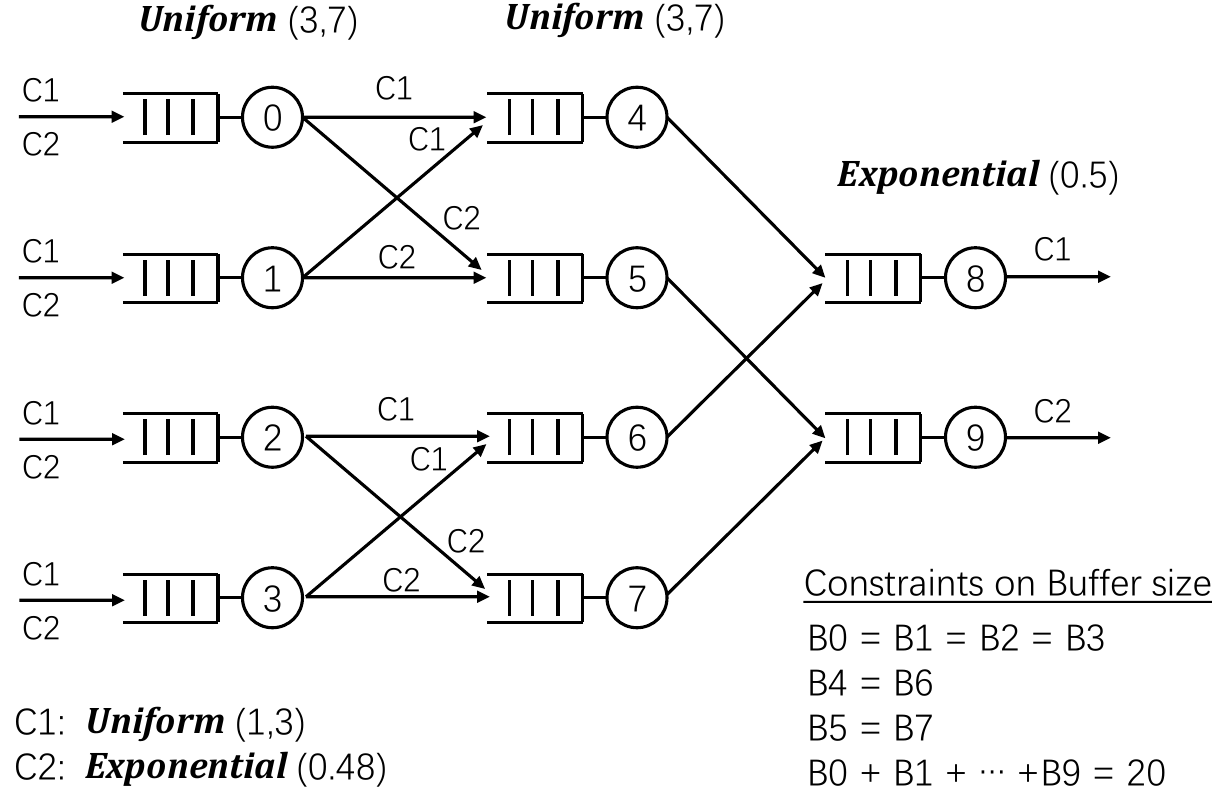}
    
    \caption{Diagram for the queueing system in the buffer allocation problem. The label B$x$ refers to the allocated buffer budget for node $x$, $x = 0, 1, \ldots, 9$. Interarrival times and service times are either uniformly distributed on an interval $[a, b]$ or exponentially distributed with a given mean.}
    \label{fig: Buffer Allocation Diagram}
\end{figure}

The decision maker's goal is to allocate 20 buffer spaces across the 10 nodes in a way that yields good system performance. 
We focus on two simulation outputs: the average waiting time of customers who complete service and the average idle rate of the servers.
A server is considered idle if their associated node is either empty or contains any blocked customers.
We impose three constraints on permissible buffer allocations based on the symmetry of the problem: Nodes 0--3 receive the same number of buffer spaces, as do Nodes 4 and 6, and Nodes 5 and 7.
These constraints reduce the number of systems under consideration from 10,015,005 to 1001.

The responses of a given system are estimated via simulation where a replication entails simulating for a period of 3000 time units, the first 2000 time units of which is treated as a warm-up period in which no statistics are collected.
The estimated responses of each system are shown in Figure~\ref{fig:Scatter Plot_All systems}, based on $50{,}000$ replications.
In this problem instance, many systems are closely grouped; in particular, more than half of the systems (506 out of 1001) make up the cluster in the lower-right corner of Figure \ref{fig:Scatter Plot_All systems}, characterized by high expected average idle rates and low expected average waiting times. 
Figure \ref{fig:Scatter Plot_One system} shows the distribution of $10{,}000$ samples from a representative system and suggests that the joint distribution of the simulation outputs vector is reasonably approximated by a bivariate normal distribution.
(For some systems, the marginal distribution of the average waiting time exhibits a slight right skew.)
We therefore use confidence regions derived under a normality assumption, as detailed in the electronic companion.
For a large majority of systems (722 out of 1001), the two components of simulation outputs are negatively correlated, indicating a common trade-off between customer waiting times and server utilization.

\begin{figure}[]
    \centering
    \subfigure[responses]
    {\includegraphics[width=0.45\textwidth]{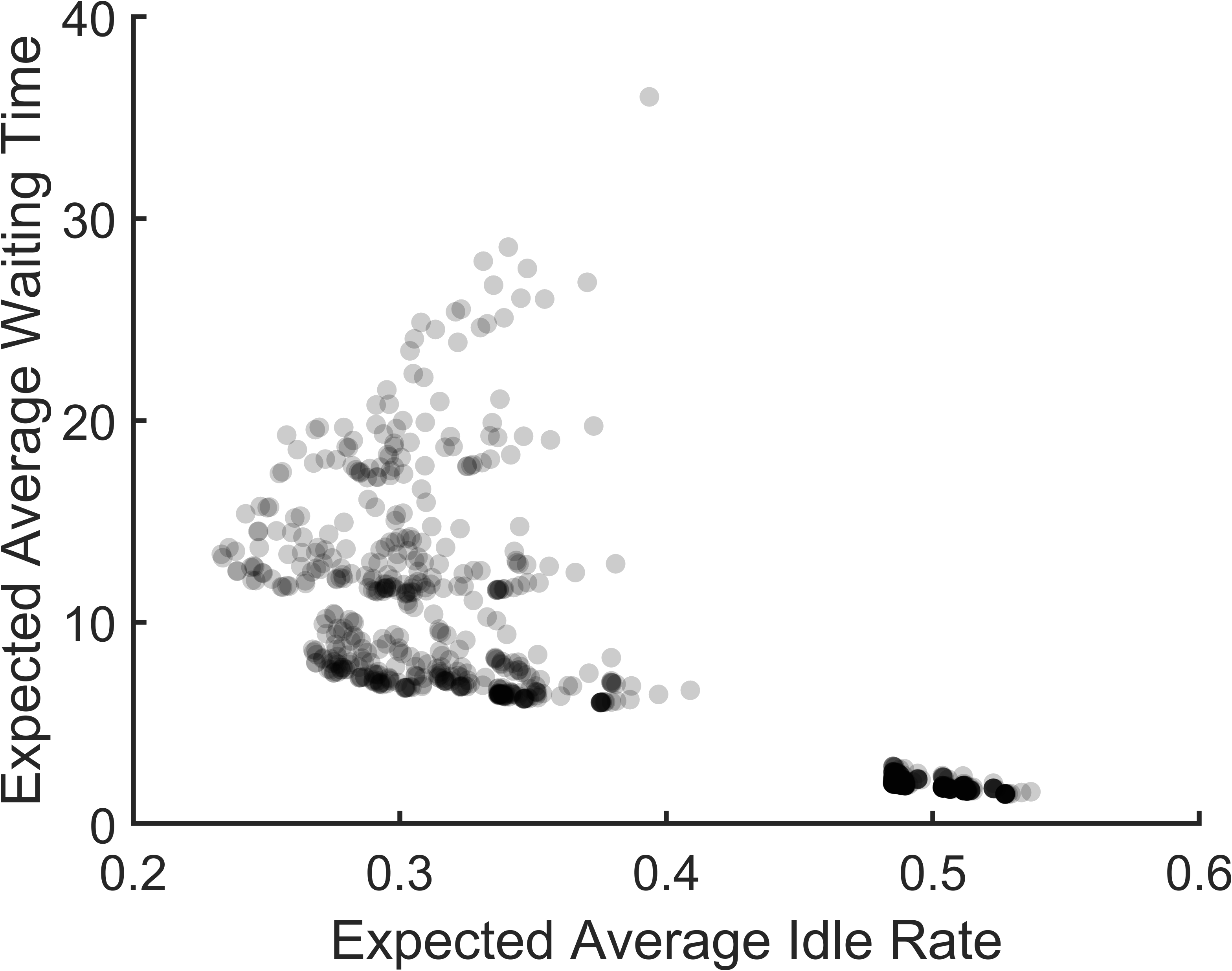}
    \label{fig:Scatter Plot_All systems}}
    \hfill
    \subfigure[Responses from One System]
    {\includegraphics[width=0.45\textwidth]{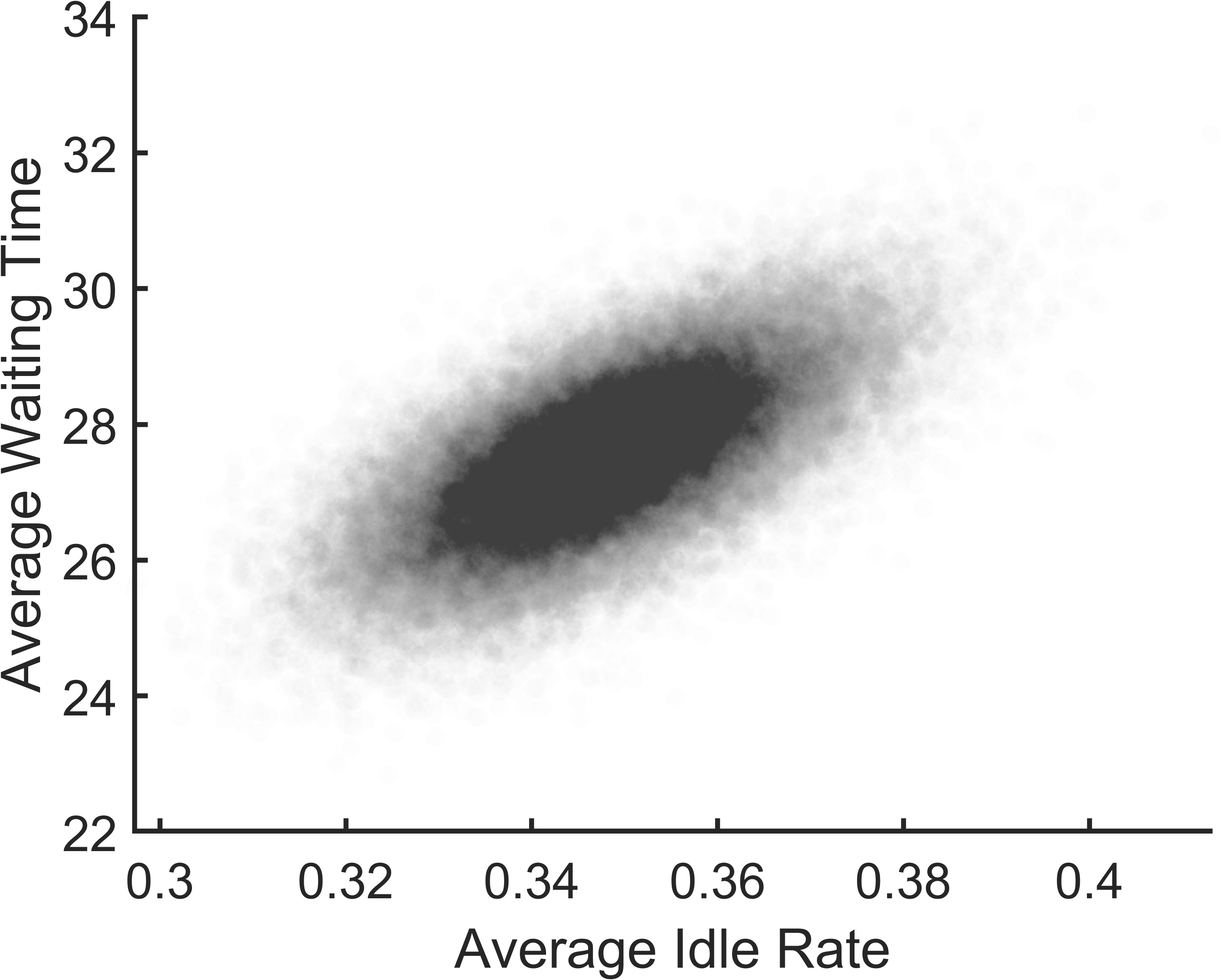}
    \label{fig:Scatter Plot_One system}}

    \caption{(a) A scatter plot of the response vectors of all systems. (b) A scatter plot of the simulation outputs vectors from $10{,}000$ replications of a representative system. (The points are semi-transparent, thus a darker area indicates a higher density of points.)
    }
    
\end{figure}

In the following subsections, we compare several screening procedures on two optimization problems posed for this queueing system: one featuring a stochastic constraint and one featuring multiple objectives.
The performance of each procedure is evaluated based on 1000 macroreplications, where on each macroreplication, each procedure simulates 50 replications at each system. For comparison purposes, all procedures receive the same set of simulation outputs on any given macroreplication.
In all experiments, we set the confidence level $1-\alpha$ to be $0.95$. 
The FOSSA procedures are implemented in MATLAB and use the built-in quadratic program solver \texttt{quadprog} with its default settings.
Source code is available at \url{https://github.com/Jinbo-Zhao/OSSA}.
While most experiments were conducted on a PC with an AMD Ryzen 5 5600X 6-Core processor (3.70 GHz) and 16 GB RAM, we leveraged a high-performance computing cluster---utilizing 48 cores on a compute node with 384GB of RAM---for the simulation replications and one screening procedure (GSPRT). Unless otherwise specified, the timing experiments were conducted on the PC.

\subsection{Optimization with Stochastic Constraints}
We first consider the case where acceptability is defined as constrained optimality. In particular, the objective is to minimize the expected average idle rate subject to a constraint that the expected average waiting time is less than or equal to $7.5$ time units.
Based on the large-sample results from Figure~\ref{fig:Scatter Plot_All systems}, an estimated 670 systems are feasible.

In this subsection, we juxtapose the three FOSSA procedures designed for constrained optimality with BootComp, a two-stage procedure that utilizes bootstrapping \citep{currie2021practical}.
BootComp proceeds in two stages, each consisting of two phases (constraints bootstrap and quality bootstrap), where the first stage uses less restrictive parameters than the second stage and is designed to prevent the oversampling of inferior systems.
The constraints bootstrap phase returns those systems whose bootstrapped secondary performance observations fall within the feasible region with probability exceeding $1-\gamma$, and the quality bootstrap phase returns those systems whose bootstrapped primary performance mean is within $100\beta$\% of that of the best system returned by the constraints bootstrap stage with probability exceeding $1-\alpha$, where $\alpha$, $\beta$, and $\gamma$ are user-specified parameters. 
BootComp aims to return a subset $\mathcal{S}^\mathrm{BC}$ of systems that are marginally feasible and near-optimal with high probabilities. The statistical guarantee sought by BootComp differs from the set-wise and system-wise PAS guarantees and is not rigorously shown to be attained, thus BootComp should be regarded as a heuristic screening procedure. Indeed, we empirically observed settings under which BootComp did not deliver its intended guarantee. 
In an attempt to provide a guarantee similar to that of FOSSA procedures, we modify the BootComp procedure for problems with expected-value constraints instead of chance constraints, as it was originally designed for.
To be precise, we modify the constraints bootstrap to return systems whose bootstrap secondary performance \emph{mean} was feasible with probability exceeding $1-\gamma$.
Furthermore, we only employ the second stage of BootComp and allow it to use all 50 samples from each system, the same as the FOSSA procedures.
These modifications enhance BootComp's screening power and alignment with established FOSSA guarantees.
We use the parameter values recommended in \cite{currie2021practical} in our experiments: $\alpha=0.95$, $\beta=0.05$ and $\gamma=0.95$.

For each of the four procedures, the probability of acceptable selection, size of the returned subset, running time and average quality of the returned subsets are reported in Table~\ref{table: Performance Constrained}.
The FOSSA procedures return a subset containing the unique optimal system in all 1000 macroreplications whereas BootComp does so on only 95 macroreplications.
The FOSSA procedures screen out 95\% of systems on average compared to 99\% by
BootComp. This can be seen in the heatmap of Figure~\ref{fig:Scatter Plot_Constrained}. (Because the results for the three FOSSA procedures are very similar, only those for FOSSA Half-Box are displayed.) 
Figure~\ref{fig:Scatter Plot_Constrained} also shows that most of systems returned by the FOSSA procedures are located near the true optimal system in the output space.
To quantify this, we measure each system's \emph{distance to acceptability} as the sum of its relative optimality gap and relative distance to feasibility, with smaller values indicating better quality; more details about this metric are provided in the electronic companion.
Figure~\ref{fig: Quality of Returned Systems (Constrained)} shows a breakdown of the quality of the solutions returned by FOSSA Half-Box relative to all solutions under consideration on a representative macroreplication.
All of the procedures return subsets consisting of high-quality solutions, though the average quality of the systems in BootComp's returned subset is more variable, as seen in the wider range of its percentiles.
In terms of timing, all four procedures take less than one second to run one macroreplication, with FOSSA Box and FOSSA Half-Box taking only several milliseconds, over 200 times faster than FOSSA Ellipsoid and BootComp.
This efficiency suggests that FOSSA procedures could be applied to stochastically constrained problems with large numbers, say millions, of systems.
In summary, the experimental results demonstrate that although BootComp tends to return a smaller subset consisting of near-optimal systems, it fails to return the optimal system with high probability. In contrast, the FOSSA procedure rapidly provides a subset delivering the PAS guarantee and requires fewer user-specified parameters.

\begin{table}[]
\caption{Performance metrics for the three FOSSA procedures and BootComp: probability of acceptable selection (PAS), size of returned subset, average quality of systems in the returned subset and wall-clock time. The $90$\% confidence intervals for PAS are two-tailed Clopper-Pearson confidence intervals. The averages and 10\% and 90\% percentiles are calculated from 1000 macroreplications.}
\label{table: Performance Constrained}
\begin{center}
\resizebox{\columnwidth}{!}{
\begin{tabular}{cc|cccc}

\multicolumn{2}{c|}{\multirow{2}{*}{Metric}}                                                                        & \multicolumn{4}{c}{Procedure}                                                            \\ \cline{3-6} 
\multicolumn{2}{c|}{}                                                                                                     & Ellipsoid           & Box                 & Half-Box            & BootComp            \\ \hline
\multicolumn{1}{c|}{\multirow{2}{*}{PAS}}                                                       & Avg                 & 1                   & 1                   & 1                   & 0.095               \\
\multicolumn{1}{c|}{}                                                                           & 90\%CI                  & {[}0.9992, 1{]}      & {[}0.9992, 1{]}      & {[}0.9992, 1{]}      & {[}0.0945, 0.0968{]} \\ \hline
\multicolumn{1}{c|}{\multirow{2}{*}{\begin{tabular}[c]{@{}c@{}}Subset\\ Size\end{tabular}}}     & Avg                 & 57.60               & 53.68               & 50.98               & 12.22               \\
\multicolumn{1}{c|}{}                                                                           & {[}10\%, 90\%{]} & {[}52, 62{]}         & {[}48, 59{]}         & {[}46, 57{]}         & {[}3, 18{]}          \\ \hline

\multicolumn{1}{c|}{\multirow{2}{*}{\begin{tabular}[c]{@{}c@{}}Average\\ Quality\end{tabular}}} & Avg                 & 0.0459              & 0.0441              & 0.0427              & 0.0411              \\
\multicolumn{1}{c|}{}                                                                           & {[}10\%, 90\%{]} & {[}0.0412, 0.0496{]} & {[}0.0393, 0.0484{]} & {[}0.0382, 0.0474{]} & {[}0.0221, 0.0506{]} \\ \hline

\multicolumn{1}{c|}{\multirow{2}{*}{\begin{tabular}[c]{@{}c@{}}Wall-Clock\\ Time (s)\end{tabular}}}     & Avg                 & 0.7527              & 0.0035              & 0.0033              & 0.8360              \\
\multicolumn{1}{c|}{}                                                                           & {[}10\%, 90\%{]} & {[}0.7037, 0.7849{]} & {[}0.0032, 0.0037{]} & {[}0.0030, 0.0035{]} & {[}0.7745, 0.9162{]} \\

\end{tabular}
}
\end{center}
\end{table}

\begin{figure}[]
    \centering
     
     \subfigure[FOSSA Half-Box]
    {\includegraphics[width=0.45\textwidth]{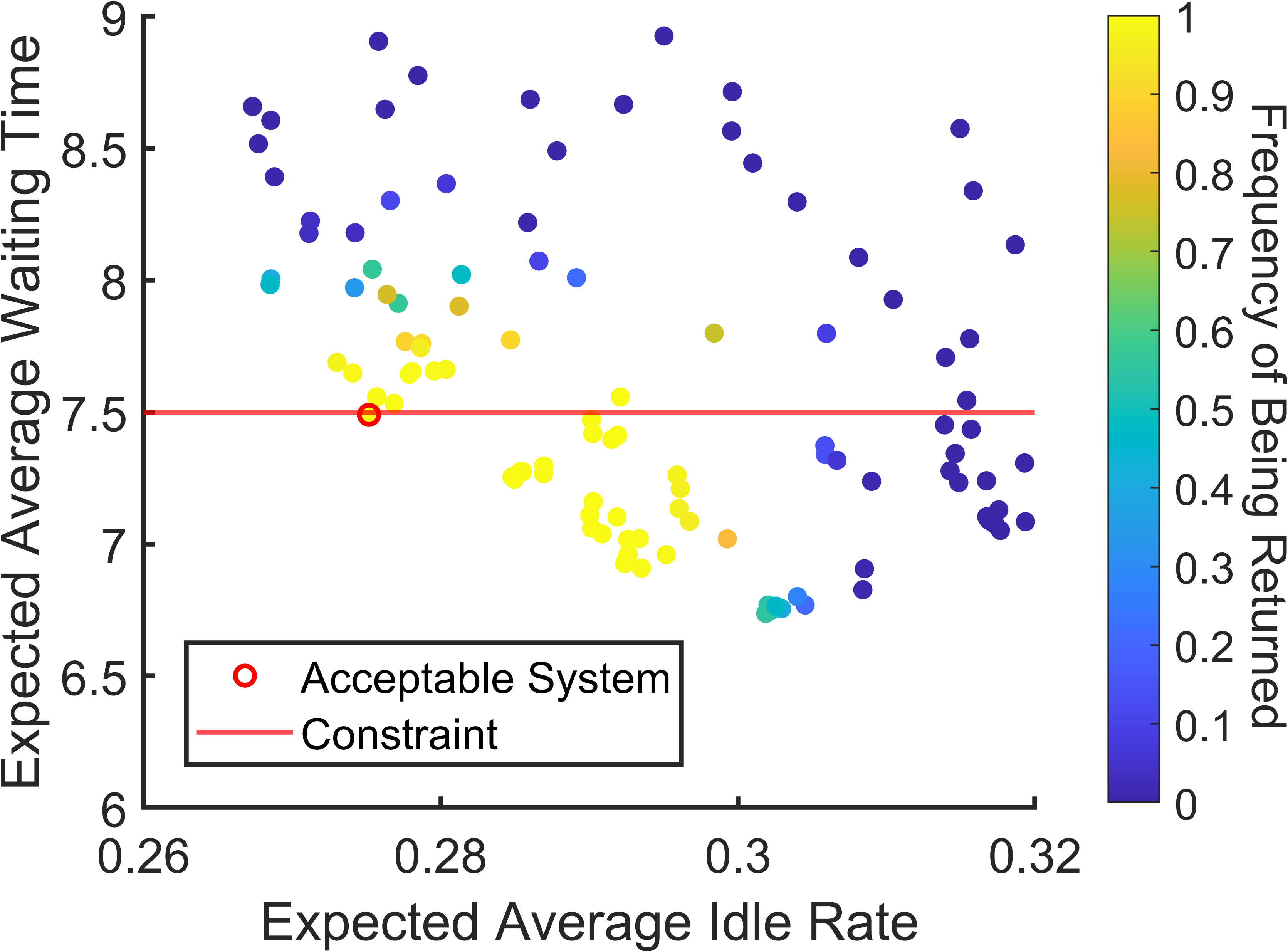}}
    \hfill
    \subfigure[BootComp]
     {\includegraphics[width=0.45\textwidth]{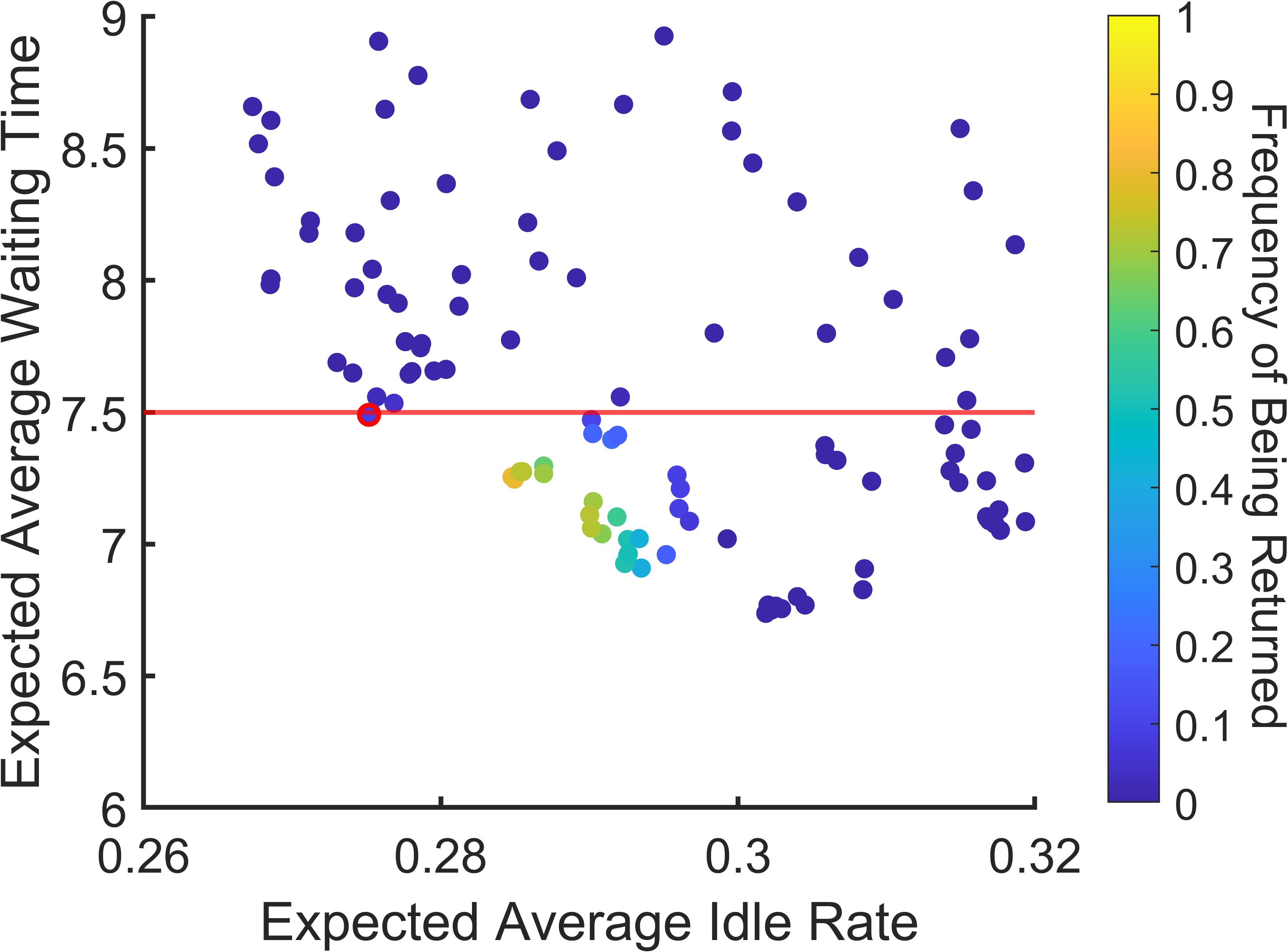}}

    \caption{Scatters plot of the estimated response vectors of systems in a neighborhood of the optimal system. The region below the horizontal red line is the feasible region and the true (constrained) optimal system is circled in red. The color of each dot indicates its frequency of being returned by (a) FOSSA Half-Box and (b) BootComp. Systems not shown were never returned in any of the 1000 macroreplications.}
    \label{fig:Scatter Plot_Constrained}
\end{figure}

\begin{figure}[]
    \centering
    \includegraphics[width=1\textwidth]{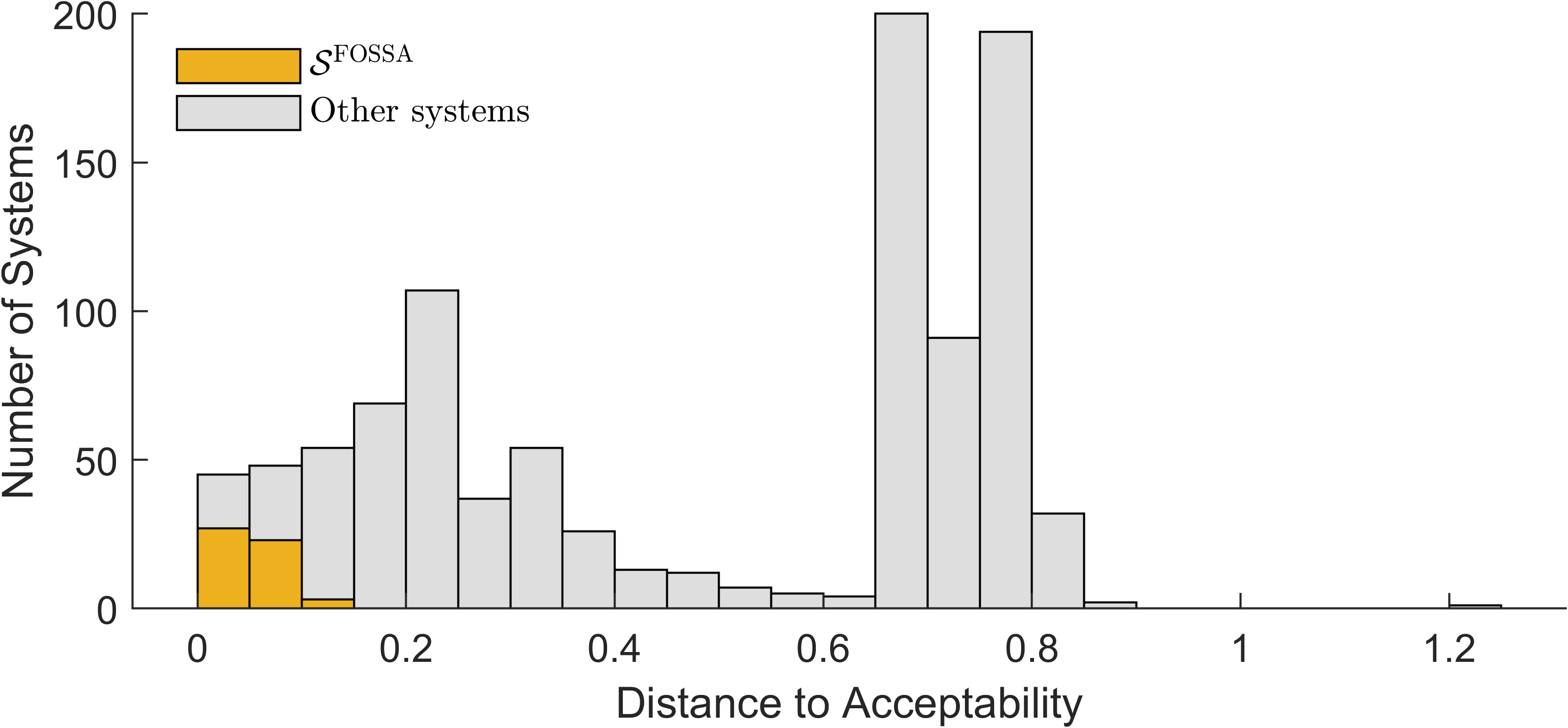}
    \caption{Histogram of the distance to acceptability of the systems returned by FOSSA Half-Box (yellow) on a representative macroreplication and all other systems under consideration (grey). Smaller distance to acceptability indicates a higher quality system.} 
    \label{fig: Quality of Returned Systems (Constrained)}
\end{figure}

\subsection{Pareto Optimality}
We next consider the case where acceptability is defined as Pareto optimality and the decision maker prefers systems with lower expected average waiting time and lower expected average idle rates.
Based on the large-sample runs, there are 72 acceptable systems.
We compare the three FOSSA procedures for Pareto optimality with the generalized sequential probability ratio test (GSPRT) of \cite{wang2019sequential}. 
GSPRT is a sequential selection procedure that maintains a common sample size across all systems and guarantees to return the exact efficient set with a probability of at least $1-\alpha$ as $\alpha$ approaches zero \citep{wang2017sequential}.
When GSPRT is stopped before its termination criteria, one can extract a subset that delivers the set-wise PAS guarantee in the aforementioned asymptotic sense. 
The most powerful screening version of GSPRT conducts hypothesis tests as observations are collected and continuously monitors a subset of systems believed to be in the efficient set until some stopping condition is met.
Due to the relatively high computational cost of running GSPRT, we choose to test the hypothesis only once, using all available samples; we believe this should only marginally weaken the procedure's screening power.
The guarantees of GSPRT rely on two assumptions: that each system's true covariance matrix is known and that an explicit likelihood function for the outputs generated by all systems is available. The latter assumption complicates the procedure's potential to be extended to non-parametric settings or cases involving the use of common random numbers.

Table \ref{table: Performance Pareto} shows the results for the three FOSSA procedures and GSPRT. All four procedures return all Pareto optimal systems on all macroreplications.
GSPRT returns subsets that are about 10\% smaller than those returned by FOSSA. The large sizes of the returned sets for all procedures (about 700 systems) is a consequence of how for this problem instance many systems are close to the true Pareto front.
The heatmaps in Figure \ref{fig:Scatter Plot_Pareto} show that the procedures frequently return systems that are close to the true Pareto front. (Again, the results for all FOSSA procedures were similar, thus we display the heatmap for only FOSSA Box.)
Figure \ref{fig:Scatter Plot_Pareto_Box} also depicts a $1-\alpha$ confidence region for the Pareto front provided by FOSSA Box.
The quality of a returned system is measured in terms of how close its response vector is to being non-dominated by all other systems, with smaller values being better. More details about this metric are given in the
electronic companion.
Figure \ref{fig:Quality of Returned Systems(Pareto)} shows the quality of individual systems returned by GSPRT and FOSSA Box on a representative macroreplication; both sets consist mainly of near-optimal systems.
We note that the returned set of GSPRT procedure turned out to be a subset of that of all FOSSA procedures on all macroreplications.

Table~\ref{table: Performance Pareto} also shows that the FOSSA procedures are computationally efficient, with FOSSA Box and FOSSA Half-Box screening all 1000 systems in less than a hundredth of a second and FOSSA Ellipsoid taking about 11 seconds on average.
On the other hand, the GSPRT procedure is more computationally expensive due to its need to solve harder optimization problems.
For this reason, GSPRT is run on a 48-core high-performance computing node and takes a total of nine hours and five minutes for the 1000 macroreplications.
The nested relationship between the returned subsets of GSPRT and FOSSA observed empirically and the timing results suggest that the FOSSA procedure could be used to ``pre-screen'' systems before running GSPRT, without compromising any screening power. Pre-screening refers to the combined approach of first applying FOSSA on all systems and then applying another, typically more computationally expensive, screening procedure on the surviving systems, without taking additional replications. Pre-screening may be promising in large-scale instances, but should be viewed as a heuristic in the case of GSPRT since we were unable to formalize a nested relationship between the FOSSA and GSPRT subsets.

Although the numerical results demonstrate that GSPRT can effectively screen out unacceptable systems on this particular problem instance, several issues hinder the procedure's wider use.
These include its assumption of known covariance matrices, considerable computational expense, and its asymptotic guarantee as $\alpha \to 0$.
In addition, GSPRT relies on maximizing the likelihood function for pairwise comparisons, which makes it less suitable for instances in which the normality assumption does not hold.
FOSSA procedures, on the other hand, suffer from none of these issues and offer only marginally less screening power than GSPRT.
FOSSA procedures' confidence guarantees, fast speed and  adaptability to non-parametric settings make them a compelling choice for screening with multiple objectives.

\begin{table}[]
\caption{Performance metrics for the three FOSSA procedures and GSPRT: set-wise probability of acceptable selection (PAS), size of returned subset, average quality of systems in the returned subset and wall-clock time. The $90$\% confidence intervals for PAS are two-tailed Clopper-Pearson confidence intervals. The averages and 10\% and 90\% percentiles are calculated from 1000 macroreplications.}
\label{table: Performance Pareto}
\begin{center}
\resizebox{\columnwidth}{!}{
\begin{tabular}{cc|cccc}
\multicolumn{2}{c|}{\multirow{2}{*}{Metric}}                                                                                                                                          & \multicolumn{4}{c}{Procedure}                                                             \\ \cline{3-6} 
\multicolumn{2}{c|}{}                                                                                                                                                                      & Ellipsoid            & Box                 & Half-Box            & GSPRT               \\ \hline
\multicolumn{1}{c|}{\multirow{2}{*}{PAS}}                                                       & \multirow{2}{*}{\begin{tabular}[c]{@{}c@{}}Avg\\ 90\% CI\end{tabular}}                & 1                    & 1                   & 1                   & 1                   \\
\multicolumn{1}{c|}{}                                                                           &                                                                                          & {[}0.9992, 1{]}       & {[}0.9992, 1{]}      & {[}0.9992, 1{]}      & {[}0.9992, 1{]}      \\ \hline
\multicolumn{1}{c|}{\multirow{2}{*}{\begin{tabular}[c]{@{}c@{}}Subset\\ Size\end{tabular}}}     & \multirow{2}{*}{\begin{tabular}[c]{@{}c@{}}Avg\\ {[}10\%, 90\%{]}\end{tabular}} & 740.25               & 739.62              & 733.02              & 676.77              \\
\multicolumn{1}{c|}{}                                                                           &                                                                                          & {[}735, 745{]}        & {[}734, 745{]}       & {[}727, 738{]}       & {[}669, 684{]}       \\ \hline

\multicolumn{1}{c|}{\multirow{2}{*}{\begin{tabular}[c]{@{}c@{}}Average\\ Quality\end{tabular}}} & \multirow{2}{*}{\begin{tabular}[c]{@{}c@{}}Avg\\ {[}10\%, 90\%{]}\end{tabular}} & 0.0043               & 0.0042              & 0.0040              & 0.0025              \\
\multicolumn{1}{c|}{}                                                                           &                                                                                          & {[}0.0040, 0.0045{]}  & {[}0.0040, 0.0044{]} & {[}0.0037, 0.0042{]} & {[}0.0023, 0.0027{]} \\ \hline

\multicolumn{1}{c|}{\multirow{2}{*}{\begin{tabular}[c]{@{}c@{}}Wall-Clock\\ Time (s) \end{tabular}}}    & \multirow{2}{*}{\begin{tabular}[c]{@{}c@{}}Avg\\ {[}10\%, 90\%{]}\end{tabular}} & 11.1479              & 0.0082              & 0.0076              & 1596.6$^*$    \\
\multicolumn{1}{c|}{}                                                                           &                                                                                          & {[}7.2991, 15.1787{]} & {[}0.0078, 0.0085{]} & {[}0.0072, 0.0080{]} & --    
\end{tabular}}
\end{center}
\end{table}

\begin{figure}[]
    \centering
     \subfigure[FOSSA Box\label{fig:Scatter Plot_Pareto_Box}]
    {\includegraphics[width=0.45\textwidth]{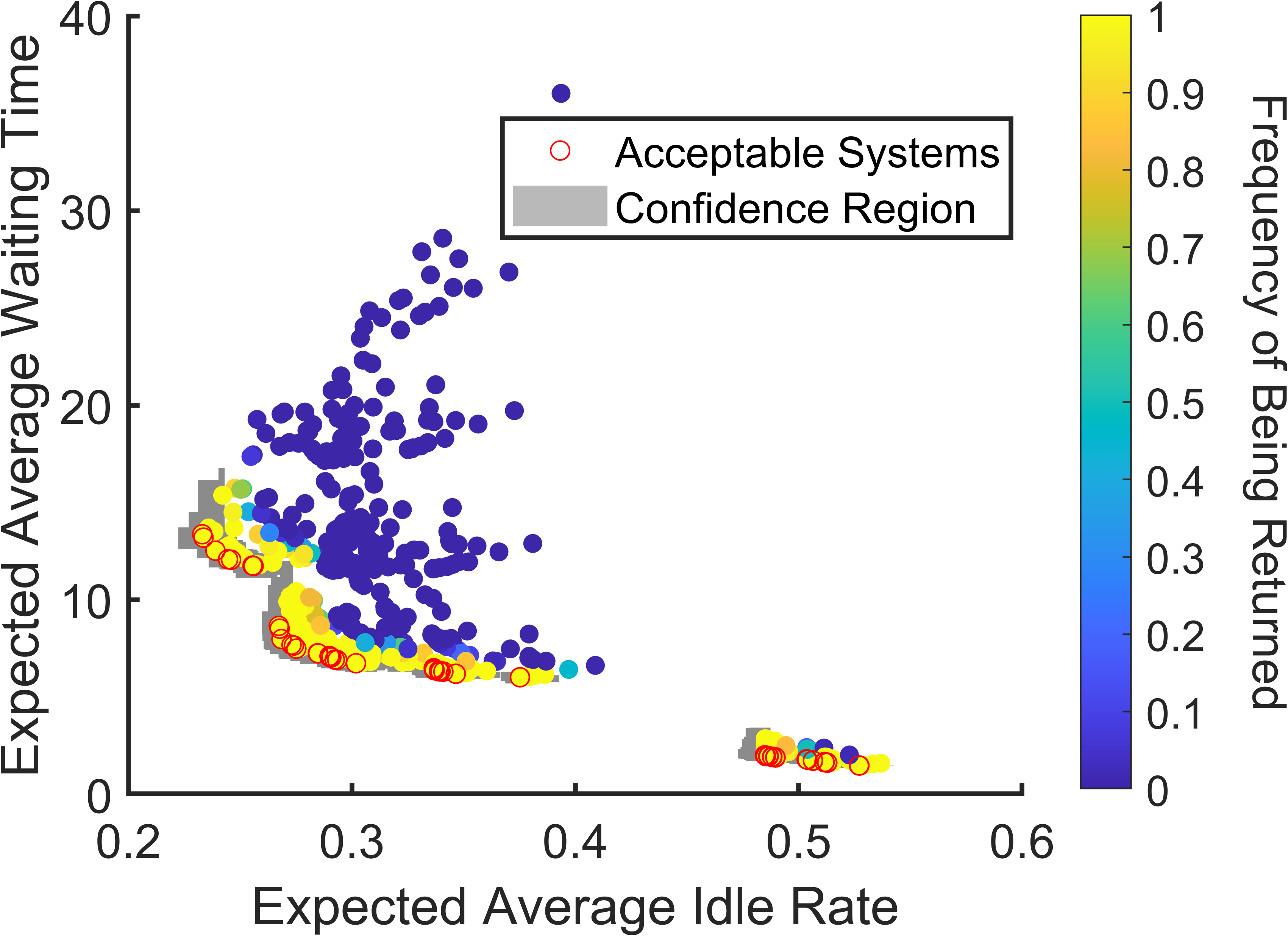}
    }
    \hfill
    \subfigure[GSPRT]
     {\includegraphics[width=0.45\textwidth]{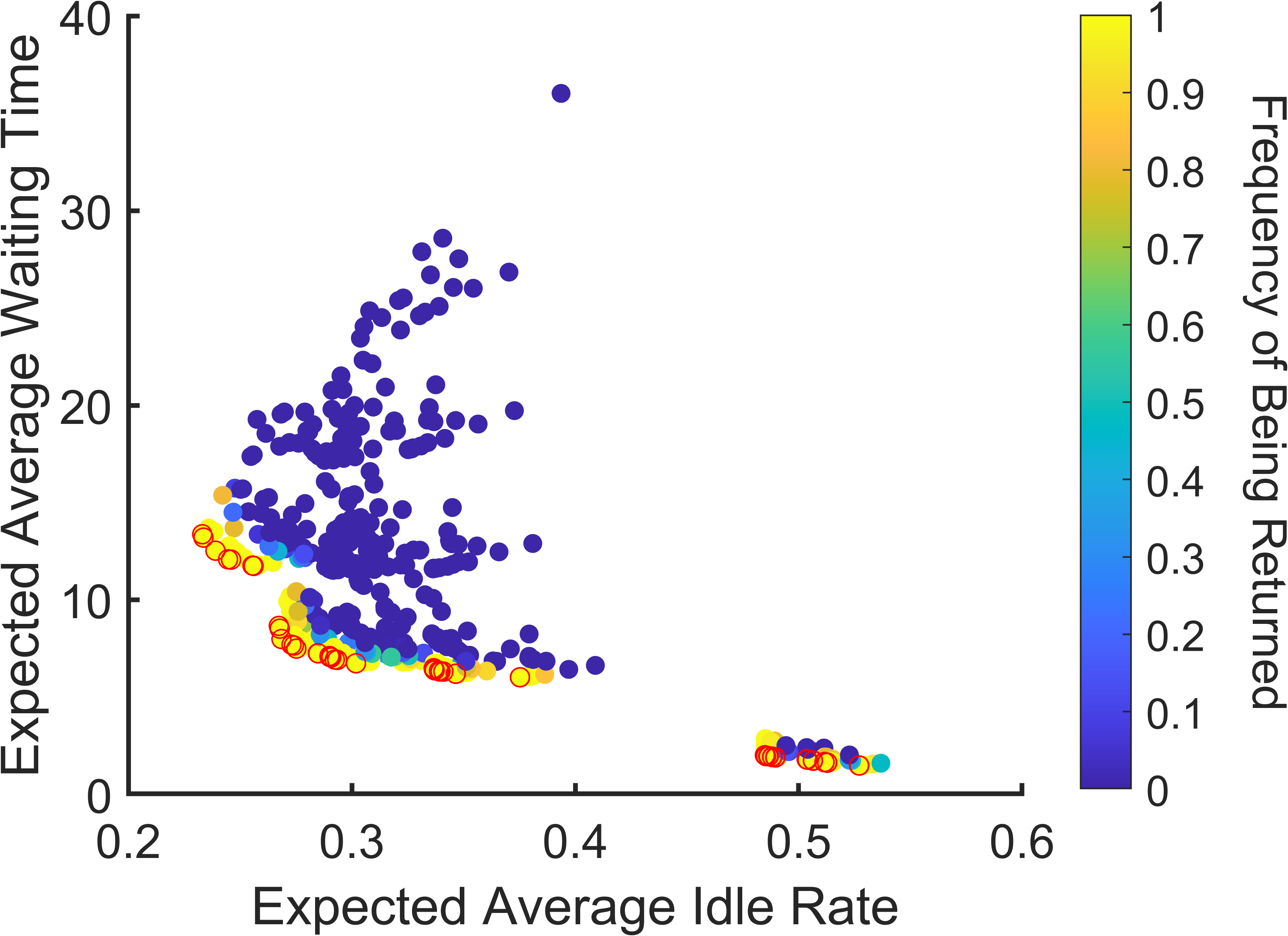}}
    \caption{Scatter plots of the response vectors of all systems under consideration. The color of each dot indicates its frequency of being returned by (a) FOSSA Box and (b) GSPRT. The 72 acceptable systems are circled in red. The grey-shaded area is a 95\% confidence region for the true Pareto front, provided by FOSSA Box on a representative macroreplication.}
    \label{fig:Scatter Plot_Pareto}
\end{figure}

\begin{figure}[]
    \centering
    \includegraphics[width=1\textwidth]{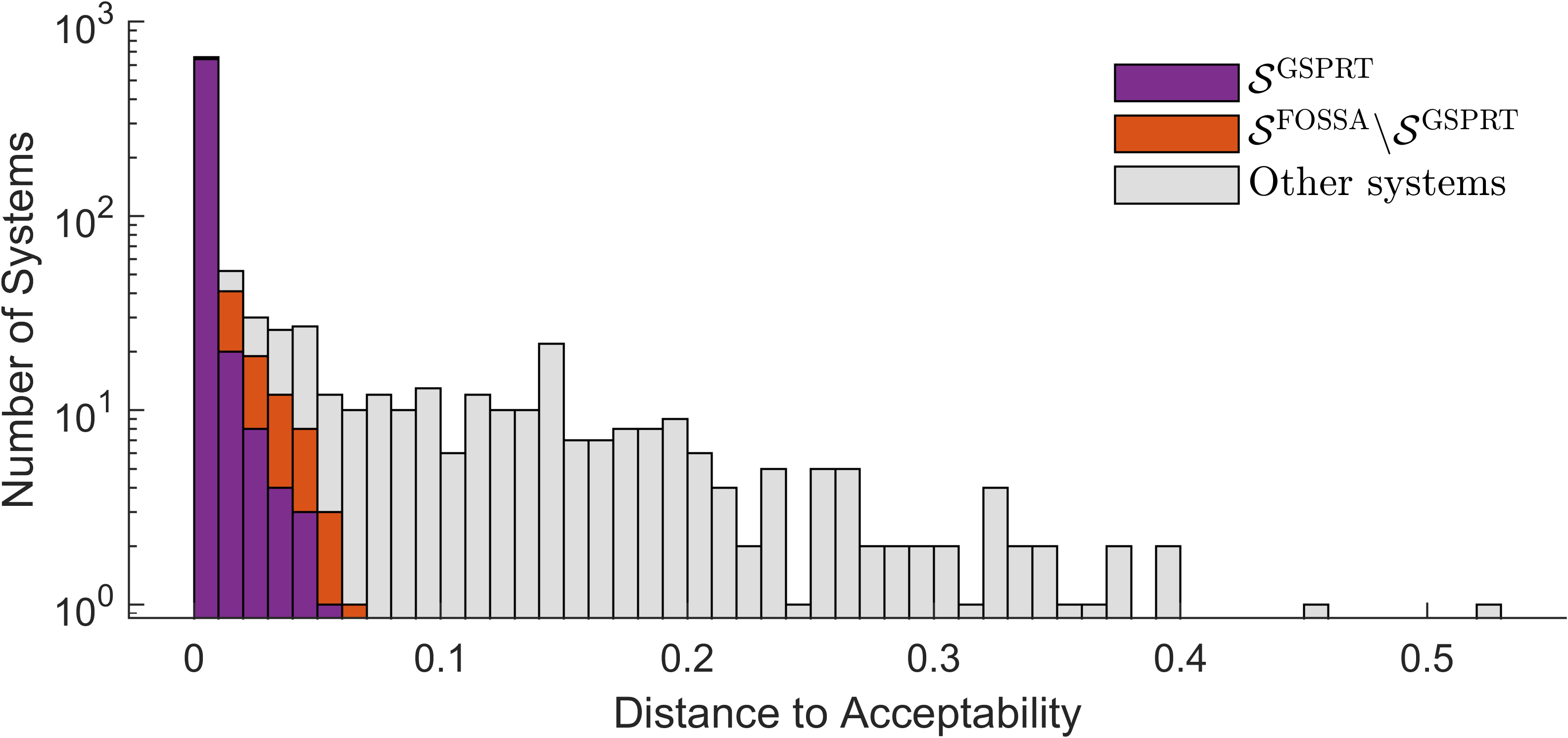}
     \caption{Histogram of the distance to acceptability of systems returned by GSPRT (purple) and additional systems returned by FOSSA Box (orange) on a representative macroreplication and all other systems under consideration (grey). Smaller distance to acceptability indicates a higher quality system.
    }
    \label{fig:Quality of Returned Systems(Pareto)}
\end{figure}

\section{Conclusion}
\label{sec:conclusion}

This article introduces OSSA, a new design framework for one-shot screening procedures that accommodates many commonly studied definitions of acceptability. The FOSSA subclass of the framework facilitates the design of procedures that are computationally efficient and easily parallelized. We introduce several FOSSA procedure that offer finite-sample confidence guarantees for optimization under stochastic constraints and multi-objective optimization, which we believe to be the first of their kind. Based on a numerical study, we recommend using the FOSSA Box procedure when the set-wise PAS guarantee is sought and the FOSSA Half-Box procedure for the system-wise PAS guarantee. Both procedures offer comparable screening power to FOSSA Ellipsoid but are significantly faster and do not require sophisticated optimization software. 

There are several promising directions for extending the OSSA framework. These include, but are not limited to, incorporating CRN within the framework, possibly through the development of specialized confidence regions that handle correlation across systems; constructing non-parametric confidence regions for FOSSA procedures; investigating more elaborate, yet practical, definitions of acceptability, such as Pareto optimality combined with constraints; and tailoring confidence regions to the definition of acceptability.
Another exciting direction for further development is leveraging functional information over the decision variable space to bolster the methods' screening power and to exclude unsimulated systems, in a manner similar to plausible screening.

\section*{Acknowledgments}
This work was partially supported by National Science Foundation Grant CMMI-2206972.
Some of the numerical experiments were conducted with the advanced computing resources provided by Texas A\&M High Performance Research Computing.
We also wish to thank Barry Nelson for helpful discussions and the creators of PyMOSO for their functions for finding phantom Pareto systems.

\bibliographystyle{plainnat}
\bibliography{bibliography}

@article{nelson2001simple,
  title={Simple Procedures for Selecting the Best Simulated System when the Number of Alternatives is Large},
  author={Nelson, Barry L and Swann, Julie and Goldsman, David and Song, Wheyming},
  journal={Operations Research},
  volume={49},
  number={6},
  pages={950--963},
  year={2001},
  publisher={INFORMS}
}

@article{boesel2003using,
  title={Using Ranking and Selection to ``Clean Up'' after Simulation Optimization},
  author={Boesel, Justin and Nelson, Barry L and Kim, Seong-Hee},
  journal={Operations Research},
  volume={51},
  number={5},
  pages={814--825},
  year={2003},
  publisher={INFORMS}
}

@article{gupta1965some,
  title={On Some Multiple Decision (Selection and Ranking) Rules},
  author={Gupta, Shanti S},
  journal={Technometrics},
  volume={7},
  number={2},
  pages={225--245},
  year={1965},
  publisher={Taylor \& Francis}
}

@INCOLLECTION{miescke:bayesdesigns99,
  author={Miescke, Klaus J},
  title = {Bayes Sampling Designs for Selection Procedures},
  booktitle = {Multivariate analysis, design of experiments, and survey sampling},
  publisher={CRC Press},
  year = {1999},
  editor = {S Ghosh},
  address = {New York},
  page = {117--142}
}

@ARTICLE{rinott:twostage,
  author = {Y Rinott},
  title = {On Two-Stage Selection Procedures and Related Probability-Inequalities},
  journal = {Communications in Statistics--Theory and Methods},
  volume = {7},
  number = {8},
  pages = {799-811},
  year = {1978},
  publisher = {Taylor \& Francis},
}

@article{dudewicz:twostage,
  title={Allocation of observations in ranking and selection with unequal variances},
  author={Edward J. Dudewicz and Siddhartha R. Dalal},
  journal={Sankhy{\=a}: The Indian Journal of Statistics},
  volume={37},
  number={1},
  pages={28--78},
  year={1975},
  publisher={Springer}
}

@BOOK{anderson:mvstats,
  title = {An Introduction to Multivariate Statistical Analysis},
  publisher = {John Wiley \& Sons},
  year = {1984},
  author = {T. W. Anderson},
  address = {New York},
  edition = {2nd}
}

@inproceedings{batur2010mean,
  author={Batur, Demet and Choobineh, F. Fred},
  booktitle={Proceedings of the 2010 Winter Simulation Conference}, 
  title={Mean-variance based ranking and selection}, 
  year={2010},
  editor={Johansson, Bj{\"o}rn and Jain, Sanjay amd Montoya-Torres, Jairo R},
  pages={1160--1166},
  publisher={Institute of Electrical and Electronics Engineers, Inc.},
  address={Piscataway, New Jersey}
}

@article{Ni2017,
author = {Ni, Eric C. and Ciocan, Dragos F. and Hunter, Susan R. and Henderson, Shane G.},
journal = {Operations Research},
number = {3},
pages = {821--836},
title = {{Efficient ranking and selection in parallel computing environments}},
volume = {65},
year = {2017}
}

@article{baturkim2010,
author = {Batur, Demet and Kim, Seong-Hee},
title = {Finding feasible systems in the presence of constraints on multiple performance measures},
journal = {ACM Transactions on Modeling and Computer Simulation (TOMACS)},
volume = {20},
number = {3},
pages = {Article 13. 1--26},
year = {2010}
}

@article{EckmanHenderson2021,
  title={Fixed-confidence, fixed tolerance guarantees for selection-of-the-best procedures},
  author={David J. Eckman and Shane G. Henderson},
  journal={ACM Transactions on Modeling and Computer Simulation (TOMACS)},
  year={2021},
  volume={31},
  number={2},
  pages={Article 7. 1--33}
}

@inproceedings{eckman2020revisiting,
  title={Revisiting subset selection},
  author={Eckman, David J and  Plumlee, Matthew and  Nelson, Barry L},
  booktitle={2020 Winter Simulation Conference (WSC)},
  pages={2972-2983},
  year={2020},
  editor={Bae, K-H and Feng, B and Kim, S and Lazarova-Molnar, S and Zheng, Z and Roeder, T and Thiesing, R},
  organization={Institute of Electrical and Electronics Engineers},
  address={Piscataway, NJ}
}

@ARTICLE{andradottir2010fully,
  title={Fully Sequential Procedures for Comparing Constrained Systems via Simulation},
  author={Andrad{\'o}ttir, Sigr{\'u}n and Kim, Seong-Hee},
  journal={Naval Research Logistics (NRL)},
  volume={57},
  number={5},
  pages={403--421},
  year={2010},
  publisher={Wiley Online Library}
}

@ARTICLE{hong2015chance,
  title={Chance Constrained Selection of the Best},
  author={Hong, L Jeff and Luo, Jun and Nelson, Barry L},
  journal={INFORMS Journal on Computing},
  volume={27},
  number={2},
  pages={317--334},
  year={2015},
  publisher={INFORMS}
}

@ARTICLE{gupta1985subset,
  title={Subset Selection Procedures: Review and Assessment},
  author={Gupta, Shanti S and Panchapakesan, S},
  journal={American Journal of Mathematical and Management Sciences},
  volume={5},
  number={3--4},
  pages={235--311},
  year={1985},
  publisher={Taylor \& Francis}
}

@ARTICLE{eckman2022plausible,
  title={Plausible Screening Using Functional Properties for Simulations With Large Solution Spaces},
  author={Eckman, David J and Plumlee, Matthew and Nelson, Barry L},
  journal={Operations Research},
  volume={70},
  number={6},
  pages={3473--3489},
  year={2022},
  publisher={INFORMS}
}

@ARTICLE{gao2015efficient,
  title={Efficient Subset Selection for the Expected Opportunity Cost},
  author={Gao, Siyang and Chen, Weiwei},
  journal={Automatica},
  volume={59},
  pages={19--26},
  year={2015},
  publisher={Elsevier}
}

@ARTICLE{lam1986new,
  title={A New Procedure for Selecting Good Populations},
  author={Lam, K},
  journal={Biometrika},
  volume={73},
  number={1},
  pages={201--206},
  year={1986},
  publisher={Oxford University Press}
}

@ARTICLE{sullivan1989restricted,
  title={Restricted Subset Selection Procedures for Simulation},
  author={Sullivan, David W and Wilson, James R},
  journal={Operations Research},
  volume={37},
  number={1},
  pages={52--71},
  year={1989},
  publisher={INFORMS}
}

@article{pei2022parallel,
  title={Parallel Adaptive Survivor Selection},
  author={Pei, Linda and Nelson, Barry L and Hunter, Susan R},
  journal={Operations Research},
  publisher={INFORMS},
  volume = {72},
  number = {1},
  pages = {336-354},
  year = {2024},
}

@INPROCEEDINGS{zhao23screening,
  title={Screening Simulated Systems for Optimization},
  author={Zhao, Jinbo and Gatica, Javier and Eckman, David J},
  booktitle={Proceedings of the 2023 Winter Simulation Conference},
  year={2023},
  address={Piscataway, New Jersey},
  organization={Institute of Electrical and Electronics Engineers, Inc.},
  editor = {G Corlu and S R Hunter and H Lam and B S Onggo and J Shortle and B Biller},
pages={1-15}
}

@article{batur2010quantile,
  title={A quantile-based approach to system selection},
  author={Batur, Demet and Choobineh, F},
  journal={European Journal of Operational Research},
  volume={202},
  number={3},
  pages={764--772},
  year={2010},
  publisher={Elsevier}
}

@article{shin2022practical,
  title={Practical nonparametric sampling strategies for quantile-based ordinal optimization},
  author={Shin, Dongwook and Broadie, Mark and Zeevi, Assaf},
  journal={INFORMS Journal on Computing},
  volume={34},
  number={2},
  pages={752--768},
  year={2022},
  publisher={INFORMS}
}

@article{bechhofer1954single,
  title={A single-sample multiple decision procedure for ranking means of normal populations with known variances},
  author={Bechhofer, Robert E},
  journal={The Annals of Mathematical Statistics},
  volume={25},
  number={1},
  pages={16--39},
  year={1954},
  publisher={JSTOR}
}

@article{solow2021novel,
  title={Novel approaches to feasibility determination},
  author={Solow, Daniel and Szechtman, Roberto and Y{\"u}cesan, Enver},
  journal={ACM Transactions on Modeling and Computer Simulation (TOMACS)},
  volume={31},
  number={1},
  pages={1--25},
  year={2021},
  publisher={ACM New York, NY, USA}
}

@article{applegate2020multi,
  title={Multi-objective ranking and selection: Optimal sampling laws and tractable approximations via {SCORE}},
  author={Applegate, Eric A and Feldman, Guy and Hunter, Susan R and Pasupathy, Raghu},
  journal={Journal of Simulation},
  volume={14},
  number={1},
  pages={21--40},
  year={2020},
  publisher={Taylor \& Francis}
}

@article{andradottir2021pareto,
  title={Pareto set estimation with guaranteed probability of correct selection},
  author={Andrad{\'o}ttir, Sigr{\'u}n and Lee, Judy S},
  journal={European Journal of Operational Research},
  volume={292},
  number={1},
  pages={286--298},
  year={2021},
  publisher={Elsevier}
}

@article{pasupathy2014stochastically,
  title={Stochastically constrained ranking and selection via {SCORE}},
  author={Pasupathy, Raghu and Hunter, Susan R and Pujowidianto, Nugroho A and Lee, Loo Hay and Chen, Chun-Hung},
  journal={ACM Transactions on Modeling and Computer Simulation (TOMACS)},
  volume={25},
  number={1},
  pages={1--26},
  year={2014},
  publisher={ACM New York, NY, USA}
}

@article{chen2008efficient,
  title={Efficient simulation budget allocation for selecting an optimal subset},
  author={Chen, Chun-Hung and He, Donghai and Fu, Michael and Lee, Loo Hay},
  journal={INFORMS Journal on Computing},
  volume={20},
  number={4},
  pages={579--595},
  year={2008},
  publisher={INFORMS}
}

@article{koenig1985procedure,
  title={A procedure for selecting a subset of size $m$ containing the $l$ best of $k$ independent normal populations, with applications to simulation},
  author={Koenig, Lloyd W and Law, Averill M},
  journal={Communications in Statistics--Simulation and Computation},
  volume={14},
  number={3},
  pages={719--734},
  year={1985},
  publisher={Taylor \& Francis}
}

@inbook{smith2018simio,
  author={Smith, Jeffrey S and Sturrock, David T},
  title={Simio and Simulation: Modeling, Analysis, Applications},
  chapter={9},
  publisher={Simio LLC},
  year={2021},
  edition={6th},
  address={Pittsburgh, PA},
  pages={355--404}
}

@article{peng2021efficient,
  title={Efficient sampling allocation procedures for optimal quantile selection},
  author={Peng, Yijie and Chen, Chun-Hung and Fu, Michael C and Hu, Jian-Qiang and Ryzhov, Ilya O},
  journal={INFORMS Journal on Computing},
  volume={33},
  number={1},
  pages={230--245},
  year={2021},
  publisher={INFORMS}
}

@incollection{chen2015ranking,
  title={Ranking and selection: Efficient simulation budget allocation},
  author={Chen, Chun-Hung and Chick, Stephen E and Lee, Loo Hay and Pujowidianto, Nugroho A},
  booktitle={Handbook of Simulation Optimization},
  editor={Michael C. Fu},
  pages={45--80},
  year={2015},
  publisher={Springer},
  volume={216},
  chapter={3}
}

@article{gao2017efficient,
  title={Efficient feasibility determination with multiple performance measure constraints},
  author={Gao, Siyang and Chen, Weiwei},
  journal={IEEE Transactions on Automatic Control},
  volume={62},
  number={1},
  pages={113--122},
  year={2017},
  publisher={IEEE}
}

@article{healey2014selection,
  title={Selection procedures for simulations with multiple constraints under independent and correlated sampling},
  author={Healey, Christopher and Andrad{\'o}ttir, Sigr{\'u}n and Kim, Seong-Hee},
  journal={ACM Transactions on Modeling and Computer Simulation (TOMACS)},
  volume={24},
  number={3},
  pages={1--25},
  year={2014},
  publisher={ACM New York, NY, USA}
}

@article{butler2001multiple,
  title={A multiple attribute utility theory approach to ranking and selection},
  author={Butler, John and Morrice, Douglas J and Mullarkey, Peter W},
  journal={Management Science},
  volume={47},
  number={6},
  pages={800--816},
  year={2001},
  publisher={INFORMS}
}

@inproceedings{szechtman2008new,
  title={A new perspective on feasibility determination},
  author={Szechtman, Roberto and Yucesan, Enver},
  booktitle={Proceedings of the 2008 Winter Simulation Conference},
  pages={273--280},
  year={2008},
  address={Piscataway, New Jersey},
  organization={Institute of Electrical and Electronics Engineers, Inc.},
editor = {Scott J. Mason and Raymond R. Hill and Lars M{\"o}nch and Oliver Rose and Thomas Jefferson and John W. Fowler},
}

@article{lee2010finding,
  title={Finding the non-dominated {P}areto set for multi-objective simulation models},
  author={Lee, Loo Hay and Chew, Ek Peng and Teng, Suyan and Goldsman, David},
  journal={IIE Transactions},
  volume={42},
  number={9},
  pages={656--674},
  year={2010},
  publisher={Taylor \& Francis}
}

@article{li2018optimal,
  title={Optimal computing budget allocation to select the nondominated systems---A large deviations perspective},
  author={Li, Juxin and Liu, Weizhi and Pedrielli, Giulia and Lee, Loo Hay and Chew, Ek Peng},
  journal={IEEE Transactions on Automatic Control},
  volume={63},
  number={9},
  pages={2913--2927},
  year={2018},
  publisher={IEEE}
}

@article{yan2012efficient,
  title={Efficient selection of a set of good enough designs with complexity preference},
  author={Yan, Shen and Zhou, Enlu and Chen, Chun-Hung},
  journal={IEEE Transactions on Automation Science and Engineering},
  volume={9},
  number={3},
  pages={596--606},
  year={2012},
  publisher={IEEE}
}

@article{jia2013efficient,
  title={Efficient computing budget allocation for finding simplest good designs},
  author={Jia, Qing-Shan and Zhou, Enlu and Chen, Chun-Hung},
  journal={IIE Transactions},
  volume={45},
  number={7},
  pages={736--750},
  year={2013},
  publisher={Taylor \& Francis}
}

@article{hunter2019introduction,
  title={An introduction to multiobjective simulation optimization},
  author={Hunter, Susan R and Applegate, Eric A and Arora, Viplove and Chong, Bryan and Cooper, Kyle and Rinc{\'o}n-Guevara, Oscar and Vivas-Valencia, Carolina},
  journal={ACM Transactions on Modeling and Computer Simulation (TOMACS)},
  volume={29},
  number={1},
  pages={1--36},
  year={2019},
  publisher={ACM New York, NY, USA}
}

@article{lee2012approximate,
  title={Approximate simulation budget allocation for selecting the best design in the presence of stochastic constraints},
  author={Lee, Loo Hay and Pujowidianto, Nugroho A and Li, Ling-Wei and Chen, Chun-Hung and Yap, Chee Meng},
  journal={IEEE Transactions on Automatic Control},
  volume={57},
  number={11},
  pages={2940--2945},
  year={2012},
  publisher={IEEE}
}

@inproceedings{dobkin1190determining,
author={Dobkin, David P and Kirkpatrick, David G},
editor={Paterson, Michael S},
title={Determining the separation of preprocessed polyhedra---A unified approach},
booktitle={Automata, Languages and Programming},
year={1990},
publisher={Springer},
address={Berlin, Heidelberg},
pages={400--413},
}

@article{wang2020distance,
title={The distance between convex sets with {M}inkowski sum structure: application to collision detection},
author={Wang, Xiangfeng and Zhang, Junping and Zhang, Wenxing},
journal={Computational Optimization and Applications},
volume={77},
pages={465--490},
year={2020},
publisher={Springer}
}

@book{efron1994introduction,
  title={An Introduction to the Bootstrap},
  author={Efron, Bradley and Tibshirani, Robert J},
  year={1994},
  publisher={CRC press, Boca Raton, FL}
}

@inproceedings{barba2014optimal,
author = {Luis Barba and Stefan Langerman},
title = {Optimal detection of intersections between convex polyhedra},
booktitle = {Proceedings of the 2015 Annual ACM-SIAM Symposium on Discrete Algorithms (SODA)},
year={2014},
pages = {1641-1654},
}

@book{toth2017handbook,
  title={Handbook of Discrete and Computational Geometry},
  author={Toth, Csaba D and O'Rourke, Joseph and Goodman, Jacob E},
  year={2017},
  publisher={CRC press, Boca Raton, FL}
}

@article{sidak1971rectangle,
  title={On probabilities of rectangles in multivariate {S}tudent distributions: their dependence on correlations},
  author={{\v{S}}id{\'a}k, Zbyn{\v{e}}k},
  journal={The Annals of Mathematical Statistics},
  volume={42},
  number={1},
  pages={169--175},
  year={1971},
  publisher={Institute of Mathematical Statistics}
}

@article{kung1975finding,
  title={On finding the maxima of a set of vectors},
  author={Kung, H T and Luccio, F and Preparata, F P},
  journal={Journal of the ACM (JACM)},
  volume={22},
  number={4},
  pages={469--476},
  year={1975},
  publisher={ACM New York, NY, USA}
}

@article{Owen1988empirical,
    author = {Owen, Art B},
    title = {Empirical likelihood ratio confidence intervals for a single functional},
    journal = {Biometrika},
    volume = {75},
    number = {2},
    pages = {237-249},
    year = {1988},
    publisher={Oxford University Press}
}

@article{patsis1997simd,
  title={{SIMD} parallel discrete-event dynamic system simulation},
   author={Patsis, Nikos T and Chen, Chun-Hung and Larson, Michael E},
  journal={IEEE Transactions on Control Systems Technology},
  volume={5},
  number={1},
  pages={30--41},
  year={1997},
  publisher={IEEE}
}

@article{currie2021practical,
  title = {A Practical Approach to Subset Selection for Multi-objective Optimization via Simulation},
  author = {Currie, Christine S M and Monks, Thomas},
  journal={ACM Transactions on Modeling and Computer Simulation (TOMACS)},
  volume={31},
  number={4},
  year={2021},
  publisher={ACM New York, NY, USA},
 pages={Article 20. 1--15}
}

@phdthesis{wang2019sequential,
  title={Sequential Procedures for the ``Selection'' Problems in Discrete Simulation Optimization},
  author={Wang, Wenyu},
  year={2019},
  school={Purdue University}
}

@phdthesis{frahm2004generalized,
  title={Generalized elliptical distributions: theory and applications},
  author={Frahm, Gabriel},
  year={2004},
  school={Universit{\"a}t zu K{\"o}ln}
}

@phdthesis{ma2018sequential,
  title={Sequential Ranking and Selection Procedures and Sample Complexity},
  author={Ma, Sijia},
  year={2018},
  school={Cornell University}
}

@inproceedings{wang2017sequential,
  title={Sequential probability ratio test for multiple-objective ranking and selection},
  author={Wang, Wenyu and Wan, Hong},
  booktitle={Proceedings of the 2017 Winter Simulation Conference},
  pages={1998--2009},
  year={2017},
editor={Victor W K Chan and Andrea D'Ambrogio and Gregory Zacharewicz and Navonil Mustafee and Gabriel Wainer and Ernest H Page},
address = {Piscataway, New Jersey},
  organization={Institute of Electrical and Electronics Engineers, Inc.},
}

@inproceedings{applegate2017phantom,
  title={Phantom pareto systems for multi-objective ranking and selection},
  author={Applegate, Eric A and Hunter, Susan R and Feldman, Guy and Pasupathy, Raghu},
  booktitle={Proceedings of the 2017 Winter Simulation Conference},
  pages={4576--4577},
  year={2017},
    editor={Victor W.K. Chan and Andrea D'Ambrogio and Gregory Zacharewicz and Navonil Mustafee and Gabriel Wainer and Ernest H. Page},
    address = {Piscataway, New Jersey},
    organization={Institute of Electrical and Electronics Engineers, Inc.},
}

@phdthesis{avci2024simulation,
  title={Simulation Optimization: Cache and Credit for Parallel Ranking \& Selection, and Dice and Slice for High-Dimensional Problems},
  author={Avci, Harun},
  year={2024},
  school={Northwestern University}
}

@article{hunter2016maximizing,
  title={Maximizing quantitative traits in the mating design problem via simulation-based Pareto estimation},
  author={Hunter, Susan R and McClosky, Benjamin},
  journal={IIE Transactions},
  volume={48},
  number={6},
  pages={565--578},
  year={2016},
  publisher={Taylor \& Francis}
}

@misc{li2025efficientbudgetallocationlargescale,
      title={Efficient Budget Allocation for Large-Scale {LLM}-Enabled Virtual Screening}, 
      author={Zaile Li and Weiwei Fan and L. Jeff Hong},
      year={2025},
      eprint={2408.09537},
      archivePrefix={arXiv},
      primaryClass={stat.ML},
      url={https://arxiv.org/abs/2408.09537}, 
}

@inproceedings{Russac2021,
 author = {Russac, Yoan and Katsimerou, Christina and Bohle, Dennis and Capp\'{e}, Olivier and Garivier, Aur\'{e}lien and Koolen, Wouter},
 booktitle = {Advances in Neural Information Processing Systems},
 editor = {M. Ranzato and A. Beygelzimer and Y. Dauphin and P.S. Liang and J. Wortman Vaughan},
 pages = {25100--25110},
 publisher = {Curran Associates, Inc.},
 title = {A/B/n Testing with Control in the Presence of Subpopulations},
 url = {https://proceedings.neurips.cc/paper_files/paper/2021/file/d35a29602005cb55aa57a5f683c8e0c2-Paper.pdf},
 volume = {34},
 year = {2021}
}

@article{holm1979simple,
  title={A simple sequentially rejective multiple test procedure},
  author={Holm, Sture},
  journal={Scandinavian Journal of Statistics},
  pages={65--70},
  volume={6},
  number={2},
  year={1979},
  publisher={JSTOR}
}

@misc{wang2026rankingandselectionmultiplecorrectanswers,
      title={Ranking-and-Selection with Multiple Correct Answers and Non-Answerable Estimates}, 
      author={Qiaoqiao Wang and Wei You},
      year={2026},
      eprint={2606.21889},
      archivePrefix={arXiv},
      primaryClass={cs.LG},
      url={https://arxiv.org/abs/2606.21889}, 
}

\newpage

\noindent{\LARGE \textbf{Electronic Companion to One-Shot Screening of Simulated Systems for Acceptability}}

\setcounter{section}{0}
\renewcommand{\thesection}{EC.\arabic{section}}

\setcounter{page}{1}
\renewcommand{\thepage}{ec\arabic{page}}
This electronic companion contains the following material:
\begin{itemize}
    \item Section~\ref{sec:Illustrating Confidence Regions}: three specific confidence regions for the expected value—ellipsoidal, box-shaped, and half-box-shaped—under the normality assumption; 
    \item Section~\ref{sec:the Envelope CI}: details of OSSA in the adaptive sampling setting;
    \item Section~\ref{sec:single_resp}: results for the single-response case;
    \item Section~\ref{sec: two acceptability with ellipsoid}: results for FOSSA using confidence regions beyond (half-)box shapes, under both stochastically constrained and Pareto optimality definitions;
    \item Section~\ref{sec:proofs}: proofs of the theoretical results;
     
    \item Section~\ref{sec:parallel}: an analysis of how FOSSA can be implemented in a parallel computing environment;
    \item Section~\ref{sec:misc_results}: miscellaneous results, including a confidence region for the responses of acceptable systems, a definition of phantom Pareto systems, and details about the metrics used to measure the quality of individual systems in the numerical experiments.
\end{itemize}

\section{Illustrating Confidence Regions Assuming Normality}
\label{sec:Illustrating Confidence Regions}

In this section, we provide detailed examples of confidence regions with three different shapes, under the assumption that the simulation outputs are normally distributed and independent across systems and the system responses are their expected values.
These assumptions allow us to present specific numerical examples of confidence regions.
\begin{assumption} \label{assump:normal}
For each System $i$, $i=1,2,\ldots,k$, the output vectors $\boldsymbol{X}_{i1},\boldsymbol{X}_{i2},\ldots,\boldsymbol{X}_{in_{i}}$ follow a multivariate normal distribution with unknown mean $\boldsymbol{\mu}_{i}=(\mu_{i1},\mu_{i2},\dots,\mu_{id})^{\intercal}$ and unknown positive definite covariance matrix $\boldsymbol{\Sigma}_{i}$.
\end{assumption}

Assumption \ref{assump:normal} is standard in the  R\&S literature as one can, in certain situations, otherwise deal with non-normally distributed outputs by batching and appealing to the Central Limit Theorem to assert that the batched means are approximately normally distributed.
We make this assumption for concreteness; other output distributions, e.g., Bernoulli, can be handled as long as valid confidence regions for the response vector can be constructed.
The mean vector $\boldsymbol{\mu}_{i}=(\mu_{i1},\mu_{i2},\dots,\mu_{id})^{\intercal}=\mathrm{E}[\boldsymbol{X}_{i1}]$ and covariance matrix $\boldsymbol{\Sigma}_{i}=\mathrm{Var}[\boldsymbol{X}_{i1}]$ are unknown but can be estimated by the sample mean vector $
\widehat{\boldsymbol{\mu}}_{i}
=n_{i}^{-1}\sum_{\ell=1}^{n_{i}}\boldsymbol{X}_{i\ell}
$ and the sample covariance matrix $
\widehat{\boldsymbol{\Sigma}}_{i}
=(n_{i}-1)^{-1}\sum_{\ell=1}^{n_{i}}(\boldsymbol{X}_{i\ell}-\widehat{\boldsymbol{\mu}}_{i})(\boldsymbol{X}_{i\ell}-\widehat{\boldsymbol{\mu}}_{i})^{\intercal}$, respectively, for $i=1,2,\dots, k$.

We will make one other assumption about how simulation outputs are generated.
\begin{assumption} \label{assump:indep}
The outputs obtained from different systems are independent, i.e., $\boldsymbol{X}_{i\ell}$ and  $\boldsymbol{X}_{jq}$ are independent for all $i,j=1,2,\ldots,k$, $i \neq j$, $\ell=1,2,\ldots,n_{i}$ and $q=1,2,\ldots,n_{j}$.
\end{assumption}

We present three specific constructions of $\mathds{C}_{ij}$ for $i, j=1,2,\dots,k$, namely, an ellipsoid, a box, and a half-box in $\mathds{R}^d$ under Assumptions~\ref{assump:normal} and \ref{assump:indep}.
Under Assumption~\ref{assump:indep}, each confidence region is constructed to achieve a confidence level of $(1-\alpha)^{1/k}$.
Table~\ref{table: confidence region} summarizes these confidence regions along with the corresponding critical values selected to ensure this per-system confidence level.

\begin{table}[tb]
\caption{A summary of the ellipsoid, box, and half-box confidence regions in different dimensions. The critical values are functions of the confidence level ($1-\alpha$), the number of systems ($k$), the number of responses ($d$) and the sample size for System $j$ ($n_{j}$).}
\label{table: confidence region}
\begin{center}
\begin{tabular}{c | c | c  | c | c} 
 Responses & Guarantee & Confidence Region & Critical Value &  Exact\\
  \hline
  \multirow{3}{*}{$d=2$} & System-wise & Half-Box ($\mathds{C}_{ii}^\mathrm{hb}$ and $\mathds{C}_{ij}^\mathrm{hb}$) & $\Lambda^\mathrm{hb}=t_{(1-\alpha)^{1/2k},n_{j}-1}$ & No \\
  \cline{2-5}
  & \multirow{2}{*}{Set-wise} & Box ($\mathds{C}_{uj}^\mathrm{box}$) & $\Lambda^\mathrm{box}=t_{1-\frac{1-(1-\alpha)^{1/2k}}{2},n_{j}-1}$ & No \\
  & & Ellipsoid ($\mathds{C}_{uj}^\mathrm{ell}$) & $\Lambda^\mathrm{ell}=F_{(1-\alpha)^{1/k},d,n_{j}-d}$ & Yes \\
  \hline
  \multirow{3}{*}{$d \geq 3$} & System-wise & Half-Box ($\mathds{C}_{ii}^\mathrm{hb}$ and $\mathds{C}_{ij}^\mathrm{hb}$) & $\Lambda^\mathrm{hb}=t_{1-\frac{1-(1-\alpha)^{1/k}}{d},n_{j}-1}$ & No \\
  \cline{2-5}
  & \multirow{2}{*}{Set-wise} & Box ($\mathds{C}_{uj}^\mathrm{box}$) & $\Lambda^\mathrm{box}=t_{1-\frac{1-(1-\alpha)^{1/k}}{2d},n_{j}-1}$ & No \\
  & & Ellipsoid ($\mathds{C}_{uj}^\mathrm{ell}$) & $\Lambda^\mathrm{ell}=F_{(1-\alpha)^{1/k},d,n_{j}-d}$ & Yes \\
\end{tabular}
\end{center}
\end{table}

We first consider the ellipsoid
$$\mathds{C}_{uj}^\mathrm{ell}=\left\{\boldsymbol{m}\in\mathds{R}^{d}\colon n_{j}(\widehat{\boldsymbol{\mu}}_{j}-\boldsymbol{m})^{\intercal}\widehat{\boldsymbol{\Sigma}}_{j}^{-1}(\widehat{\boldsymbol{\mu}}_{j}-\boldsymbol{m})\leq \frac{d(n_{j}-1)}{n_{j}-d}\Lambda^\mathrm{ell}  \right\},$$
where $\Lambda^\mathrm{ell}= F_{(1-\alpha)^{1/k},d,n_{j}-d}$, the $(1-\alpha)^{1/k}$ quantile of the $F$ distribution with $d$ numerator degrees of freedom and $n_{j}-d$ denominator degrees of freedom.
This ellipsoid is centered at $\widehat{\boldsymbol{\mu}}_{j}$ and its orientation is characterized by $\widehat{\boldsymbol{\Sigma}}_{j}^{-1}$. (To ensure the existence of $\widehat{\boldsymbol{\Sigma}}_{j}^{-1}$, it is necessary that $n_j \geq d+1$.)
Under Assumptions~\ref{assump:normal} and \ref{assump:indep}, $\mathds{C}_{uj}^\mathrm{ell}$ is a $(1-\alpha)^{1/k}$ confidence region for $\boldsymbol{\mu}_{j}$ \citep{anderson:mvstats}.

Another commonly used confidence region is the hyperrectangle
$$\mathds{C}_{uj}^\mathrm{box}=\left\{\boldsymbol{m}\in \mathds{R}^{d} \colon \widehat{\mu}_{jr}-\Lambda^\mathrm{box}\sqrt{\widehat{\Sigma}_{j}(r,r)/n_{j}} \leq 
m_{r}
\leq \widehat{\mu}_{jr}+\Lambda^\mathrm{box}\sqrt{\widehat{\Sigma}_{j}(r,r)/n_{j}} \ \mathrm{for}\  r=1,2\dots,d \right\},$$ 
where $\Lambda^\mathrm{box}=t_{\beta^\mathrm{box},n_{j}-1}$ is the $\beta^\mathrm{box}$ quantile of a Student's $t$-distribution with $n_{j}-1$ degrees of freedom.
When $d=2$, Corollary 2 in \cite{sidak1971rectangle} offers a way to split the confidence multiplicatively, resulting in $\beta^\mathrm{box}=1-\frac{1-(1-\alpha)^{1/2k}}{2}$. More generally, when $d \geq 3$, we use a Bonferroni correction with $\beta^\mathrm{box}=1-\frac{1-(1-\alpha)^{1/k}}{2d}$. 
$\mathds{C}_{uj}^\mathrm{box}$ is centered at $\boldsymbol{\widehat{\mu}}_{j}$ and formed by taking the Cartesian product of $d$ two-sided confidence intervals, one for each component of $\boldsymbol{\mu}_{j}$.

Both $\mathds{C}_{uj}^\mathrm{ell}$ and $\mathds{C}_{uj}^\mathrm{box}$ are better suited to deliver the stronger set-wise PAS guarantee.
For the system-wise PAS guarantee, we turn to the confidence regions
$$\mathds{C}_{ii}^\mathrm{hb}=\left\{\boldsymbol{m}\in \mathds{R}^{d} \colon m_{r}\geq \widehat{\mu}_{ir}-\Lambda^\mathrm{hb}\sqrt{\widehat{\Sigma}_{i}(r,r)/n_{i}} \ \mathrm{for}\  r=1,2\dots,d \right\},$$ 
and 
$$\mathds{C}_{ij}^\mathrm{hb}=\left\{\boldsymbol{m}\in \mathds{R}^{d} \colon m_{r}\leq \widehat{\mu}_{jr}+\Lambda^\mathrm{hb}\sqrt{\widehat{\Sigma}_{j}(r,r)/n_{j}} \ \mathrm{for}\  r=1,2\dots,d  \right\} \text{ for all } j\neq i,$$
where $\Lambda^\mathrm{hb}=t_{\beta^\mathrm{hb},n_{j}-1}$.
When $d=2$, $\beta^\mathrm{hb}=(1-\alpha)^{1/2k}$, and when $d\geq3$, $\beta^\mathrm{hb}=1-\frac{1-(1-\alpha)^{1/k}}{d}$.

\section{OSSA with Adaptive Sampling}
\label{sec:the Envelope CI}

In this section, we detail how OSSA procedures can be extended to the adaptive sampling setting by considering a sequence of confidence regions for the true configuration $\truemean$ that hold simultaneously over time.
To make this assertion more concrete, we suppose that an adaptive sampling procedure proceeds in discrete stages, indexed by $t\in\mathds{N}$, where in each stage a single additional observation is obtained from some selected system.

\begin{assumption}
\label{assump:conf_envelope}
For any user specified $\alpha \in (0, 1)$, for each System $i=1,2,\dots,k$, there exists a sequence of confidence regions for $\truemean$, denoted by $\{\mathds{C}_{i}(t)\}_{t\in\mathds{N}}$ such that 
 $$\mathrm{P}(\truemean\in\mathds{C}_{i}(t) \text{ for all } t \in \mathds{N})\geq1-\alpha.$$
\end{assumption}

In the adaptive sampling setting, we consider the repeated use of an OSSA procedure at each stage, using the available data, to produce a sequence of subsets over time, denoted by $\{\mathcal{S}(t)\}_{t \in \mathbb{N}}$.
If the OSSA procedure uses a sequence of always-valid confidence regions satisfying Assumption~\ref{assump:conf_envelope} in the same way it would for the one-shot setting, i.e., for the intersection-detection subproblem, then it can be shown that the procedure's PAS guarantees are valid at all times.
\begin{theorem} 
\label{theorem5}
For all $t\in\mathds{N}$, the subset $\mathcal{S}^{system}(t)=\{i\colon \mathds{A}_{i} \cap \mathds{C}_{i}(t) \neq \emptyset \}$ delivers the system-wise PAS guarantee.
\end{theorem}
When a common sequence of confidence regions is used for screening all systems, i.e., $\mathds{C}_{u}(t)=\mathds{C}_{1}(t)=\mathds{C}_{2}(t)=\dots=\mathds{C}_{k}(t)$ for all $t\in\mathds{N}$, the resulting subsets each achieve the set-wise PAS guarantee.
\begin{theorem}
\label{theorem6}
For all $t\in\mathds{N}$, the subset $\mathcal{S}^{set}(t)=\{i\colon \mathds{A}_{i} \cap \mathds{C}_{u}(t) \neq \emptyset \}$ delivers the set-wise PAS guarantee.
\end{theorem}

We draw upon the idea of confidence \textit{envelopes} utilized in the Envelope Procedure of \cite{ma2018sequential} to exemplify a sequence of always-valid confidence intervals for each system's true response when simulation outputs are normally distributed and the response is the expected output. These sequences of confidence intervals for each system can be combined to produce a sequence of confidence regions (a confidence envelope) for $\truemean$ satisfying Assumption~\ref{assump:conf_envelope}.
For clarity and ease of application, the construction of these so-called confidence \textit{envelopes} is reorganized here to provide explicit confidence regions.
For notational convenience, we assume $d=1$ throughout this section.
For $d \geq 2$, confidence envelopes can be constructed separately for each component of the response and then combined into a higher-dimensional confidence envelope by taking their Cartesian product.

We first study the problem of constructing a confidence envelope for the expected output $\mu_i$ (formerly $\theta_i$) of a given System $i$.
We require that a maximum sample size $N_i$ for System $i$ is specified for all $i = 1, 2, \ldots, k$ and that some number of initial observations are obtained to estimate the variance of System $i$.
Let $n_{0i}\geq2$ denote the number of initial observations and let $\hat{\mu}_{0i}=(n_{0i})^{-1}\sum_{\ell=1}^{n_{0i}}X_{i\ell}$ and $S_i^2=(n_{0i}-1)^{-1}\sum_{\ell=1}^{n_{0i}}(X_{i\ell}-\hat{\mu}_{0i})^2$ denote the sample mean and sample variance, respectively, of these observations.
We denote the sample size of System $i$ at stage $t$ as $n_{i}(t)$ and the corresponding sample mean by $\hat{\mu}_{i}(t)=(n_{i}(t))^{-1}\sum_{\ell=1}^{n_{i}(t)}X_{i\ell}$.
Consider a sequence of lower bounds on $\mu_{i}$ of the form 
$$L_i(t)=\hat{\mu}_{i}(t)-\frac{\eta_i S_i}{\sqrt{n_i(t)}} \text{ for all } t\in\mathds{N} \text{ such that } n_{i}(t)\leq N_i,$$
where $\eta$ is a critical value that depends on the initial sample size $n_{0i}$, the desired confidence level $1-\alpha$ and the maximum sample size $N_i$.
When $\eta_i$ is chosen to satisfy, $$\eta_i^2\geq8(n_0-1)\left[\left(\frac{\mathrm{log}_2N}{2^{n_{0i}/2-1}\alpha}\right)^{2/n_{0i}}-\frac{1}{2}\right],$$
the sequence $\{L_i(t)\colon t \in \mathds{N}\}$ forms a lower bound for a $1-\alpha/2$ confidence envelope for $\mu_i$, i.e.,
$$\mathrm{P}(L_i(t) \leq \mu_i \text{ for all } t \in \mathds{N} \text{ such that } n_i(t) \leq N_i) \geq 1-\alpha/2.$$ 
Likewise, the sequence of upper bounds of the form
$$U_i(t)=\hat{\mu}_{i}(t)+\frac{\eta_i S_i}{\sqrt{n_i(t)}} \text{ for all } t\in\mathds{N} \text{ such that } n_{i}(t)\leq N_i$$
forms an upper bound for a $1-\alpha/2$ confidence envelope for $\mu_i$.
These results can be combined to give a single $1-\alpha$ confidence envelope for $\mu_i$, thereby fulfilling Assumption~\ref{assump:conf_envelope}.
\begin{theorem}
$$\mathrm{P}\left(L_i(t) \leq \mu_i\leq U_i(t) \text{ for all } t\in\mathds{N} \text{ such that } n_i(t)\leq N_i \right) \geq 1-\alpha.$$
\end{theorem}

\section{Optimization of a Single Response}
\label{sec:single_resp}

In this section, we illustrate the main aspects of the \ourfrmwrk\ framework on the well-studied setting where the simulation model has a single response of interest, i.e., $d=1$.
Meanwhile, the response of each system is the mean of normally distributed outputs.
We represent the scalar response of System $i$ by $\mu_{i}$ for $i=1,2,\dots,k$.
The conventional definition of acceptability is optimality, i.e., having the best (smallest) response.
The set of acceptable systems is $\mathcal{A}(\truemean)=\{i \colon {\mu}_{i}\leq\mu_{j} \ \textrm{for all} \ j\neq i\}$, where $\mathbf{M}_0=(\mu_{1},\mu_{2},\dots,\mu_{k})$.
The acceptable region for a given System $i$ is the polyhedron $\mathds{A}_{i}=\{\mathbf{M}\in\mathds{R}^{k}\colon {m}_{i}\leq {m}_{j} \mathrm{\ for\ all\ }j\neq i\}$, which satisfies Assumption \ref{assump: Acceptability} where $\mathds{F}=\mathds{R}$ and $\prec$ is equivalent to $<$. 
Other definitions of acceptability in this setting include the following: having an response that exceeds the best response by no more than a given tolerance; that is less than or equal to a given standard; whose difference relative to a given standard is less than a given tolerance; that is within the top $m$ responses; and that exceeds some fraction of the best response. \ourfrmwrk/\oursimplerfrmwk\ procedures can be devised for all of these cases.

An example of an existing screening procedure designed to deliver the system-wise PAS guarantee for this setting and definition of acceptability as optimality is the Screen-to-the-Best (STTB) procedure of \cite{nelson2001simple}. The STTB procedure requires a common sample size across all systems but has been extended in \cite{boesel2003using} to handle unequal sample sizes. The extended version returns the subset 
\begin{equation*}
\mathcal{S}^{\mathrm{STTB}}=\left\{i\colon \widehat{\mu}_{i}\leq \widehat{\mu}_{j}+\sqrt{t_{(1-\alpha)^{{1}/{(k-1)}},n_{i}-1}^{2}{\widehat{\sigma}^{2}_{i}}/{n_{i}} + t_{(1-\alpha)^{{1}/{(k-1)}},n_{j}-1}^{2}{\widehat{\sigma}^{2}_{j}}/{n_{j}}} \text{\ for \ all} \ j \neq i \right\},    
\end{equation*}
where $\widehat{\sigma}^{2}_{i}$ and $\widehat{\sigma}^{2}_{j}$ are the sample variances of the outputs from System $i$ and System $j$, respectively.
The STTB procedure controls the differences between the response of a given system being screened and those of all the other systems. 
This necessitates making pairwise comparisons, leading to a time complexity of $O(k^{2})$ for constructing the returned subset $\mathcal{S}^{\mathrm{STTB}}$.
By defining the confidence region for screening System $i$ as
$$
\mathds{C}_{i}^{\mathrm{STTB}}=\left\{\mathbf{M}\in\mathds{R}^{k}\colon m_{j}-m_{i}
\leq
\widehat{\mu}_{j}-\widehat{\mu}_{i}
+
\sqrt{t_{(1-\alpha)^{{1}/{(k-1)}},n_{i}-1}^{2}{\widehat{\sigma}^{2}_{i}}/{n_{i}} + t_{(1-\alpha)^{{1}/{(k-1)}},n_{j}-1}^{2}{\widehat{\sigma}^{2}_{j}}/{n_{j}}} \text{\ for\ all\ }j\neq i \right\},
$$
it can seen that the STTB procedure fits within the \ourfrmwrk\ framework.

We next illustrate two FOSSA procedures for this setting, one for the system-wise PAS guarantee and one for the set-wise PAS guarantee.
When the system-wise PAS guarantee is desired, the confidence region $\mathds{C}^\mathrm{os}_{i}=\bigtimes\limits_{j=1}^{k}\mathds{C}^\mathrm{os}_{ij}$ is used where $$\mathds{C}^\mathrm{os}_{ii}=\left\{m\in \mathds{R}\colon m\geq\widehat{\mu}_{i}-\Lambda^\mathrm{os}\sqrt{\widehat{\sigma}^{2}_{i}/n_{i}}\right\}$$ 
and 
$$\mathds{C}^\mathrm{os}_{ij}=\left\{m\in \mathds{R}\colon m\leq\widehat{\mu}_{j}+\Lambda^\mathrm{os}\sqrt{\widehat{\sigma}^{2}_{j}/n_{j}}\right\} \text{ for all } j\neq i$$
and $\Lambda^\mathrm{os}=t_{(1-\alpha)^{{1}/{k}},n_{j}-1} $.
$\mathds{C}^\mathrm{os}_{ii}$ and $\mathds{C}^\mathrm{os}_{ij}$ are one-sided confidence intervals for $\mu_i$ and $\mu_j$, $j \neq i$, respectively, with confidence level $(1-\alpha)^{1/k}$.
The procedure returns those systems for which $\mathds{A}_{i}\cap\mathbb{M}^{*,\mathrm{os}}_{i}\neq\emptyset$, where $\mathbb{M}^{*,\mathrm{os}}_{i}$ is defined according to Proposition~\ref{Proposition: Best and Worst }.
The set $\mathbb{M}^{*,\mathrm{os}}_{i}$ consists of a single point,
$\mathbf{M}_{i}^{*}=(m^{*}_{i1},m^{*}_{i2},\dots,m^{*}_{ik})$, where
$m^{*}_{ii}=\widehat{\mu}_{i}-\Lambda^\mathrm{os}\sqrt{\widehat{\sigma}^{2}_{i}/n_{i}}$ and $m^{*}_{ij}=\widehat{\mu}_{j}+\Lambda^\mathrm{os}\sqrt{\widehat{\sigma}^{2}_{j}/n_{j}}$ for all $j\neq i$.
The system-wise FOSSA procedure returns the subset
\begin{align*}
\mathcal{S}^{\mathrm{system}}
&=\{i\colon \mathbf{M}^{*}_{i}\in\mathds{A}_{i}\}
=\{i\colon{m}_{ii}^{*}\leq{m}_{ij}^{*}\ \mathrm{for\ all\ } j\neq i\}\\
&=\left\{i\colon \widehat{\mu}_{i} -\Lambda^\mathrm{os}_{i} \sqrt{{\widehat{\sigma}^{2}_{i}}/{n_{i}} }\leq \widehat{\mu}_{j}+\Lambda^\mathrm{os}_{j}\sqrt{{\widehat{\sigma}^{2}_{j}}/{n_{j}}}  \ {\rm for\ all} \ j\neq i\right\}.
\end{align*}
\begin{proposition} \label{proposition1}
 $\mathcal{S}^{\mathrm{system}}$ delivers the system-wise PAS guarantee with confidence level $1-\alpha$.
\end{proposition}
The system-wise FOSSA procedure is a special case of the Decoupled STTB procedure introduced in \cite{zhao23screening} with the tolerance $\delta$ set as 0.

When the set-wise PAS guarantee is desired, the common confidence region $\mathds{C}^\mathrm{ts}_{u} = \bigtimes\limits_{j=1}^{k}\mathds{C}^\mathrm{ts}_{uj}$ is used where $$
\mathds{C}^\mathrm{ts}_{uj}=
\left\{m\in \mathds{R}\colon\widehat{\mu}_{j}-\Lambda^\mathrm{ts}\sqrt{{\widehat{\sigma}^{2}_{j}}/{n_{j}}}
\leq m \leq
\widehat{\mu}_{j}+\Lambda^\mathrm{ts}\sqrt{{\widehat{\sigma}^{2}_{j}}/{n_{j}}}\right\} \text{ for } j = 1, 2, \ldots, k,$$
and $\Lambda^\mathrm{ts}=t_{\frac{1+(1-\alpha)^{{1}/{k}}}{2},n_{j}-1}$. 
Each $\mathds{C}^\mathrm{ts}_{uj}$ is a two-sided confidence interval for $\mu_j$ with confidence level $(1-\alpha)^{1/k}$. Determining the non-emptiness of $\mathds{A}_{i} \cap \mathbb{M}_i^{*,\mathrm{ts}}$ again involves checking if a single configuration $\mathbf{M}_{i}^{*}=(m^{*}_{i1},m^{*}_{i2},\dots,m^{*}_{ik})$ is in $\mathds{A}_{i}$, where
$m^{*}_{ii}=\widehat{\mu}_{i}-\Lambda^\mathrm{ts}\sqrt{\widehat{\sigma}^{2}_{i}/n_{i}}$ and $m^{*}_{ij}=\widehat{\mu}_{j}+\Lambda^\mathrm{ts}\sqrt{\widehat{\sigma}^{2}_{j}/n_{j}}$ for all $j\neq i$.
The set-wise FOSSA procedure returns the subset 
\begin{align*}
\mathcal{S}^{\mathrm{set}}
&=\{i\colon \mathbf{M}^{*}_{i}\in\mathds{A}_{i}\}
=\{i\colon{m}_{ii}^{*}\leq{m}_{ij}^{*}\ \mathrm{for\ all\ } j\neq i\}\\
&=\left\{i\colon \widehat{\mu}_{i} -\Lambda^\mathrm{ts}_{i} \sqrt{{\widehat{\sigma}^{2}_{i}}/{n_{i}} }\leq \widehat{\mu}_{j}+\Lambda^\mathrm{ts}_{j}\sqrt{{\widehat{\sigma}^{2}_{j}}/{n_{j}}} \ {\rm for\ all} \ j\neq i\right\}.
\end{align*}
\begin{proposition} \label{proposition2}
$\mathcal{S}^{\mathrm{set}}$ delivers the set-wise PAS guarantee with confidence level $1-\alpha$.
\end{proposition}

Algorithm \ref{Algorithm1: Single Response Procedure} shows the steps for constructing $\mathcal{S}^{\mathrm{system}}$ and $\mathcal{S}^{\mathrm{set}}$.
By setting $\Lambda$ equal to $\Lambda^\mathrm{os}$ or $\Lambda^\mathrm{ts}$, the returned subset $\mathcal{S}$ of Algorithm \ref{Algorithm1: Single Response Procedure} corresponds to $\mathcal{S}^{\mathrm{system}}$ or $\mathcal{S}^{\mathrm{set}}$, respectively.
From the serial ``for'' loops in Algorithm \ref{Algorithm1: Single Response Procedure}, it is evident that constructing $\mathcal{S}^{\mathrm{system}}$ or $\mathcal{S}^{\mathrm{set}}$ has time complexity $O(k)$.

\begin{algorithm}[tb]
\caption{FOSSA for optimizing a single response}
\label{Algorithm1: Single Response Procedure}
\begin{algorithmic}[1]
\Require $\widehat{\mu}_{i},\widehat{\sigma}^{2}_{i},n_{i}$ for $i=1,2,\dots,k$ and $\Lambda$
\State $\mathcal{S}\gets\{1,2,\dots,k\}$
\For{$i=1,2,\dots,k$}
    \State $l_i\gets\widehat{\mu}_{i}-\Lambda\sqrt{\widehat{\sigma}^{2}_{i}/n_{i}}$
    \State $u_i\gets\widehat{\mu}_{i}+\Lambda\sqrt{\widehat{\sigma}^{2}_{i}/n_{i}}$
\EndFor
\State $\tau\gets\min_{i\in\{1,2,\dots,k\}}u_i$
\For{$i=1,2,\dots,k$}
    \If{$l_i>\tau$}
        \State $\mathcal{S}\gets\mathcal{S}\setminus\{i\}$
    \EndIf
\EndFor
\State \Return $\mathcal{S}$
\end{algorithmic}
\end{algorithm}

\begin{figure}
    \centering
    \includegraphics[width=14cm]{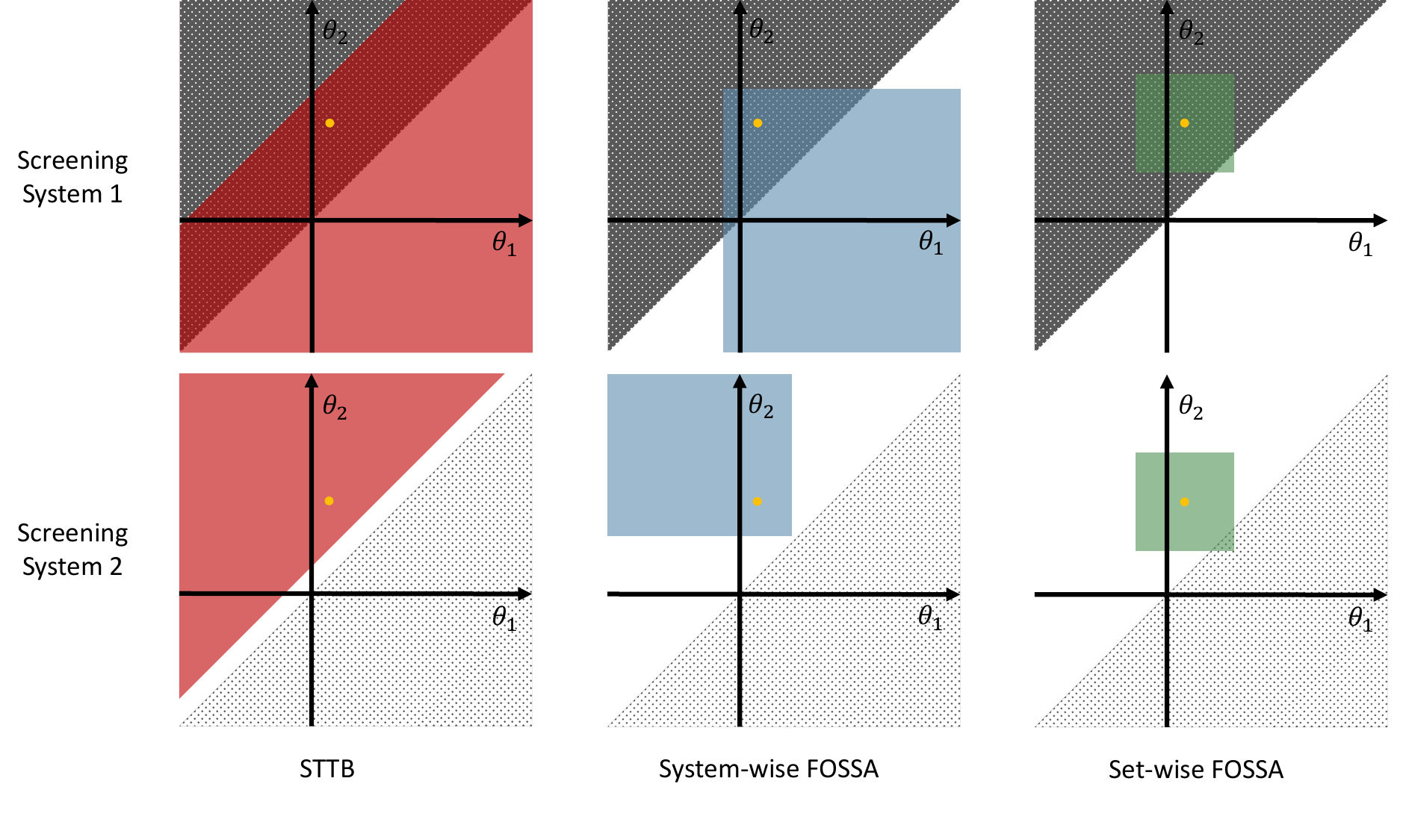}
    
    \caption{Depiction of STTB, system-wise FOSSA and set-wise FOSSA procedures for a problem with $k=2$ systems and $d=1$ response. The yellow dot indicates the sample mean vector $\widehat{\mathbf{M}}=(\widehat{\mu}_{1},\widehat{\mu}_{2})$. For each procedure, the colored areas represent the confidence regions for $\mathbf{M}_0=(\theta_{1},\theta_{2})$, i.e., $\mathds{C}_1$ and $\mathds{C}_2$, the gray region with white dots is $\mathds{A}_{1}$ and the white region with gray dots is $\mathds{A}_{2}$.
    In the case of screening System 1, the top-left corners of the blue and green regions are $\mathbb{M}^{*}_{1}$; whereas in the case of screening System 2, the bottom-right corners of these regions are $\mathbb{M}^{*}_{2}$.}
    \label{fig:Fig1}
\end{figure}

Figure \ref{fig:Fig1} shows the confidence regions used by STTB and the two FOSSA procedures. 
It can be seen from the figure that STTB will be more powerful at screening systems than system-wise FOSSA because of how the geometry of $\mathds{C}^{\mathrm{STTB}}_{i}$ better aligns with that of $\mathds{A}_{i}$.
This observation is a consequence of a nested relationship between $\mathcal{S}^{\mathrm{STTB}}$ and $\mathcal{S}^{\mathrm{system}}$, as described in Proposition \ref{Prop:STTB subset system}.
\begin{proposition} \label{Prop:STTB subset system}
For any fixed data, $\widehat{\mu}_{i}$ and $\widehat{\sigma}^{2}_{i}$ for $i=1,2,\dots,k$, and confidence level $1-\alpha$,  $\mathcal{S}^{\mathrm{STTB}} \subseteq \mathcal{S}^{\mathrm{system}}$.
\end{proposition}
\noindent {\bf Proof of Proposition ~\ref{Prop:STTB subset system}:}

Fix any $1-\alpha> (1/2)^{k-1}$ , $k\geq2$, and data $\widehat{\mu}_{i}$ and $\widehat{\sigma}^{2}_{i}$ for $i=1,2,\dots,k$.
For any $i$ and $j$, $i \neq j$,
\begin{align*}
   \sqrt{t_{(1-\alpha)^{{1}/{(k-1)}},n_{i}-1}^{2}{\widehat{\sigma}^{2}_{i}}/{n_{i}} + t_{(1-\alpha)^{{1}/{(k-1)}},n_{j}-1}^{2}{\widehat{\sigma}^{2}_{j}}/{n_{j}}}
   &\leq t_{(1-\alpha)^{{1}/{(k-1)}},n_{i}-1}\sqrt{{\widehat{\sigma}^{2}_{i}}/{n_{i}}}+t_{(1-\alpha)^{{1}/{(k-1)}},n_{j}-1}\sqrt{{\widehat{\sigma}^{2}_{j}}/{n_{j}}}\\
   &\leq t_{(1-\alpha)^{{1}/{k}},n_{i}-1}\sqrt{{\widehat{\sigma}^{2}_{i}}/{n_{i}}}+t_{(1-\alpha)^{{1}/{k}},n_{j}-1}\sqrt{{\widehat{\sigma}^{2}_{j}}/{n_{j}}}\\ &=\Lambda^{os}_{i}\sqrt{{\widehat{\sigma}^{2}_{i}}/{n_{i}}}+\Lambda^{os}_{j}\sqrt{{\widehat{\sigma}^{2}_{j}}/{n_{j}}}.
\end{align*}
Thus the event
$$\left\{\widehat{\mu}_{i}\leq \widehat{\mu}_{j}+ \sqrt{t_{(1-\alpha)^{{1}/{(k-1)}},n_{i}-1}^{2}{\widehat{\sigma}^{2}_{i}}/{n_{i}} + t_{(1-\alpha)^{{1}/{(k-1)}},n_{j}-1}^{2}{\widehat{\sigma}^{2}_{j}}/{n_{j}}} \text{ for all } j \neq i \right\}$$ implies that $$\left\{\widehat{\mu}_{i} -\Lambda^\mathrm{os}_{i} \sqrt{{\widehat{\sigma}^{2}_{i}}/{n_{i}} }\leq \widehat{\mu}_{j}+\Lambda^\mathrm{os}_{j}\sqrt{{\widehat{\sigma}^{2}_{j}}/{n_{j}}} \text{ for all } j \neq i\right\}.$$
Therefore, if $i\in\mathcal{S}^\mathrm{STTB}$, then $i\in\mathcal{S}^\mathrm{system}$.
Because the choice of $i$ is arbitrary, we have shown that $\mathcal{S}^{\mathrm{STTB}} \subseteq \mathcal{S}^{\mathrm{system}}$.
$\square$

We previously remarked that the time complexity of obtaining $\mathcal{S}^{\mathrm{STTB}}$ is $O(k^{2})$ and that of obtaining $\mathcal{S}^{\mathrm{system}}$ is $O(k)$. 
Coupled with Proposition~\ref{Prop:STTB subset system}, this observation suggests that system-wise FOSSA procedure could be useful for \emph{pre-screening}, to accelerate the construction of $\mathcal{S}^{\mathrm{STTB}}$ without compromising any screening power.
Our use of the term pre-screening refers to the application of a screening procedure to a set of systems and their respective outputs obtained from prior sampling followed by applying another screening procedure on only those systems that survived the initial screening, without taking additional replications.
\cite{zhao23screening} provide experimental timing results showing the pre-screening potential of system-wise FOSSA along with a version of Proposition \ref{Prop:STTB subset system} (without proof). That tutorial also contains numerical experiments comparing STTB, system-wise FOSSA, and set-wise FOSSA.

\section{FOSSA with a General Confidence Region}
\label{sec: two acceptability with ellipsoid}
In this section, we use $\mathds{C}_u^{\mathrm{ell}} = \bigtimes\limits_{j=1}^{k} \mathds{C}_{uj}^{\mathrm{ell}}$ as an example to demonstrate how FOSSA with multiple responses works when the confidence regions have a general shape rather than a box or half-box.
\subsection{Optimization with Stochastic Constraints}
When $\mathds{C}_{ii}$ has a general shape other than a (half-)box, determining $\tau_{ii}$, i.e., checking whether $\mathds{C}_{ii} \cap \mathds{F} = \emptyset$ and solving $\min m_1$ s.t.\ $\boldsymbol{m} \in \mathds{C}_{ii} \cap \mathds{F}$ can be challenging.
Similar difficulties arise when determining $\tau_{ij}$.
However, when $\mathds{C}_{ij}$ is an ellipsoid for all $i, j = 1, 2, \dots, k$, checking whether the intersection is empty can be determined analytically, and the associated optimization problems reduce to well-structured convex optimization problems.
We will use $\mathds{C}_{u}^\mathrm{ell}=\bigtimes\limits_{j=1}^{k}\mathds{C}_{uj}^\mathrm{ell}$ as an example to demonstrate how the geometry of an ellipsoid and the acceptable regions can be exploited to solve the feasibility and optimization problems that arise.

A FOSSA procedure using the confidence region  $\mathds{C}_{u}^\mathrm{ell}=\bigtimes\limits_{j=1}^{k}\mathds{C}_{uj}^\mathrm{ell}$ will deliver the set-wise PAS guarantee. We call this procedure FOSSA Ellipsoid. In this case, one can cheaply and exactly determine the non-emptiness of $\mathds{C}_{ui}^\mathrm{ell}\cap\mathds{F}^{c}$ in (\ref{optimization problem 2_Constrained Opt}) for $i=1,2,\dots,k$ by making use of the projection of the ellipsoid on each dimension. 
\begin{proposition}
\label{Proposition: ellipsoid projection}
$\mathds{C}_{ui}^\mathrm{ell}\cap\mathds{F}^{c}=\emptyset$ if and only if $\widehat{{\mu}}_{ir} +\sqrt{\frac{d(n_{i}-1)}{n_{i}(n_{i}-d)} \Lambda^\mathrm{ell}\widehat{\Sigma}_{i}(r,r)}\leq\mu^{\dagger}_{r}$ for all $r=2,3,\dots,d$.
\end{proposition}
As for (\ref{optimization problem 1_Constrained Opt}), the non-emptiness of $\mathds{C}_{ui}^\mathrm{ell}\cap\mathds{F}$ for any $i=1,2,\dots,k$ can be determined by solving a linearly constrained quadratic program.
More specifically, $\mathds{C}_{ui}^\mathrm{ell}\cap\mathds{F}=\emptyset$ if and only if $\min\limits_{\boldsymbol{m}\in\mathds{F}}\ n_{i}(\widehat{\boldsymbol{\mu}}_{i}-\boldsymbol{m})^{\intercal}\widehat{\boldsymbol{\Sigma}}_{i}^{-1}(\widehat{\boldsymbol{\mu}}_{i}-\boldsymbol{m})>\frac{d(n_{i}-1)}{n_{i}-d} \Lambda^\mathrm{ell}$.
For the case where $\mathds{C}_{ui}^\mathrm{ell}\cap\mathds{F}\neq\emptyset$, the optimization problem defining $\tau_{ii}$ is a quadratically constrained linear program, and therefore amounts to solving a second-order conic program assuming $\widehat{\boldsymbol{\Sigma}}_{i}^{-1}$ is positive definite. Both quadratic programs can be solved by the interior-point method with iteration complexity of $O(\sqrt{d}\log (1/\epsilon))$, i.e., it takes at most $O(\sqrt{d}\log(1/\epsilon))$ iterations to find an $\epsilon$-precise solution. 
Meanwhile, the time complexity of each iteration of the interior-point method is $O(d^3)$. The two optimization problems in (\ref{optimization problem 1_Constrained Opt}) can alternatively be solved exactly by exhaustively enumerating all possible active sets, with time complexity $O(2^{d}d^{3})$. This approach does not require any numerical optimization algorithms and can be more efficient for instances with few responses.

The \oursimplerfrmwk\ Ellipsoid procedure for constrained optimality is summarized in Algorithm~\ref{Constrained-Optimal Procedure(Ellipsoid)}. It is evident from the serial ``for'' loops that the procedure's time complexity in terms of the number of systems is $O(k)$.

\begin{algorithm}[tb]
\caption{FOSSA Ellipsoid for Constrained Optimality}
\label{Constrained-Optimal Procedure(Ellipsoid)}
\begin{algorithmic}[1]
\Require $\widehat{\boldsymbol{\mu}}_{i},\widehat{\boldsymbol{\Sigma}}_{i},n_i$ for $i=1,2,\dots,k$
\State $\mathcal{S}\gets\{1,2,\dots,k\}$
\For{$i=1,2,\dots,k$}
    \If{$\mathds{C}_{ui}^{\mathrm{ell}}\cap\mathds{F}\neq\emptyset$}
        \State $l_i\gets\min m_{i1}\ \mathrm{s.t.}\ \boldsymbol{m}_{i}\in\mathds{C}_{ui}^{\mathrm{ell}}\cap\mathds{F}$
    \Else
        \State $l_i\gets\infty$
    \EndIf
    \If{$\mathds{C}_{ui}^{\mathrm{ell}}\cap\mathds{F}^{c}=\emptyset$}
        \State $u_i\gets\widehat{\mu}_{ir}+\sqrt{\frac{d(n_i-1)}{n_i(n_i-d)}\Lambda^\mathrm{ell}\widehat{\Sigma}_{i}(r,r)}$
    \Else
        \State $u_i\gets\infty$
    \EndIf
\EndFor
\State $\tau\gets\min_{i\in\{1,2,\dots,k\}}u_i$
\For{$i=1,2,\dots,k$}
    \If{$l_i>\tau$}
        \State $\mathcal{S}\gets\mathcal{S}\setminus\{i\}$
    \EndIf
\EndFor
\State \Return $\mathcal{S}$
\end{algorithmic}
\end{algorithm}

\subsection{Pareto Optimality}
\label{subsec:pareto_ellipsoid}

In this subsection, we present an analogous FOSSA Ellipsoid procedure for Pareto optimality where the responses are expected values and the outputs are normally distributed. 
For the $\mathds{C}_{u}^\mathrm{ell}$ confidence region, the corresponding $\mathbb{M}^{*}_{i}$ is a non-convex region, suggesting that checking the non-emptiness of $\mathds{A}_i \cap \mathbb{M}^{*}_{i}$ is no easier than checking the non-emptiness of $\mathds{A}_i \cap \mathds{C}_{i}$. 
In this case, we find it convenient to work with the set of response vectors dominated by (or equal to) all points in $\mathds{C}_{uj}^{\mathrm{ell}}$ for $j = 1, 2, \ldots, k$. We define $\boldsymbol{w}_{j}=(w_{j1},w_{j2},\dots,w_{jd})$, where $w_{jr}=\max m_{jr} \text{ s.t. } \boldsymbol{m}_{j}\in \mathds{C}_{ij}$ for $r=1,2,\dots,d$, and $\mathds{W}(\boldsymbol{w}_{j})=\{\boldsymbol{m}\colon m_{r}\geq w_{jr}\ \mathrm{for\ }r=1,2,\dots,d\}$.
We refer to $\boldsymbol{w}_{j}$ as the \emph{dominated vector} and $\mathds{W}(\boldsymbol{w}_{j})$ as the \emph{dominated region} of the confidence region $\mathds{C}_{uj}^{\mathrm{ell}}$.
Proposition \ref{Pareto_Proposition} shows the connection between the sets $\mathds{W}(\boldsymbol{w}_{j})$ for $j \neq i$ and the task of checking the non-emptiness of $\mathds{A}_{i} \cap \mathds{C}_{u}^{\mathrm{ell}}$.
\begin{proposition}
\label{Pareto_Proposition}
For any $i=1,2,\dots,k$, 
$\mathds{A}_{i}\cap\mathds{C}_{u}^{\mathrm{ell}}\neq\emptyset$ if and only if $\mathds{C}_{ui}^{\mathrm{ell}}\cap(\bigcup\limits_{j\neq i} \mathds{W}(\boldsymbol{w}_{j}))^{c}\neq\emptyset$.
\end{proposition}
Proposition~\ref{Pareto_Proposition} allows us to screen systems by instead checking the non-emptiness of $\mathds{C}_{ii}\cap(\bigcup\limits_{j\neq i} \mathds{W}(\boldsymbol{w}_{j}))^{c}$, which is more tractable.
An analogous result for generic confidence regions is given in Proposition~\ref{proposition: Pareto_opt_full} in \ref{sec:proofs}.
Although the region $(\bigcup\limits_{j\neq i} \mathds{W}(\boldsymbol{w}_{j}))^{c}$ is non-convex, it can be expressed as the finite union of half-boxes whose corner points are associated with \emph{phantom Pareto systems}.
Phantom Pareto systems are dummy response vectors created by combining the responses of systems in the Pareto front of a set of points; we provide a more rigorous definition in~\ref{ec:phantom_pareto}.
Phantom Pareto systems have been used in multi-objective ranking and selection to aid in analyzing the probability of misclassifying systems relative to a known Pareto front \citep{applegate2017phantom}.
An algorithm for identifying phantom Pareto systems can be found in the online supplement of \cite{applegate2020multi}.

We represent the phantom Pareto systems associated with the Pareto front of the set $\{\boldsymbol{w}_{j}\colon j\neq i\}$ by $\{\boldsymbol{p}_{1},\boldsymbol{p}_{2},\dots,\boldsymbol{p}_{h}\}$, where $\boldsymbol{p}_{q}=(p_{q1},p_{q2},\dots,p_{qd})^{\intercal}$ for $q = 1, 2, \ldots, h$. 
It can be shown that $(\bigcup\limits_{j\neq i} \mathds{W}(\boldsymbol{w}_{j}))^{c} = \bigcup\limits_{q=1}^{h} \mathds{B}_{p}(\boldsymbol{p}_{q})$ where $\mathds{B}_{p}(\boldsymbol{p}_{q})=\{\boldsymbol{m}\colon m_{r}<p_{qr}\  \text{for}\ r=1,2,\dots,d\}$; full details can be found in the proof of Theorem 2 in the online supplement of \cite{hunter2016maximizing}. 
The non-emptiness of $\mathds{C}_{ui}\cap\bigcup\limits_{q=1}^{h} \mathds{B}_{p}(\boldsymbol{p}_{q})$ can then be determined by solving $h$ linearly constrained quadratic programs, where the number of Phantom Pareto systems, $h$, is $O(k^{\lfloor d/2 \rfloor})$.
To be more specific, $\mathds{C}_{ui}$ intersects a given $\mathds{B}_{p}(\boldsymbol{p}_{q})$ if and only if the optimal value of the optimization problem $\min\ (\widehat{\boldsymbol{\mu}}_{i}-\boldsymbol{m})^{\intercal}\widehat{\boldsymbol{\Sigma}}_{i}^{-1}(\widehat{\boldsymbol{\mu}}_{i}-\boldsymbol{m})\ \mathrm{s.t.}\  \boldsymbol{m} \leq \boldsymbol{p}_{q}$ is less than  $\frac{d(n_{i}-1)}{n_{i}(n_{i}-d)}\Lambda^\mathrm{ell}$.
This decomposition by phantom Pareto systems is shown in Figure~\ref{fig:Fig3} for the case where $d=2$.

\begin{figure}
    \centering
     \includegraphics[width=0.7\textwidth]{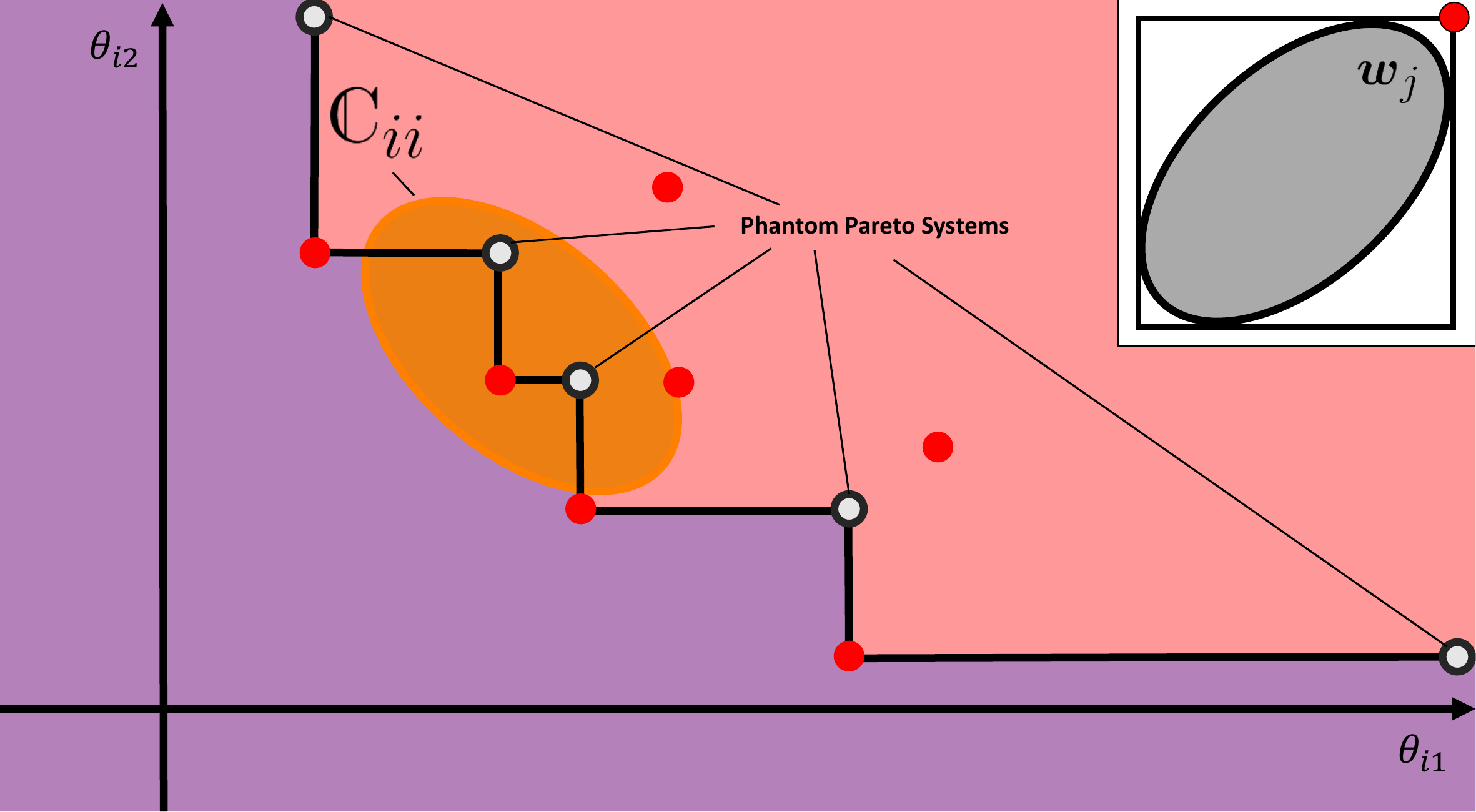}  
    
    \caption{Screening System $i$ for Pareto optimality using the FOSSA procedure with ellipsoid confidence regions. The vectors $\boldsymbol{w}_{j}$ are the dominated corners of the smallest rectangles containing the confidence regions $\mathds{C}_{uj}$ for $j\neq i$ and are shown as red dots. 
    The region $\cup_{j\neq i} \mathds{W}(\boldsymbol{w}_{j})$ is shaded a light red. The complementary region, shaded purple, can be decomposed in terms of the phantom Pareto systems, represented as grey dots. (The $\theta_{i1}$ coordinate value of the lower-right phantom Pareto system and the $\theta_{i2}$ coordinate value of the upper-left phantom Pareto system are both infinity.) Because $\mathds{C}_{ii}$ (shown as an orange ellipsoid) overlaps $(\cup_{j\neq i} \mathds{W}(\boldsymbol{w}_{j}))^{c}$ (the purple region), System $i$ will be returned.
    }
    \label{fig:Fig3}
\end{figure}

Algorithm~\ref{alg:pareto-ellipse} summarizes the FOSSA Ellipsoid procedure for Pareto optimality, which delivers the set-wise PAS guarantee.
In Line 5 of Algorithm~\ref{alg:pareto-ellipse}, the projections of the ellipsoid $\mathds{C}_{uj}^{\mathrm{ell}}$ on each dimension are used to determine the components of the vector $\boldsymbol{w}_{j}$.

\begin{algorithm}[tb]
\caption{FOSSA Ellipsoid for Pareto Optimality}
\label{alg:pareto-ellipse}
\begin{algorithmic}[1]
\Require $\widehat{\boldsymbol{\mu}}_{i},\widehat{\boldsymbol{\Sigma}}_{i},n_i$ for $i=1,2,\dots,k$, and $\Lambda^{\mathrm{ell}}$
\State $\mathcal{S}\gets\{1,2,\dots,k\}$
\For{$i=1,2,\dots,k$}
    \For{$r=1,2,\dots,d$}
        \State $w_{ir}\gets\widehat{\mu}_{ir}+\sqrt{\frac{d(n_i-1)}{n_i(n_i-d)}\Lambda^{\mathrm{ell}}\widehat{\Sigma}_{i}(r,r)}$
    \EndFor
\EndFor
\For{$i=1,2,\dots,k$}
    \If{$\mathds{C}_{ui}^{\mathrm{ell}}\cap\left(\displaystyle\bigcup_{j\neq i}\mathds{W}(\boldsymbol{w}_{j})\right)^c=\emptyset$}
        \State $\mathcal{S}\gets\mathcal{S}\setminus\{i\}$
        \State \textbf{break}
    \EndIf
\EndFor
\State \Return $\mathcal{S}$
\end{algorithmic}
\end{algorithm}

\section{Proofs of Theoretical Results}
\label{sec:proofs}

\begin{lemma}
    \label{lemma:more_fav}
    For any definition of acceptability and any set of non-empty confidence regions $\mathds{C}_{i1}, \mathds{C}_{i2}, \ldots, \mathds{C}_{ik}$ collectively satisfying Assumptions~\ref{assump: Acceptability}--\ref{assump: closed & bounded},
    \begin{enumerate}[(i)]
        \item $\mathds{B}^{*}(\mathds{C}_{ii})\neq\emptyset$, and for any $\boldsymbol{m} \in \mathds{C}_{ii}\setminus \mathds{B}^{*}(\mathds{C}_{ii})$, there exists some $\underline{\boldsymbol{m}} \in \mathds{B}^{*}(\mathds{C}_{ii})$ such that $\underline{\boldsymbol{m}} \prec \boldsymbol{m}$; and
        \item for any $j \neq i$, $\mathds{W}^{*}(\mathds{C}_{ij}) \neq \emptyset$, and for any $\boldsymbol{m}\in\mathds{C}_{ij}\setminus\mathds{W}^{*}(\mathds{C}_{ij})$, there exists some $\overline{\boldsymbol{m}} \in \mathds{W}^{*}(\mathds{C}_{ij})$ such that $\overline{\boldsymbol{m}} \succ \boldsymbol{m}$.
    \end{enumerate} 
\end{lemma}

\noindent {\bf Proof of Lemma~\ref{lemma:more_fav}:} To prove (i), fix an arbitrary $\boldsymbol{m} \in \mathds{C}_{ii}$. If $\boldsymbol{m} \in \mathds{B}^{*}(\mathds{C}_{ii})$, then $\mathds{B}^{*}(\mathds{C}_{ii})$ is non-empty. If instead $\boldsymbol{m}\notin \mathds{B}^{*}(\mathds{C}_{ii})$, then by the definition of $ \mathds{B}^{*}(\mathds{C}_{ii})$, $\mathds{B}(\boldsymbol{m})\cap \mathds{C}_{ii}\neq\{\boldsymbol{m}\}$. 
Also, $\mathds{B}(\boldsymbol{m})\cap \mathds{C}_{ii}\neq\emptyset$ because $\boldsymbol{m} \in \mathds{B}(\boldsymbol{m})$ and $\boldsymbol{m}\in\mathds{C}_{ii}$.
By Assumption \ref{assump: closed & bounded}, $\mathds{B}(\boldsymbol{m})\cap \mathds{C}_{ii} $ is a closed bounded set. 
Thus, by Assumption \ref{assump: existence}, $\mathds{B}^{*}(\mathds{B}(\boldsymbol{m})\cap \mathds{C}_{ii})\neq \emptyset$. 

Meanwhile, the event $\{\mathds{B}^{*}(\mathds{B}(\boldsymbol{m})\cap \mathds{C}_{ii})=\{\boldsymbol{m}\}\}$ is equivalent to the event $\{\boldsymbol{m}\in \mathds{B}^{*}(\mathds{C}_{ii})\}$, which contradicts the fact that $\boldsymbol{m} \notin \mathds{B}^{*}(\mathds{C}_{ii})$. 
So, there must exist some element $\underline{\boldsymbol{m}} \in \mathds{B}^{*}(\mathds{B}(\boldsymbol{m})\cap \mathds{C}_{ii})$ such that $\underline{\boldsymbol{m}} \neq \boldsymbol{m}$.
Because $\underline{\boldsymbol{m}} \in \mathds{B}(\boldsymbol{m})$ and $\underline{\boldsymbol{m}} \neq \boldsymbol{m}$, it follows that $\underline{\boldsymbol{m}} \prec \boldsymbol{m}$. Then by the transitive property in Assumption~\ref{assump: Acceptability}, $\mathds{B}(\underline{\boldsymbol{m}}) \subseteq \mathds{B}(\boldsymbol{m})$. This then implies that $\mathds{B}(\underline{\boldsymbol{m}})\cap \mathds{C}_{ii}=\mathds{B}(\underline{\boldsymbol{m}})\cap \mathds{B}(\boldsymbol{m}) \cap\mathds{C}_{ii}=\{\underline{\boldsymbol{m}}\}$
where the last equality comes from the fact that $\underline{\boldsymbol{m}} \in \mathds{B}^{*}(\mathds{B}(\boldsymbol{m})\cap \mathds{C}_{ii})$.
Thus, $\underline{\boldsymbol{m}}\in \mathds{B}(\underline{\boldsymbol{m}}) \cap \mathds{C}_{ii}$, and by the definition of $\mathds{B}^{*}(\mathds{C}_{ii})$, $ \underline{\boldsymbol{m}}\in \mathds{B}^{*}(\mathds{C}_{ii})$.

The proof for (ii) is similar. $\square$

\bigskip

\noindent {\bf Proof of Proposition~\ref{Proposition: Best and Worst }:}
Fix an abitrary $i \in \{1, 2, \ldots, k\}$. For notational convenience, we denote $\mathds{B}^{*}(\mathds{C}_{ii})$ by $\mathbb{M}^{*}_{ii}$ and $\mathds{W}^{*}(\mathds{C}_{ij})$ by $\mathbb{M}^{*}_{ij}$ for all $j\neq i$, hence,
$$\bigtimes\limits_{j=1}^{k}\mathbb{M}^{*}_{ij} \equiv \bigtimes\limits_{j=1}^{i-1}\mathds{W}^{*}(\mathds{C}_{ij})\times \mathds{B}^{*}(\mathds{C}_{ii})\times \bigtimes\limits_{j=i+1}^{k}\mathds{W}^{*}(\mathds{C}_{ij}).$$

When $\mathds{B}^{*}(\mathds{C}_{ii})\cap\mathds{F}\neq\emptyset$, it follows directly from the definitions of $\mathds{B}^{*}(\mathds{C}_{ii})$ and $\mathds{W}^{*}(\mathds{C}_{ij})$ that $\bigtimes\limits_{j=1}^{k}\mathbb{M}^{*}_{ij}\subseteq\mathbb{M}^{*}_{i}$. 

To show that $\mathbb{M}^{*}_{i}\subseteq\bigtimes\limits_{j=1}^{k}\mathbb{M}^{*}_{ij}$, we show that $\mathbb{M}^{*}_{i}\setminus\bigtimes\limits_{j=1}^{k}\mathbb{M}^{*}_{ij}=\emptyset$.
Suppose to the contrary that there exists some $\mathbf{M}\in\mathbb{M}^{*}_{i}\setminus\bigtimes\limits_{j=1}^{k}\mathbb{M}^{*}_{ij}$.
Then the set of indices $\mathcal{K}=\{\kappa\in\{1,2,\dots,k\}\ \mathrm{s.t.}\ \boldsymbol{m}_{\kappa}\notin \mathbb{M}^{*}_{i\kappa}\}$ must be non-empty. 
By Lemma~\ref{lemma:more_fav}, for each $\kappa \in \mathcal{K}$, if $\kappa = i$, then there exists a vector $\underline{\boldsymbol{m}}_{i} \in \mathbb{M}^*_{ii}$ such that $\underline{\boldsymbol{m}}_{i} \prec \boldsymbol{m}_{i}$, whereas if $\kappa = j \neq i$, then there exists $\overline{\boldsymbol{m}}_{j} \in \mathbb{M}^*_{ij}$ such that $\overline{\boldsymbol{m}}_{j} \succ \boldsymbol{m}_{j}$.
Define the matrix
\[ \mathbf{M}' = (\boldsymbol{m}'_{1}, \boldsymbol{m}'_{2}, \ldots, \boldsymbol{m}'_{k}) \quad \text{where} \quad \boldsymbol{m}'_{\kappa} = \begin{cases} m_{\kappa} & \text{if } \kappa \notin \mathcal{K}, \\ \overline{\boldsymbol{m}}_{i} & \text{if } \kappa \in \mathcal{K} \text{ and } \kappa = i, \\ \underline{\boldsymbol{m}}_{j} & \text{if } \kappa \in \mathcal{K} \text{ and } \kappa = j \neq i. \end{cases} \]
The matrix $\mathbf{M}'$ resembles $\mathbf{M}$, but the rows corresponding to indices in $\mathcal{K}$ have been replaced with other vectors.
By construction, $\mathbf{M}'\in \mathds{C}_{i}\setminus\{\mathbf{M}\} \ $ and $\boldsymbol{m}'_{i} \preccurlyeq  \boldsymbol{m}_{i}$ and $\boldsymbol{m}'_{j} \succcurlyeq \boldsymbol{m}_{j} \mathrm{\ for\ all\ }j\neq i$. Therefore, $\mathbf{M}\notin\mathbb{M}^{*}_{i}$, a contradiction.

Combining the two implications, we have shown that $\mathbb{M}^{*}_{i} =\bigtimes\limits_{j=1}^{k}\mathbb{M}^{*}_{ij}$. $\square$

\bigskip

\noindent {\bf Proof of Theorem~\ref{theorem: FOSSA}:}
Because $\mathbb{M}^{*}_{i}\subseteq\mathds{C}_{i}$, $\mathds{A}_{i} \cap \mathbb{M}^{*}_{i}\neq\emptyset$ implies $\mathds{A}_{i} \cap \mathds{C}_{i}\neq\emptyset$.
We prove the converse implication by contradiction.
Suppose that $\mathds{A}_{i}\cap\mathds{C}_{i}\neq\emptyset$, but $\mathds{A}_{i} \cap \mathbb{M}^{*}_{i}=\emptyset$.
Hence, there exists some $\mathbf{M}\in\mathds{A}_{i}\cap\mathds{C}_{i}\ \mathrm{s.t.}\ \mathbf{M}\notin\mathbb{M}^{*}_{i}$.

We consider two cases. Suppose $\mathds{B}^{*}(\mathds{C}_{ii})\cap\mathds{F}=\emptyset$. By the definition of $\mathds{B}^{*}(\mathds{C}_{ii})$, there is no configuration in $\mathds{C}_{i}$ for which the response vector of System $i$ is feasible. From the definition of acceptability given in Assumption~\ref{assump: Acceptability}, $\mathds{A}_{i}\cap\mathds{C}_{i} = \emptyset$, a contradiction.

On the other hand, suppose $\mathds{B}^{*}(\mathds{C}_{ii})\cap\mathds{F}\neq\emptyset$.
Because $\mathbf{M}\in\mathds{C}_{i}\setminus\mathbb{M}^{*}_{i}$, by Proposition \ref{Proposition: Best and Worst }, there exists a non-empty set of indices $\mathcal{K}=\{\kappa\in\{1,2,\dots,k\}\ \mathrm{s.t.}\ \boldsymbol{m}_{\kappa}\notin \mathbb{M}^{*}_{i\kappa}\}$. By passing over $\kappa \in \mathcal{K}$ and replacing the $\kappa$th row of $\mathbf{M}$ with arbitrary elements in $\mathbb{M}^{*}_{i\kappa}$, we can construct a $\mathbf{M}'\in \mathds{C}_{i}\setminus\{\mathbf{M}\} \ $ such that $\boldsymbol{m}'_{i} \preccurlyeq  \boldsymbol{m}_{i}$ and $\boldsymbol{m}'_{j} \succcurlyeq \boldsymbol{m}_{j}$ for all $j\neq i$.
By construction, $\mathbf{M}' \in \mathbb{M}^*_i$.

Furthermore, by Assumption~\ref{assump: Acceptability}, $\mathbf{M}\in\mathds{A}_{i}$ implies that $\boldsymbol{m}_{i}\in\mathds{F}$ and $\boldsymbol{m}_{j}\nprec \boldsymbol{m}_{i}$ for all $j\neq i$. 
We next prove by contradiction that $\boldsymbol{m}_{j}'\nprec \boldsymbol{m}_{i}'$ for all $j\neq i$.
Suppose there exists an index $g \neq i$ such that $\boldsymbol{m}_{g}'\prec \boldsymbol{m}_{i}'$.
Because $\boldsymbol{m}'_{i} \preccurlyeq  \boldsymbol{m}_{i} \mathrm{\ and\ } \ \boldsymbol{m}'_{j} \succcurlyeq \boldsymbol{m}_{j} \mathrm{\ for\ all\ }j\neq i$, it follows that $\boldsymbol{m}_{g}\preccurlyeq\boldsymbol{m}_{g}'\prec \boldsymbol{m}_{i}'\preccurlyeq\boldsymbol{m}_{i}$.
Thus, by the transitive property of Assumption~\ref{assump: Acceptability}, $\boldsymbol{m}_{g}\prec\boldsymbol{m}_{i}$, which contradicts the fact that $\boldsymbol{m}_{j}\nprec \boldsymbol{m}_{i}$ for all $j\neq i$. Therefore $\boldsymbol{m}'_{j} \nprec \boldsymbol{m}'_{i}$ for all $j \neq i$.
In addition, $\boldsymbol{m}_{i}\in\mathds{F}$ and $\boldsymbol{m}'_{i} \preccurlyeq  \boldsymbol{m}_{i}$ together imply that $\boldsymbol{m}'_{i} \in \mathds{F}$ and $\mathbf{M}'\in \mathds{A}_{i}$, and thus $\mathbf{M}'\in \mathds{A}_{i}\cap \mathbb{M}^{*}_{i}$, a contradiction.
$\square$

\bigskip
Before proving Proposition~\ref{Proposition for Constrained Optimality}, we first prove a supporting lemma.
\begin{lemma}
\label{lemma: Opt is bounded}
For the definition of acceptability as constrained optimality,
\begin{enumerate}[(i)]
\item if $\mathds{C}_{ii}\cap\mathds{F}\neq\emptyset$, the optimization problem $\min m_{1}\mathrm{\ s.t.\ }\boldsymbol{m}\in\mathds{C}_{ii}\cap\mathds{F}$ is bounded;
\item if $\mathds{C}_{ij}\cap\mathds{F}^{c}=\emptyset$, the optimization problem $\max m_{1}\mathrm{\ s.t.\ }\boldsymbol{m}\in\mathds{C}_{ij}$ is bounded for all $j\neq i$.
\end{enumerate}
\end{lemma}

\noindent{\bf Proof of Lemma~\ref{lemma: Opt is bounded}:}
Proof of (i): For an arbitrary $\boldsymbol{m}^{a}\in \mathds{C}_{ii}\cap\mathds{F}$, $\mathds{B}(\boldsymbol{m}^{a})\cap\mathds{C}_{ii}$ is closed and bounded, by Assumption (\ref{assump: closed & bounded}). 
Therefore, $\mathds{B}(\boldsymbol{m}^{a})\cap\mathds{C}_{ii}\cap\mathds{F}$ is also closed and bounded.
Hence, $\min m_{1}\mathrm{\ s.t.\ }\boldsymbol{m}\in\mathds{B}({\boldsymbol{m}^{a}})\cap\mathds{C}_{ii}\cap\mathds{F}$ is bounded, because the objective function is continuous and the feasible region is compact.
Denote the optimizer of $\min m_{1}\mathrm{\ s.t.\ }\boldsymbol{m}\in\mathds{B}({\boldsymbol{m}^{a}})\cap\mathds{C}_{ii}\cap\mathds{F}$ as $\boldsymbol{m}^{*-}$ and its objective function value as $m^{*-}_{1}$.
Since $\boldsymbol{m}^{a}\in \mathds{F}$, for any $\boldsymbol{m}'\in\mathds{B}({\boldsymbol{m}^{a}})^{c}\cap\mathds{F}$, $m'_{1} > m^{*-}_{1}$.
To be specific, $\boldsymbol{m}'\in\mathds{F}$ but $\boldsymbol{m}' \notin \mathds{B}({\boldsymbol{m}^{a}})$, which indicates that $m'_{1} > {m}^{a}_{1}$ and, therefore, $m'_{1} > m^{*-}_{1}$.
Because the choice of ${\boldsymbol{m}^{a}}$ is arbitrary, we can conclude that $\boldsymbol{m}^{*-}$ is also the optimizer of the optimization problem  $\min m_{1}\mathrm{\ s.t.\ }\boldsymbol{m}\in\mathds{C}_{ii}\cap\mathds{F}$.

Proof of (ii): The event $\mathds{C}_{ij}\cap\mathds{F}^{c}=\emptyset$ implies that $\mathds{C}_{ij}\subseteq\mathds{F}$.
For an arbitrary $\boldsymbol{m}^{b}\in \mathds{C}_{ij}$, $\mathds{W}(\boldsymbol{m}^{b})\cap\mathds{C}_{ij}$ is closed and bounded, by Assumption \ref{assump: closed & bounded}. 
Hence, $\min m_{1}\mathrm{\ s.t.\ }\boldsymbol{m}\in\mathds{W}({\boldsymbol{m}^{b}})\cap\mathds{C}_{ij}$ is bounded, because the objective function is continuous and the feasible region is compact.
Denote the optimizer of $\max m_{1}\mathrm{\ s.t.\ }\boldsymbol{m}\in\mathds{W}({\boldsymbol{m}^{b}})\cap\mathds{C}_{ij}$ as $\boldsymbol{m}^{*+}$ and its objective value as $m^{*+}_{1}$.
For any vector $\boldsymbol{m}''$ in $\mathds{W}({\boldsymbol{m}^{b}})^{c}\cap\mathds{C}_{ij}$, it follows that $\boldsymbol{m}''\in\mathds{C}_{ij}\subseteq\mathds{F}$.
From the definition of acceptability, the primary response for any vector in $\mathds{W}({\boldsymbol{m}^{b}})^{c}\cap\mathds{F}$ must be less than ${m}^{b}_{1}$.
Because both $\boldsymbol{m}''$ and $\boldsymbol{m}^{b}$ are in $\mathds{F}$, $m''_{1}$ must be less than ${m}^{b}_{1}$ for any $\boldsymbol{m}''$ in $\mathds{W}({\boldsymbol{m}^{b}})^{c}\cap\mathds{C}_{ij}$.
Because the choice of ${\boldsymbol{m}^{b}}$ is arbitrary, we can conclude that $\boldsymbol{m}^{*+}$ is also the optimizer of the optimization problem  $\max m_{1}\mathrm{\ s.t.\ }\boldsymbol{m}\in\mathds{C}_{ij}$.
$\square$

\bigskip
\noindent{\bf Proof of Proposition~\ref{Proposition for Constrained Optimality}:} 
Proving Proposition~\ref{Proposition for Constrained Optimality} is equivalent to proving that for all $i=1,2,\ldots,k$, $\mathds{A}_{i}\cap\bigtimes\limits_{j=1}^{k}\mathbb{M}^{*}_{ij}\neq\emptyset$ if and only if $ \tau_{ii}\neq\infty \text{ and } \tau_{ii}\leq \tau_{ij} \ \mathrm{for\ all\ } j\neq i$.

We first prove the forward implication.
Because $\mathds{A}_{i}\cap\bigtimes\limits_{j=1}^{k}\mathbb{M}^{*}_{ij}\neq\emptyset$, there exists an $\mathbf{M}\in\mathds{A}_{i}\cap\bigtimes\limits_{j=1}^{k}\mathds{C}_{ij}$ such that $\boldsymbol{m}_{i}\in\mathds{F}$ and $m_{i1}\leq m_{j1}$ for all $j \text{ s.t. }\boldsymbol{m}_{j}\in \mathds{F}$, by the definition of $\mathds{A}_{i}$ and the fact that $\bigtimes\limits_{j=1}^{k}\mathbb{M}^{*}_{ij}\subseteq\bigtimes\limits_{j=1}^{k}\mathds{C}_{ij}$.
Thus, $\boldsymbol{m}_{j}\in\mathds{C}_{ij}$ for all $j=1,2,\dots,k$. 
Combined with the fact that $\boldsymbol{m}_{i}\in\mathds{F}$, we know that $\mathds{C}_{ii}\cap\mathds{F}\neq\emptyset$, thus $\tau_{ii}\neq\infty$ by the definition of $\tau_{ii}$ and Lemma \ref{lemma: Opt is bounded}.
Then because $\mathds{C}_{ii}\cap\mathds{F}\neq\emptyset$, $\tau_{ii}=\min m_{1}\ \mathrm{s.t.}\ \boldsymbol{m}\in \mathds{C}_{ii}\cap\mathds{F}$ and $\boldsymbol{m}_{i}$ is in the feasible region of this optimization problem, we can conclude that $\tau_{ii}\leq m_{i1}$.
Divide the remaining systems into two disjoint classes based on their feasibility: define $\mathcal{J}_{1}=\{j=1,2,\dots,k\colon j\neq i \text{ and }\boldsymbol{m}_{j}\notin\mathds{F}\}$ and $\mathcal{J}_{2}=\{ j=1,2,\dots,k\colon j\neq i \text{ and }\boldsymbol{m}_{j}\in\mathds{F}\}$. We now prove that $\tau_{ij}\geq\tau_{ii}$ for all $j\in\mathcal{J}_{1}\cup\mathcal{J}_{2}$.
For all $j\in\mathcal{J}_{1}$, $\boldsymbol{m}_{j}\in\mathds{C}_{ij}\cap\mathds{F}^{c}$ indicates that $\mathds{C}_{ij}\cap\mathds{F}^{c}\neq\emptyset$, therefore $\tau_{ij}=\infty\geq\tau_{ii}$ by definition.
For all $j\in\mathcal{J}_{2}$, $\boldsymbol{m}_{j}\in\mathds{C}_{ij}\cap\mathds{F}$ indicates that $\mathds{C}_{ij}\cap\mathds{F}\neq\emptyset$.
Therefore, $\tau_{ij}=\max m_{1}\ \mathrm{s.t.}\ \boldsymbol{m}\in\mathds{C}_{ij}$ and $\tau_{ij}\geq m_{j1}$, because $\boldsymbol{m}_{j}\in\mathds{C}_{ij}$.
Thus, for all $j\in\mathcal{J}_{2}$, $\tau_{ij}\geq m_{j1}\geq m_{i1} \geq\tau_{ii}$.
We have thus shown that $\tau_{ii}\neq\infty$ and $\tau_{ij}\geq\tau_{ii}$ for all $j\neq i$.

We next prove the reverse implication.
By the definition of $\tau_{ii}$ and Lemma \ref{lemma: Opt is bounded}, $\tau_{ii}=\infty$ if and only if $\mathds{C}_{ii}\cap\mathds{F}=\emptyset$.
Hence $\tau_{ii}\neq\infty$ indicates that $\mathds{C}_{ii}\cap\mathds{F}\neq\emptyset$, and there exists an $\boldsymbol{m}'_{i}\in\mathds{C}_{ii}\cap\mathds{F}$ with $m'_{i1}=\tau_{ii}$.
Again divide the systems other than System $i$ into two disjoint classes: $\mathcal{J}_{3}=\{j=1,2,\dots,k\colon j\neq i \text{ and }\tau_{ij}\neq\infty\}$ and   $\mathcal{J}_{4}=\{ j=1,2,\dots,k\colon j\neq i \text{ and }\tau_{ij}=\infty\}$.
For $j\in\mathcal{J}_{3}$, $\tau_{ij}\neq\infty$ implies that $\mathds{C}_{ij}\subseteq\mathds{F}$ and that there exists an $\boldsymbol{m}'_{j} \in \mathds{C}_{ij} \subseteq \mathds{F}$ for which $m'_{j1}\leq\tau_{ij}$.
For $j\in\mathcal{J}_{4}$, $\tau_{ij}=\infty$ implies that $\mathds{C}_{ij}\cap\mathds{F}^{c}\neq\emptyset$, thus there exists an $\boldsymbol{m}'_{j}\in\mathds{C}_{ij}\cap\mathds{F}^{c}$.
Let $\mathbf{M}'=(\boldsymbol{m}'_{1},\boldsymbol{m}'_{2},\dots,\boldsymbol{m}'_{k})^{\intercal}$. Because  $\boldsymbol{m}'_{i}\in\mathds{C}_{ii}$ and $\boldsymbol{m}'_{j}\in\mathds{C}_{ij}$ for all $j=1,2,\dots,k$, we have that $\mathbf{M}'\in\mathds{A}_{i}\cap\bigtimes\limits_{j=1}^{k}\mathds{C}_{ij}$.
We have thus proven that $\mathds{A}_{i}\cap\bigtimes\limits_{j=1}^{k}\mathds{C}_{ij}\neq\emptyset$, and, by Theorem \ref{theorem: FOSSA}, that $\mathds{A}_{i}\cap\bigtimes\limits_{j=1}^{k}\mathbb{M}^{*}_{ij}\neq\emptyset$.
$\square$

\bigskip

\begin{lemma}
    \label{lem:ellipse_proj}
    For any symmetric positive-definite matrix $\mathbf{A} \in \mathds{R}^{d \times d}$ and scalar $b > 0$, the projection of an ellipsoid defined by $\{\boldsymbol{x} \in \mathds{R}^{d} \colon \boldsymbol{x}^{\intercal}\mathbf{A}\boldsymbol{x}\leq b\}$ onto dimension $r$, $r = 1, 2, \ldots, d$, is $\left[-\sqrt{b a_{r, r}^{-1}}, \sqrt{b a_{r, r}^{-1}} \right]$ where $a_{i,j}$ is the $(i,j)$th element of $\mathbf{A}$.
    
\end{lemma}

\noindent {\bf Proof of Lemma~\ref{lem:ellipse_proj}:}
Fix $r \in \{1, 2, \ldots, d\}$ and consider the optimization problem
\begin{equation}
    \label{eqn:ellipse_opt}
    \min_{\boldsymbol{x} \in \mathds{R}^d} \boldsymbol{e}_{r}^{\intercal}\boldsymbol{x}\ \mathrm{s.t.}\ \boldsymbol{x}^{\intercal}\mathbf{A}\boldsymbol{x}\leq b,
\end{equation}
where $\boldsymbol{e}_{r}$ is a unit vector whose $r$th element is $1$.
A first-order KKT condition states that the optimal solution $\boldsymbol{x}^*$ necessarily satisfies $\boldsymbol{e}_{r}-2\lambda \mathbf{A}\boldsymbol{x}^{*}=0$, where $\lambda$ is a scalar. Thus, $\boldsymbol{x}^{*}=(2\lambda)^{-1}\mathbf{A}^{-1}\boldsymbol{e}_{r}$. Since the sole constraint in (\ref{eqn:ellipse_opt}) must be active at optimality, we also have that $\boldsymbol{x}^{*\intercal}\mathbf{A}^{-1}\boldsymbol{x}^{*}= b$. Hence, the optimal solution to (\ref{eqn:ellipse_opt}) is $\boldsymbol{x}^{*}=-\sqrt{b/a^{-1}_{r,r}}\mathbf{A}^{-1}\boldsymbol{e}_{r}$ and the optimal objective function value is $ \boldsymbol{e}_{r}^{\intercal}\boldsymbol{x}^{*}=-\sqrt{b/a^{-1}_{r,r}}\boldsymbol{e}_{r}^{\intercal}\mathbf{A}^{-1}\boldsymbol{e}_{r}=-\sqrt{ba^{-1}_{r,r}}$.

Similarly, the optimal solution and objective function value for
$\max_{\boldsymbol{x} \in \mathds{R}^d} \boldsymbol{e}_{r}^{\intercal}\boldsymbol{x}\ \mathrm{s.t.}\ \boldsymbol{x}^{\intercal}\mathbf{A}\boldsymbol{x}\leq b$
are $\sqrt{b/a^{-1}_{r,r}}\mathbf{A}^{-1}\boldsymbol{e}_{r}$ and $\sqrt{ba^{-1}_{r,r}}$, respectively. $\square$

\bigskip

\noindent {\bf Proof of Proposition~\ref{Proposition: ellipsoid projection}:}
The event $\{\mathds{C}_{ui}^\mathrm{ell}\cap\mathds{F}^{c}=\emptyset\}$ is equivalent to $\{\mathds{C}_{ui}^\mathrm{ell}\subseteq\mathds{F}\}$. 
By Lemma~\ref{lem:ellipse_proj}, the projection of the ellipsoid $\mathds{C}_{ui}^\mathrm{ell}$ onto the $r$th dimension is
\[ \left[\widehat{{\mu}}_{ir}-\sqrt{\frac{d(n_{i}-1)}{n_{i}(n_{i}-d)} \Lambda^\mathrm{ell}\widehat{\Sigma}_{i}(r,r)}, \: \widehat{{\mu}}_{ir}+\sqrt{\frac{d(n_{i}-1)}{n_{i}(n_{i}-d)} \Lambda^\mathrm{ell}\widehat{\Sigma}_{i}(r,r)}\right]. \]
For $\mathds{F} = \{\boldsymbol{m} \in \mathds{R}^d\colon m_r \leq \mu_r^{\dagger} \text{ for } r = 2, 3, \ldots, d\}$, it follows that
$\mathds{C}_{ui}^\mathrm{ell}\subseteq\mathds{F}$ if the Cartesian product of the projection of $\mathds{C}_{ui}^\mathrm{ell}$ on each dimension is contained in $\mathds{F}$, i.e., if 
\[ \widehat{\mu}_{ir} + \sqrt{\frac{d(n_{i}-1)}{n_{i}(n_{i}-d)} \Lambda^\mathrm{ell}\widehat{\Sigma}_{i}(r,r)} \leq \mu_r^{\dagger} \text{ for all } r=2,3,\dots,d. \quad \square \]

Propositions \ref{Proposition: box projection} and \ref{Proposition: half-box projection} give analogous equivalent conditions for checking the plausible feasibility of a System $j$ for the box and half-box confidence regions, respectively; their proofs follow similar arguments to that of Proposition~\ref{Proposition: ellipsoid projection} and are omitted.

\begin{proposition}
\label{Proposition: box projection}
$\mathds{C}_{uj}^\mathrm{box}\cap\mathds{F}^{c}=\emptyset$ if and only if $\widehat{{\mu}}_{jr} +\Lambda^\mathrm{box}\sqrt{\widehat{\Sigma}_{j}(r,r)/n_{j}}\leq\mu^{\dagger}_{r} \text{ for all }r=2,3,\dots,d$.
\end{proposition}

\begin{proposition}
\label{Proposition: half-box projection}
$\mathds{C}_{ij}^\mathrm{hb}\cap\mathds{F}^{c}=\emptyset$ if and only if $\widehat{{\mu}}_{jr} +\Lambda^\mathrm{hb}\sqrt{\widehat{\Sigma}_{j}(r,r)/n_{j}}\leq\mu^{\dagger}_{r}\text{ for all }r=2,3,\dots,d$.
\end{proposition}

\bigskip

As stated previously in Section~\ref{subsec:pareto_ellipsoid}, for a closed confidence region $\mathds{C}_{ij}$ that is bounded above in each coordinate, we define $\boldsymbol{w}_{ij}=(w_{ij1},w_{ij2},\dots,w_{ijd})$, where $w_{ijr}=\max\limits_{\boldsymbol{m}\in \mathds{C}_{ij}} m_{r}$ for $r=1,2,\dots,d$ and $\mathds{W}^{+}(\boldsymbol{m})\equiv\{\boldsymbol{m}'\colon \boldsymbol{m}'\succ_{p} \boldsymbol{m}\}$.
Let $\mathcal{C}^{1}=\{j\neq i\colon\boldsymbol{w}_{ij}\notin \mathds{C}_{ij}\}$ and $\mathcal{C}^{2}=\{j\neq i\colon\boldsymbol{w}_{ij}\in \mathds{C}_{ij}\}$.
When $\mathcal{C}^1 = \emptyset$, $\mathbb{M}^{*}_{i}=\bigtimes\limits_{j=1}^{k}\mathbb{M}^{*}_{ij} = \bigtimes\limits_{j=1}^{i-1}\mathds{W}^{*}(\mathds{C}_{ij})\times \mathds{B}^{*}(\mathds{C}_{ii})\times \bigtimes\limits_{j=i+1}^{k}\mathds{W}^{*}(\mathds{C}_{ij})$ is a singleton, in which case the non-emptiness of $\mathds{A}_{i} \cap \mathbb{M}^{*}_{i}$ can be easily checked.
When $\mathcal{C}^1 \neq \emptyset$, we will make use of Proposition \ref{proposition: Pareto_opt_full}, which handles the general case where confidence regions for each system may vary in geometry.
Proposition \ref{Pareto_Proposition}, which appears in the paper, is a corollary of Proposition \ref{proposition: Pareto_opt_full}, where $\mathds{C}_{ij} = \mathds{C}_{uj}^{\mathrm{ell}}$ for all $j = 1, 2, \ldots, k$.

\begin{lemma}
\label{lemma: Dominated Corner of Confidence Region}
For the definition of acceptability as Pareto optimality,
\begin{enumerate}[(i)]
    \item for any $\boldsymbol{m}_{ii}\in\mathds{C}_{ii}$ and any
    $j\in\mathcal{C}^{1}$, there exists $\boldsymbol{m}_{ij}\in\mathds{C}_{ij} $ such that $ \boldsymbol{m}_{ij}\nprec_{p} \boldsymbol{m}_{ii}$ if and only if $\boldsymbol{m}_{ii}\notin \mathds{W}(\boldsymbol{w}_{ij})$; and
    \item for any $\boldsymbol{m}_{ii}\in\mathds{C}_{ii}$ and any $j\in\mathcal{C}^{2}$, there exists $\boldsymbol{m}_{ij}\in\mathds{C}_{ij} $ such that $ \boldsymbol{m}_{ij}\nprec_{p} \boldsymbol{m}_{ii}$ if and only if $\boldsymbol{m}_{ii}\notin \mathds{W}^{+}(\boldsymbol{w}_{ij})$.
\end{enumerate}

\end{lemma}
\noindent {\bf Proof of Lemma~\ref{lemma: Dominated Corner of Confidence Region}:}
Proof of (i): 
We will first prove the forward implication by contradiction.
Assume that there exists $\boldsymbol{m}_{ij} \in \mathds{C}_{ij}$ such that $\boldsymbol{m}_{ij}\nprec_{p} \boldsymbol{m}_{ii}$ and  $\boldsymbol{m}_{ii}\in \mathds{W}(\boldsymbol{w}_{ij})$.
By the definition of $\mathds{B}(\cdot)$ and $\boldsymbol{w}_{ij}$, $\mathds{C}_{ij}\subseteq \mathds{B}(\boldsymbol{w}_{ij})$.
Combined with the fact that $\boldsymbol{m}_{ii}\in \mathds{W}(\boldsymbol{w}_{ij})$, we conclude that $\boldsymbol{m}_{ij}\preccurlyeq_{p}\boldsymbol{w}_{ij}\preccurlyeq_{p}\boldsymbol{m}_{ii}$.
Since $j \in \mathcal{C}^1$, $\boldsymbol{m}_{ij} \neq \boldsymbol{m}_{ii}$. It follows that $\boldsymbol{m}_{ij}\prec_{p}\boldsymbol{m}_{ii}$, a contradiction.

We next prove the reverse implication.
Since $\boldsymbol{m}_{ii} \notin \mathds{W}(\boldsymbol{w}_{ij})$, there exists some $c\in\{1,2,\dots,d\}$ such that $m_{iic}< \max\limits_{\boldsymbol{m}\in \mathds{C}_{ij}} m_{c}$.
Let $\boldsymbol{m}^{*}_{c} =  \argmax\limits_{\boldsymbol{m}\in \mathds{C}_{ij}} m_{c}$, hence $\boldsymbol{m}^{*}_{c} \in \mathds{C}_{ij}$.
Because $\boldsymbol{m}_{ii}$ outperforms $\boldsymbol{m}^{*}_{c}$ in objective $c$, $\boldsymbol{m}^{*}_{c}\nprec_{p} \boldsymbol{m}_{ii}$.

Proof of (ii):
We will first prove the forward implication by contradiction.
Assume that there exists $\boldsymbol{m}_{ij}\in\mathds{C}_{ij} $ and $ \boldsymbol{m}_{ij}\nprec_{p} \boldsymbol{m}_{ii}$ and  $\boldsymbol{m}_{ii}\in \mathds{W}^{+}(\boldsymbol{w}_{ij})$.
By the definition of $\mathds{B}(\cdot)$ and $\boldsymbol{w}_{ij}$, $\mathds{C}_{ij}\subseteq \mathds{B}(\boldsymbol{w}_{ij})$.
Combined with the fact that $\boldsymbol{m}_{ii}\in \mathds{W}^+(\boldsymbol{w}_{ij})$, we conclude that $\boldsymbol{m}_{ij}\preccurlyeq_{p}\boldsymbol{w}_{ij}\prec_{p}\boldsymbol{m}_{ii}$.
Hence, $\boldsymbol{m}_{ij}\prec_{p}\boldsymbol{m}_{ii}$, a contradiction.

We next prove the reverse implication.
Because $\boldsymbol{m}_{ii} \notin \mathds{W}^{+}(\boldsymbol{w}_{ij})$, either
$\boldsymbol{m}_{ii} = \boldsymbol{w}_{ij}$ or $\boldsymbol{m}_{ii} \in \mathds{W}(\boldsymbol{w}_{ij})^{c}$.
If $\boldsymbol{m}_{ii}=\boldsymbol{w}_{ij}$, then since $j\in\mathcal{C}^{2}$, letting $\boldsymbol{m}_{ij}=\boldsymbol{w}_{ij}$ implies that $\boldsymbol{m}_{ij}\in\mathds{C}_{ij}$ and $ \boldsymbol{m}_{ij} \nprec_{p} \boldsymbol{m}_{ii}$.
If instead $\boldsymbol{m}_{ii} \notin \mathds{W}(\boldsymbol{w}_{ij})$,
then there exists some $c\in\{1,2,\dots,d\}$ such that $m_{iic}<\max\limits_{\boldsymbol{m}\in \mathds{C}_{ij}} m_{c}$.
Let $\boldsymbol{m}^{*}_{c} = \argmax\limits_{\boldsymbol{m}\in \mathds{C}_{ij}} m_{c}$, hence $\boldsymbol{m}^{*}_{c} \in \mathds{C}_{ij}$.
Because $\boldsymbol{m}_{ii}$ outperforms $\boldsymbol{m}^{*}_{c}$ in objective $c$, $\boldsymbol{m}^{*}_{c}\nprec_{p} \boldsymbol{m}_{ii}$.
$\square$

\begin{proposition}
\label{proposition: Pareto_opt_full}
For the definition of acceptability as Pareto optimality, if $\boldsymbol{w}_{ij}\notin \mathds{C}_{ij}$ for all $j\neq i$, then 
$$\mathds{A}_{i}\cap\mathds{C}_{i}\neq\emptyset \text{ if and only if } \mathds{C}_{ii}\cap\left(\bigcup\limits_{j\in\mathcal{C}^{1}} \mathds{W}(\boldsymbol{w}_{ij})\cup\bigcup\limits_{j\in\mathcal{C}^{2}} \mathds{W}^{+}(\boldsymbol{w}_{ij})\right)^{c}\neq\emptyset.$$
\end{proposition}

\noindent {\bf Proof of Proposition ~\ref{proposition: Pareto_opt_full}:}
We prove the forward implication by contradiction.
Assume that $\mathds{A}_{i}\cap\mathds{C}_{i}\neq\emptyset$ and $\mathds{C}_{ii}\cap\left(\bigcup\limits_{j\in\mathcal{C}^{1}} \mathds{W}(\boldsymbol{w}_{ij})\cup\bigcup\limits_{j\in\mathcal{C}^{2}} \mathds{W}^{+}(\boldsymbol{w}_{ij})\right)^{c}=\emptyset$.
Because $\mathds{A}_{i}\cap\mathds{C}_{i}\neq\emptyset$, there exists a configuration $\mathbf{M}'\in\mathds{A}_{i}\cap\bigtimes\limits_{j=1}^{k}\mathbb{M}^{*}_{ij}$. By the definition of acceptability, $\boldsymbol{m}_{ii}'\nsucc_{p}\boldsymbol{m}_{ij}'$ for all $j\neq i$ and because $\mathds{M}_{ij}^* \subseteq \mathds{C}_{ij}$, $\boldsymbol{m}_{ii}' \in \mathds{C}_{ii}$ and $\boldsymbol{m}_{ij}' \in \mathds{C}_{ij}$ for all $j \neq i$.
By Lemma \ref{lemma: Dominated Corner of Confidence Region}, $\boldsymbol{m}_{ii}'\notin\mathds{W}(\mathds{C}_{ij})$ for $j\in\mathcal{C}^{1}$ and $\boldsymbol{m}_{ii}'\notin\mathds{W}^{+}(\mathds{C}_{ij})$ for all $j\in\mathcal{C}^{2}$.
Therefore $\boldsymbol{m}_{ii}'\in \left(\bigcup\limits_{j\in\mathcal{C}^{1}} \mathds{W}(\boldsymbol{w}_{ij})\cup\bigcup\limits_{j\in\mathcal{C}^{2}} \mathds{W}^{+}(\boldsymbol{w}_{ij})\right)^{c}$.
Because $\boldsymbol{m}_{ii}'\in\mathds{C}_{ii}$, $\mathds{C}_{ii}\cap\left(\bigcup\limits_{j\in\mathcal{C}^{1}} \mathds{W}(\boldsymbol{w}_{ij})\cup\bigcup\limits_{j\in\mathcal{C}^{2}} \mathds{W}^{+}(\boldsymbol{w}_{ij})\right)^{c}\neq\emptyset$.

We next prove the reverse implication.
Assume there exists a vector $\boldsymbol{m}'_{ii} \in\mathds{C}_{ii}\cap\left(\bigcup\limits_{j\in\mathcal{C}^{1}} \mathds{W}(\boldsymbol{w}_{ij})\cup\bigcup\limits_{j\in\mathcal{C}^{2}} \mathds{W}^{+}(\boldsymbol{w}_{ij})\right)^{c}$.
Then for all $j\in\mathcal{C}^{1}$, $\boldsymbol{m}'_{ii}\notin \mathds{W}(\boldsymbol{w}_{ij})$ and for all $j\in\mathcal{C}^{2}$, $\boldsymbol{m}'_{ii}\notin \mathds{W}^{+}(\boldsymbol{w}_{ij})$.
By Lemma \ref{lemma: Dominated Corner of Confidence Region}, for any $j \neq i$ we can find a vector
$\boldsymbol{m}'_{ij}\in\mathds{C}_{ij}$ such that $\boldsymbol{m}'_{ij} \nprec_{p} \boldsymbol{m}_{ii}$.
By the definition of acceptability, $\mathbf{M}' \equiv \bigtimes\limits_{j=1}^{k}\boldsymbol{m}'_{ij}\in\mathds{A}_{i}$, and by construction, $\mathbf{M}' \in\bigtimes\limits_{j=1}^{k}\mathds{C}_{ij}=\mathds{C}_{i}$.
Therefore, $\mathds{A}_{i}\cap\mathds{C}_{i}\neq\emptyset$.
$\square$

\section{Implementation of FOSSA Procedures in Parallel Computing Environments}
\label{sec:parallel}
In this section, we show that FOSSA procedures can be naturally parallelized through a divide-and-conquer scheme under a master-worker computing framework without compromising their screening power.
FOSSA procedures can be naively parallelized by tasking each worker with checking the non-emptiness of $\mathds{A}_{i}\cap\mathds{C}_{i}$ for a subset of systems.
However, this approach requires that the simulation outputs (or at least the sufficient statistics) for \emph{all} systems be shared with \emph{all} workers, which leads to significant memory and/or communication overhead.
Alternatively, under a divide-and-conquer scheme, FOSSA procedures can be applied locally on each worker on a subset of systems and then once more on the master on all systems returned by the workers.
If the simulation runs are also parallelized via a divide-and-conquer scheme, then each worker can perform screening locally after finishing its runs, without needing to await the results from other workers' runs.

The main result we prove (Theorem~\ref{theorem: Divide and Conquer}) states that the final set of systems returned by the master when applying FOSSA procedures under a divide-and-conquer scheme is the same as would have been obtained by applying that FOSSA procedure on the same output data on a single processor.
The result holds for any definition of acceptability and any FOSSA procedure satisfying Assumptions~\ref{assump: Acceptability}--\ref{assump: closed & bounded}.

For notational convenience, we extend the definitions of the acceptable region, confidence region, and returned subset to handle the situation in which only a subset of systems are being screened.
Let $\mathcal{X}=\{1,2,\dots,k\}$ denote the set of indices of all systems under consideration and let $\mathcal{L}\subseteq\mathcal{X}$ be a subset of systems.
Under the divide-and-conquer scheme, a worker would receive the data for some subset of systems, $\mathcal{L}$.
For any $i\in\mathcal{L}$, define
\begin{align*}
\mathds{A}_{i}(\mathcal{L}) &\equiv \{\mathbf{M}\in \mathds{R}^{|\mathcal{L}| \times d}\colon \boldsymbol{m}_{i}\in\mathds{F} \text{ and }  \boldsymbol{m}_{i}\nsucc \boldsymbol{m}_{j} \text{ for all } j\in\mathcal{L} \text{ and } j \neq i \}, \quad \text{and} \\
\mathds{C}_{i}(\mathcal{L}) &\equiv \bigtimes\limits_{j\in\mathcal{L}}\mathds{C}_{ij}
\end{align*}
where $|\mathcal{L}|$ is the cardinality of $\mathcal{L}$.
Note that the matrix $\mathbf{M}$ in the definition of $\mathds{A}_i(\mathcal{L})$ is of lower dimension ($|\mathcal{L}| \times d$) than the overall configuration ($k \times d)$ and that the response vectors are indexed with respect to $\mathcal{X}$ instead of $\mathcal{L}$.
In words, $\mathds{A}_i(\mathcal{L})$ is the set of configurations of the response vectors of systems in $\mathcal{L}$ for which system $i$ is acceptable relative to the other systems in $\mathcal{L}$, and $\mathds{C}_i(\mathcal{L})$ is the Cartesian product of the $(1-\alpha)^{1/k}$ confidence regions for $\boldsymbol{\mu}_{i}$ for $i \in \mathcal{L}$ based on there being $k$ (not $|\mathcal{L}|$) systems under consideration.
The set of systems returned when applying the FOSSA procedure to only systems in $\mathcal{L}$ is then
$$
\mathcal{S}(\mathcal{L}) \equiv \{i\in\mathcal{L}\colon \mathds{A}_{i}(\mathcal{L})\cap\mathds{C}_{i}(\mathcal{L})\neq\emptyset \}.
$$

Let $\mathds{W}^{+}(\boldsymbol{m})\equiv\{\boldsymbol{m}'\colon \boldsymbol{m}'\succ \boldsymbol{m}\}$ denote the set of response vectors that are less preferable than $\boldsymbol{m}$ and $\mathds{W}^{+}(\mathds{C}_{ij})\equiv\bigcap\limits_{\boldsymbol{m}_{ij}\in\mathds{C}_{ij}} \mathds{W}^{+}(\boldsymbol{m}_{ij})$ denote the set of response vectors that are less preferable than all response vectors in $\mathds{C}_{ij}$. A consequence of these definitions is that if
$\boldsymbol{m}\notin\mathds{W}^{+}(\mathds{C}_{ij})$, then there exists some $\boldsymbol{m_{ij}}\in\mathds{C}_{ij}$ such that $\boldsymbol{m}\nsucc\boldsymbol{m_{ij}}$.

We proceed to prove Theorem~\ref{theorem: Divide and Conquer} via four lemmas: Lemmas~\ref{lemma: FOSSA on a subset}--\ref{lemma: Parallel necessity}.

\begin{lemma}
\label{lemma: FOSSA on a subset}
$\mathds{A}_{i}(\mathcal{L})\cap\mathds{C}_{i}(\mathcal{L})\neq\emptyset$ if and only if $\mathds{C}_{ii}\cap\mathds{F}\cap\left(\bigcup\limits_{j\in\mathcal{L}\backslash\{i\}}\mathds{W}^{+}(\mathds{C}_{ij})\right)^{c}\neq\emptyset$.
\end{lemma}
\noindent {\bf Proof of Lemma~\ref{lemma: FOSSA on a subset}:}
We first prove the forward implication by contradiction.
Assume that $\mathds{A}_{i}(\mathcal{L})\cap\mathds{C}_{i}(\mathcal{L})\neq\emptyset$ and $\mathds{C}_{ii}\cap\mathds{F}\cap\left(\bigcup\limits_{j\in\mathcal{L}\backslash\{i\}}\mathds{W}^{+}(\mathds{C}_{ij})\right)^{c}=\emptyset$.
Because $\mathds{A}_{i}(\mathcal{L})\cap\mathds{C}_{i}(\mathcal{L})\neq\emptyset$, there exists some $\mathbf{M}'\in\mathds{A}_{i}(\mathcal{L})\cap\mathds{C}_{i}(\mathcal{L})$, where $\boldsymbol{m}_{ii}'\in\mathds{F}$ and $\boldsymbol{m}_{ii}'\nsucc\boldsymbol{m}_{ij}'$ for all $j\in\mathcal{L}\backslash\{i\}$.
For all $j\in\mathcal{L}\backslash\{i\}$, $\boldsymbol{m}'_{ij}\in\mathds{C}_{ij}$ and $\boldsymbol{m}_{ii}'\notin  \mathds{W}^{+}(\boldsymbol{m}'_{ij})$, hence
$\boldsymbol{m}_{ii}'\notin\mathds{W}^{+}(\mathds{C}_{ij})$.
Therefore, $\boldsymbol{m}_{ii}'\in\left(\bigcup\limits_{j\in\mathcal{L}\backslash\{i\}}\mathds{W}^{+}(\mathds{C}_{ij})\right)^{c}$.
Because $\boldsymbol{m}_{ii}'\in\mathds{C}_{ii}\cap\mathds{F}$, we have that $\boldsymbol{m}_{ii}' \in \mathds{C}_{ii}\cap\mathds{F}\cap\left(\bigcup\limits_{j\in\mathcal{L}\backslash\{i\}}\mathds{W}^{+}(\mathds{C}_{ij})\right)^{c}$ and thus $\mathds{C}_{ii}\cap\mathds{F}\cap\left(\bigcup\limits_{j\in\mathcal{L}\backslash\{i\}}\mathds{W}^{+}(\mathds{C}_{ij})\right)^{c}\neq\emptyset$.

We next prove the reverse implication. Assume that $\mathds{C}_{ii}\cap\mathds{F}\cap\left(\bigcup\limits_{j\in\mathcal{L}\backslash\{i\}}\mathds{W}^{+}(\mathds{C}_{ij})\right)^{c}\neq\emptyset$. Then there exists some $\boldsymbol{m}'_{ii}\in \mathds{C}_{ii}\cap\mathds{F}\cap\left(\bigcup\limits_{j\in\mathcal{L}\backslash\{i\}}\mathds{W}^{+}(\mathds{C}_{ij})\right)^{c}$ and $\boldsymbol{m}'_{ii}\notin \mathds{W}^{+}(\mathds{C}_{ij})$ for all $j\in\mathcal{L}\backslash\{i\}$.
Hence, for all $j\in\mathcal{L}\backslash\{i\}$, we can find a
$\boldsymbol{m}'_{ij}\in\mathds{C}_{ij}$ such that $\boldsymbol{m}'_{ii}\nsucc\boldsymbol{m}'_{ij}$.
By the definitions of $\mathds{A}_{i}(\mathcal{L})$ and $\mathds{C}_{i}(\mathcal{L})$, $\bigtimes\limits_{j\in\mathcal{L}}\boldsymbol{m}'_{ij}\in\mathds{A}_{i}(\mathcal{L})$ and $\bigtimes\limits_{j\in\mathcal{L}}\boldsymbol{m}'_{ij}\in\mathds{C}_{i}(\mathcal{L})$.
Therefore, $\mathds{A}_{i}(\mathcal{L})\cap\mathds{C}_{i}(\mathcal{L})\neq\emptyset$.
$\square$

\begin{lemma}
\label{lemma: Parallel Sufficency}
For any $\mathcal{L}\subseteq\mathcal{X}$, $\mathcal{L}\backslash\mathcal{S}(\mathcal{L})\subseteq\mathcal{X}\backslash\mathcal{S}(\mathcal{X})$.
\end{lemma}

\noindent {\bf Proof of Lemma~\ref{lemma: Parallel Sufficency}:}
Fix a subset $\mathcal{L} \subseteq \mathcal{X}$ and an arbitrary $i\in \mathcal{L}\backslash\mathcal{S}(\mathcal{L})$.
By the definition of $\mathcal{S}(\mathcal{L})$,
$\mathds{A}_{i}(\mathcal{L})\cap\mathds{C}_{i}(\mathcal{L})=\emptyset$ and by Lemma~\ref{lemma: FOSSA on a subset}, $\mathds{C}_{ii}\cap\mathds{F}\cap\left(\bigcup\limits_{j\in\mathcal{L}\backslash\{i\}}\mathds{W}^{+}(\mathds{C}_{ij})\right)^{c}=\emptyset$.
We thus have that
\begin{align*}
\mathds{C}_{ii}\cap\mathds{F}\cap\left(\bigcup\limits_{j\in\mathcal{X}\backslash\{i\}}\mathds{W}^{+}(\mathds{C}_{ij})\right)^{c} &= \mathds{C}_{ii}\cap\mathds{F}\cap\left(\bigcup\limits_{j\in\mathcal{X}\backslash\mathcal{L}}\mathds{W}^{+}(\mathds{C}_{ij})\cup\bigcup\limits_{j\in\mathcal{L}\backslash\{i\}}\mathds{W}^{+}(\mathds{C}_{ij})\right)^{c} \\
    &= \mathds{C}_{ii}\cap\mathds{F}\cap\left(\bigcup\limits_{j\in\mathcal{L}\backslash\{i\}}\mathds{W}^{+}(\mathds{C}_{ij})\right)^{c}\cap\left(\bigcup\limits_{j\in\mathcal{X}\backslash\mathcal{L}}\mathds{W}^{+}(\mathds{C}_{ij})\right)^{c}
    &= \emptyset.
\end{align*}
Applying Lemma~\ref{lemma: FOSSA on a subset} (with the set $\mathcal{X}$) then implies that 
$\mathds{A}_{i}(\mathcal{X})\cap\mathds{C}_{i}(\mathcal{X})=\emptyset$ and  
$i\notin\mathcal{S}(\mathcal{X})$.
Because the choice of $i\in \mathcal{L}\backslash\mathcal{S}(\mathcal{L})$ and $\mathcal{L}$ were arbitrary, we have shown that for any $\mathcal{L}\subseteq\mathcal{X}$, $\mathcal{L}\backslash\mathcal{S}(\mathcal{L})\subseteq\mathcal{X}\backslash\mathcal{S}(\mathcal{X})$.
$\square$

\begin{lemma}\label{lemma: W+ nested relationship}
For any $j,j'\in\mathcal{X}$, $j \neq j'$, if $\mathds{C}_{ij}\cap\mathds{W}^{+}(\mathds{C}_{ij'})\neq\emptyset$, then $\mathds{W}^{+}(\mathds{C}_{ij}) \subset \mathds{W}^{+}(\mathds{C}_{ij'})$.
\end{lemma}
\noindent {\bf Proof of Lemma~\ref{lemma: W+ nested relationship}:} 
$\mathds{C}_{ij}\cap\mathds{W}^{+}(\mathds{C}_{ij'})\neq\emptyset$ implies that there exists some $\boldsymbol{m}'_{ij}\in \mathds{C}_{ij}\cap\mathds{W}^{+}(\mathds{C}_{ij'})$.
By the transitivity of acceptability, $\mathds{W}^{+}(\boldsymbol{m}'_{ij})\subseteq\mathds{W}^{+}(\mathds{C}_{ij'})$.
Hence, $\mathds{W}^{+}(\mathds{C}_{ij}) \subset \mathds{W}^{+}(\mathds{C}_{ij'})$, where the subset is proper because $\boldsymbol{m}'_{ij}\in\mathds{W}^{+}(\mathds{C}_{ij'})$ is not in $\mathds{W}^{+}(\mathds{C}_{ij'})$.
$\square$

\begin{lemma}\label{lemma: Parallel necessity}
For any $\mathcal{S}(\mathcal{X})\subseteq\mathcal{L}\subseteq\mathcal{X}$, $\mathcal{S}(\mathcal{L})\subseteq\mathcal{S}(\mathcal{X})$.
\end{lemma}

\noindent {\bf Proof of Lemma~\ref{lemma: Parallel necessity}:} 
It suffices to prove that for any $i \in \mathcal{L}\backslash\mathcal{S}(\mathcal{X})$, $i \notin \mathcal{S}(\mathcal{L})$.
By Lemma~\ref{lemma: FOSSA on a subset}, $i \notin \mathcal{S}(\mathcal{L})$ if and only if $\mathds{C}_{ii}\cap\mathds{F}\cap\left(\bigcup\limits_{j\in\mathcal{L}\backslash\{i\}}\mathds{W}^{+}(\mathds{C}_{ij})\right)^{c}=\emptyset$. The event on the right-hand side can re-expressed as
\[\mathds{C}_{ii}\cap\mathds{F}\cap\left(\bigcup\limits_{j\in\mathcal{L}\backslash(\mathcal{S}(\mathcal{X})\cup\{i\})}\mathds{W}^{+}(\mathds{C}_{ij})\right)^{c}\cap\left(\bigcup\limits_{j\in\mathcal{S}(\mathcal{X})}\mathds{W}^{+}(\mathds{C}_{ij})\right)^{c} = \emptyset,\]
which is implied by the event
\begin{equation}
\label{eqn:parallel_necessity_key_event}
\mathds{C}_{ii}\cap\mathds{F}\cap\left(\bigcup\limits_{j\in\mathcal{S}(\mathcal{X})}\mathds{W}^{+}(\mathds{C}_{ij})\right)^{c} = \emptyset.
\end{equation}

We proceed to prove via contradiction that for any $i \in \mathcal{L}\backslash\mathcal{S}(\mathcal{X})$, (\ref{eqn:parallel_necessity_key_event}) holds, which will imply that $i \notin \mathcal{S}(\mathcal{L})$.
Fix an arbitrary $i \in \mathcal{L}\backslash\mathcal{S}(\mathcal{X})$ and suppose that
\begin{equation}
\label{eqn:parallel_necessity_key_event_contra}
\mathds{C}_{ii}\cap\mathds{F}\cap\left(\bigcup\limits_{j\in\mathcal{S}(\mathcal{X})}\mathds{W}^{+}(\mathds{C}_{ij})\right)^{c} \neq \emptyset.
\end{equation}
Because $i\notin \mathcal{S}(\mathcal{X})$, Lemma~\ref{lemma: FOSSA on a subset} implies that 
\begin{equation}
\label{eqn:parallel_necessity_Lemma5_result}
\mathds{C}_{ii}\cap\mathds{F}\cap\left(\bigcup\limits_{j\in\mathcal{X}\backslash\{i\}}\mathds{W}^{+}(\mathds{C}_{ij})\right)^{c}=\emptyset.
\end{equation}
From (\ref{eqn:parallel_necessity_key_event_contra}), we have that $\mathds{C}_{ii}\cap\mathds{F}\neq\emptyset$.
Then because of (\ref{eqn:parallel_necessity_Lemma5_result}),
$\mathds{C}_{ii}\cap\mathds{F}\subseteq\bigcup\limits_{j\in\mathcal{X}\backslash\{i\}}\mathds{W}^{+}(\mathds{C}_{ij}$).

We next claim that there exists some $i'\in\mathcal{X}\backslash(\mathcal{S}(\mathcal{X})\cup\{i\})$ such that 
\begin{equation}
\label{eqn:system_b_condition1}
\mathds{C}_{ii}\cap\mathds{F}\cap\mathds{W}^{+}(\mathds{C}_{ii'})\neq\emptyset
\end{equation}
and
\begin{equation}
\label{eqn:system_b_condition2}
\mathds{W}^{+}(\mathds{C}_{ii'})\nsubseteq\bigcup\limits_{j\in\mathcal{S}(\mathcal{X})}\mathds{W}^{+}(\mathds{C}_{ij}).
\end{equation}
To the contrary, suppose that for 
all $i'\in\mathcal{X}\backslash(\mathcal{S}(\mathcal{X})\cup\{i\})$, either
$\mathds{C}_{ii} \cap \mathds{F} \cap \mathds{W}^{+}(\mathds{C}_{ii'}) = \emptyset$ or $\mathds{W}^{+}(\mathds{C}_{ii'})\subseteq\bigcup\limits_{j\in\mathcal{S}(\mathcal{X})}\mathds{W}^{+}(\mathds{C}_{ij})$.
Let
$$\mathcal{I}_1 = \left\{i'\in\mathcal{X}\backslash(\mathcal{S}(\mathcal{X})\cup\{i\}) \colon \mathds{W}^{+}(\mathds{C}_{ii'})\subseteq\bigcup\limits_{j\in\mathcal{S}(\mathcal{X})}\mathds{W}^{+}(\mathds{C}_{ij})\right\}$$
and
$$\mathcal{I}_2 = \left\{i'\in\mathcal{X}\backslash(\mathcal{S}(\mathcal{X})\cup\{i\}) \colon \mathds{C}_{ii} \cap \mathds{F} \cap \mathds{W}^{+}(\mathds{C}_{ii'}) = \emptyset \text{ and } \mathds{W}^{+}(\mathds{C}_{ii'})\nsubseteq\bigcup\limits_{j\in\mathcal{S}(\mathcal{X})}\mathds{W}^{+}(\mathds{C}_{ij}) \right\}.$$
By definition, $\mathcal{I}_1 \cup \mathcal{I}_2 = \mathcal{X}\backslash(\mathcal{S}(\mathcal{X}) \cup \{i\})$ and $\mathcal{I}_1 \cap \mathcal{I}_2 = \emptyset$.
Furthermore, for those $i' \in \mathcal{I}_2$, $\mathds{C}_{ii} \cap \mathds{F} \subseteq \left(\mathds{W}^{+}(\mathds{C}_{ii'})\right)^c$.
All together, we have that
\begin{align}
\notag
\mathds{C}_{ii}\cap\mathds{F}\cap\bigcap\limits_{j\in\mathcal{X}\backslash\{i\}}\left(\mathds{W}^{+}(\mathds{C}_{ij})\right)^{c} &=\mathds{C}_{ii}\cap\mathds{F}\cap\bigcap\limits_{j\in\mathcal{I}_1}\left(\mathds{W}^{+}(\mathds{C}_{ij})\right)^{c}\cap\bigcap\limits_{j\in\mathcal{I}_2}\left(\mathds{W}^{+}(\mathds{C}_{ij})\right)^{c}\cap\bigcap\limits_{j\in\mathcal{S}(\mathcal{X})}\left(\mathds{W}^{+}(\mathds{C}_{ij})\right)^{c} \\
\label{eqn:system_b_contra}
&=\mathds{C}_{ii}\cap\mathds{F}\cap\bigcap\limits_{j\in\mathcal{S}(\mathcal{X})}\left(\mathds{W}^{+}(\mathds{C}_{ij})\right)^{c},
\end{align}
where in the second equality we have used the facts that 
$\bigcap\limits_{j\in\mathcal{S}(\mathcal{X})}\left(\mathds{W}^{+}(\mathds{C}_{ij})\right)^{c} \subseteq \bigcap\limits_{j\in\mathcal{I}_1}\left(\mathds{W}^{+}(\mathds{C}_{ij})\right)^{c}$
and $\mathds{C}_{ii}\cap\mathds{F} \subseteq \bigcap\limits_{j \in \mathcal{I}_2} \left(\mathds{W}^{+}(\mathds{C}_{ij})\right)^{c}$.
This leads to a contradiction, since (\ref{eqn:parallel_necessity_Lemma5_result}) states that the left-hand-side of (\ref{eqn:system_b_contra}) is empty, whereas
(\ref{eqn:parallel_necessity_key_event_contra}) states that the set on the right-hand side of (\ref{eqn:system_b_contra}) is non-empty.

The remainder of the proof involves showing that $\mathds{C}_{ii'}\cap\mathds{F}\cap\left(\bigcup\limits_{j\in\mathcal{S}(\mathcal{X})}\mathds{W}^{+}(\mathds{C}_{ij})\right)^{c}\neq\emptyset$ and leveraging this result to build a chain of implications on the existence of more and more different systems in $\mathcal{X}\backslash\mathcal{S}(\mathcal{X})$ until we eventually exhaust the set and obtain a contradiction.

For this particular $i'$, we claim that $\mathds{C}_{ii'}\cap\mathds{F}\neq\emptyset$, i.e., System $i'$ is plausibly feasible. 
Again, we employ a proof by contradiction and assume instead that $\mathds{C}_{ii'}\cap\mathds{F}=\emptyset$. For any $\boldsymbol{m}_{ii'}\in\mathds{C}_{ii'}$, $\boldsymbol{m}_{ii'}\in\mathds{F}^{c}$, and thus by Assumption~\ref{assump: Acceptability}, $\mathds{W}^{+}(\boldsymbol{m}_{ii'})=\emptyset$ since System $i'$ is infeasible.
Therefore, we have $\mathds{W}^{+}(\mathds{C}_{ii'})\equiv\bigcap\limits_{\boldsymbol{m}_{ii'}\in\mathds{C}_{ii'}} \mathds{W}^{+}(\boldsymbol{m}_{ii'})=\emptyset$. This, however, contradicts the fact that $\mathds{F}\cap\mathds{W}^{+}(\mathds{C}_{ii'})\neq\emptyset$ from (\ref{eqn:system_b_condition1}). Hence, $\mathds{C}_{ii'}\cap\mathds{F}=\emptyset$.

Additionally, for any $\boldsymbol{m}_{ii'}\in\mathds{C}_{ii'}$, $\boldsymbol{m}_{ii'}\notin\mathds{W}^{+}(\boldsymbol{m}_{ii'})$ by the definition of $\mathds{W}^{+}(\cdot)$. Hence $\mathds{C}_{ii'}\cap\mathds{W}^{+}(\mathds{C}_{ii'})=\emptyset$. 

We next show that because $i'\notin\mathcal{S}(\mathcal{X})$,
\begin{equation}
\label{eqn:system_b_lemma_5_like}
\mathds{C}_{ii'}\cap\mathds{F}\cap\left(\bigcup\limits_{j\in\mathcal{S}(\mathcal{X})}\mathds{W}^{+}(\mathds{C}_{ij})\right)^{c}\neq\emptyset. 
\end{equation}
Suppose to the contrary that $\mathds{C}_{ii'}\cap\mathds{F}\cap\left(\bigcup\limits_{j\in\mathcal{S}(\mathcal{X})}\mathds{W}^{+}(\mathds{C}_{ij})\right)^{c}=\emptyset$.
We have already proven that $\mathds{C}_{ii'}\cap\mathds{F}\neq\emptyset$, thus, $\mathds{C}_{ii'}\cap\mathds{F}\cap\left(\bigcup\limits_{j\in\mathcal{S}(\mathcal{X})}\mathds{W}^{+}(\mathds{C}_{ij})\right)^{c}=\emptyset$ implies that $(\mathds{C}_{ii'}\cap\mathds{F})\subseteq\bigcup\limits_{j\in\mathcal{S}(\mathcal{X})}\mathds{W}^{+}(\mathds{C}_{ij}).$
Using a similar argument as before, this implies that there exists some $j'\in\mathcal{S}(\mathcal{X})$ such that $\mathds{C}_{ii'}\cap\mathds{F}\cap\mathds{W}^{+}(\mathds{C}_{ij'})\neq\emptyset$.
Since $\mathds{C}_{ii'}\cap\mathds{W}^{+}(\mathds{C}_{ij'}) \neq \emptyset$, by Lemma~\ref{lemma: W+ nested relationship}, we have $\mathds{W}^{+}(\mathds{C}_{ii'})\subset\mathds{W}^{+}(\mathds{C}_{ij'})\subseteq\bigcup\limits_{j\in\mathcal{S}(\mathcal{X})}\mathds{W}^{+}(\mathds{C}_{ij})$, which contradicts (\ref{eqn:system_b_condition2}).
Therefore, (\ref{eqn:system_b_lemma_5_like}) holds.

Because $i'\in\mathcal{X}\backslash(\mathcal{S}(\mathcal{X})\cup\{i\})$ and (\ref{eqn:system_b_lemma_5_like}) holds, we can make a similar argument to that made for (\ref{eqn:system_b_condition1}) and (\ref{eqn:system_b_condition2}).
In particular, there exists some $i''\in\mathcal{X}\backslash(\mathcal{S}(\mathcal{X})\cup\{i'\})$ such that 
\begin{equation}
\label{eqn:system_b_condition1_like}
\mathds{C}_{ii'}\cap\mathds{F}\cap\mathds{W}^{+}(\mathds{C}_{ii''})\neq\emptyset   
\end{equation}
 and $\mathds{W}^{+}(\mathds{C}_{ii''})\nsubseteq\bigcup\limits_{j\in\mathcal{S}(\mathcal{X})}\mathds{W}^{+}(\mathds{C}_{ij})$.
This application of the argument is slightly different from the previous because we can show by contradiction that $i'' \neq i$.
Suppose that $i'' =i$. 
We know from (\ref{eqn:system_b_condition1}) that $\mathds{C}_{ii}\cap\mathds{W}^{+}(\mathds{C}_{ii'})\neq\emptyset$, which implies $\mathds{W}^{+}(\mathds{C}_{ii}) \subset \mathds{W}^{+}(\mathds{C}_{ii'})$ by Lemma~\ref{lemma: W+ nested relationship}.
Then 
$$\mathds{C}_{ii'}\cap\mathds{F}\cap\mathds{W}^{+}(\mathds{C}_{ii''})=\mathds{C}_{ii'}\cap\mathds{F}\cap\mathds{W}^{+}(\mathds{C}_{ii})\subseteq\mathds{C}_{ii'}\cap\mathds{F}\cap\mathds{W}^{+}(\mathds{C}_{ii'})=\emptyset,$$
which contradicts (\ref{eqn:system_b_condition1_like}).
Therefore $i'' \neq i$ and $i''\in\mathcal{X}\backslash(\mathcal{S}(\mathcal{X})\cup\{i,i'\})$.
Because we cannot repeat this process an infinite number of times---as $\mathcal{X}\backslash\mathcal{S}(\mathcal{X})$ is a finite set---we will eventually arrive at a contradiction. Therefore, we have shown via contradiction that $i \in \mathcal{S}(\mathcal{L})$.
$\square$

\bigskip

\noindent {\bf Theorem~\ref{theorem: Divide and Conquer}.}
\emph{If Assumptions~\ref{assump: Acceptability}--\ref{assump: closed & bounded} hold and confidence regions $\mathds{C}_{ij}$ for all $i,j=1,2,\dots,k$ are held fixed, FOSSA procedures can be parallelized using a divide-and-conquer scheme without compromising their screening power.}

\bigskip

\noindent {\bf Proof of Theorem~\ref{theorem: Divide and Conquer}:} 
Let $\mathcal{S}(\mathcal{X})$ denote the set of systems that would be returned if a given FOSSA procedure were applied on a single processor.
In the parallel computing setting, let $w$ be the number of worker processors and let $\{\mathcal{L}_1, \mathcal{L}_2, \ldots, \mathcal{L}_w\}$ be a partition of $\mathcal{X}$, i.e., $\bigcup_{v=1}^{w}\mathcal{L}_{v}=\mathcal{X}$ and $\mathcal{L}_{v} \cap \mathcal{L}_{v'} = \emptyset$ for all $v, v' \in \mathcal{X}$, $v \neq v'$. Worker processor $v$ receives the necessary data pertaining to the systems in $\mathcal{L}_v$ and applies the FOSSA procedure to $\mathcal{L}_v$, returning a subset $\mathcal{S}(\mathcal{L}_v)$ to the master processor. 
By Lemma~\ref{lemma: Parallel Sufficency}, $\mathcal{S}(\mathcal{X})\subseteq\bigcup_{v=1}^{w}\mathcal{S}(\mathcal{L}_{v})$.
On the master processor, the same FOSSA procedure is applied on the union of the subsets returned by all worker processors, and the set $\mathcal{S}(\bigcup_{v=1}^{w}\mathcal{S}(\mathcal{L}_{v}))$ is ultimately returned.
By Lemma~\ref{lemma: Parallel necessity}, $\mathcal{S}(\bigcup_{v=1}^{w}\mathcal{S}(\mathcal{L}_{v}))=\mathcal{S}(\mathcal{X})$. Thus, there is no loss of screening power.
$\square$

\section{Miscellaneous Results}
\label{sec:misc_results}

\subsection{A Confidence Region for the Response Vectors of Acceptable Systems}

FOSSA procedures that use common confidence regions when screening all systems can also provide a confidence region for the response vectors of all acceptable systems.

\begin{proposition}
$\bigcup_{j\in\mathcal{S}}\mathds{C}_{uj}\cap\left(\bigcup\limits_{h=1}^{k} \mathds{W}(\mathds{C}_{uh})\right)^{c}$ is a $1-\alpha$ confidence region for the response vectors of all acceptable systems, where  $\mathds{W}(\mathds{C}_{uh}):=\{\boldsymbol{m}\colon \boldsymbol{m}\succ\boldsymbol{m}_h \text{ for all }\boldsymbol{m}_h\in\mathds{C}_{uh}\}$.
\label{proposition: FOSSA Confidence region}
\end{proposition}
\noindent {\bf Proof of Proposition~\ref{proposition: FOSSA Confidence region}:}
\begin{align*}
&\mathrm{P}\left(\bigcap_{i\in\mathcal{A}(\mathbf{M}_{0})}\left\{\boldsymbol{\theta}_{i}\in\bigcup_{j\in\mathcal{S}}\mathds{C}_{uj}\cap\left(\bigcup\limits_{h=1}^{k} \mathds{W}(\mathds{C}_{uh})\right)^{c}\right\}\right)\\
&\geq \mathrm{P}\left(\bigcap_{i\in\mathcal{A}(\mathbf{M}_{0})}\left\{\boldsymbol{\theta}_{i}\in\mathds{C}_{ui}\cap \left(\bigcup\limits_{h=1}^{k} \mathds{W}(\mathds{C}_{uh})\right)^{c}\right\}\right)\\
&\geq \mathrm{P}\left(\bigcap_{i\in\mathcal{A}(\mathbf{M}_{0})} \left\{\{\boldsymbol{\theta}_{i}\in\mathds{C}_{ui}\} \cap \{\truemean \in \mathds{A}_{i} \cap \mathds{C}_{u}\}\right\}\right)\\
&= \mathrm{P}\left(\{\truemean \in\mathds{C}_{u}\}\cap\bigcap_{i\in\mathcal{A}(\mathbf{M}_{0})}\{\boldsymbol{\theta}_{i}\in\mathds{C}_{ui}\}\right)\\
&= \mathrm{P}\left(\{\truemean \in\mathds{C}_{u}\}\cap\left\{\bigtimes_{i=1}^{k}\boldsymbol{\theta}_{i}\in\bigtimes_{i=1}^{k}\mathds{C}_{ui}\right\}\right)\\
&= \mathrm{P}\left(\truemean \in\mathds{C}_{u}\right)\\
&\geq 1-\alpha. \quad \square
\end{align*}
The second inequality comes from the fact that for all $i\in\mathcal{A}(\mathbf{M}_{0})$, $\truemean \in \mathds{A}_{i} \cap \mathds{C}_{u}$ implies that $\boldsymbol{\theta}_{i}\in\left(\bigcup\limits_{h=1}^{k} \mathds{W}(\mathds{C}_{uh})\right)^c$.
To see this, assume by contradiction that there exists a System $i'\in\mathcal{A}(\mathbf{M}_{0})$ such that $\truemean \in \mathds{A}_{i'} \cap \mathds{C}_{u}$, but $\boldsymbol{\theta}_{i'}\in\bigcup\limits_{h=1}^{k} \mathds{W}(\mathds{C}_{uh})$.
Therefore, there exists a System $h'$ such that $\boldsymbol{\theta}_{i'}\in \mathds{W}(\mathds{C}_{uh'})$.
Meanwhile, $\truemean \in \mathds{C}_{u}$ implies that $\boldsymbol{\theta}_{h'}\in\mathds{C}_{uh'}$. Combining these two results, we have that $\boldsymbol{\theta}_{i'}\succ\boldsymbol{\theta}_{h'}$, meaning System $i'$ is unacceptable, which contradicts the assertion that $\truemean \in \mathds{A}_{i'}$.

When acceptability is defined as Pareto optimality, $\mathds{W}(\mathds{C}_{uj})=\mathds{W}(\boldsymbol{w}_j)$ for all $j=1,2,\dots,k$, where $\boldsymbol{w}_j$ is as defined in Section~\ref{subsec:pareto_ellipsoid}. As a result, a confidence region for the Pareto front of a bi-objective problem can be visualized by superimposing plots of $\mathds{C}_{uj}$ for all $j\in\mathcal{S}$ (i.e., plotting their union) and cutting out the part that overlaps with $\bigcup\limits_{h=1}^{k} \mathds{W}(\boldsymbol{w}_h)$.

\subsection{Definition of Phantom Pareto Systems}
\label{ec:phantom_pareto}

Let $\mathcal{P} = \{\boldsymbol{\nu}_1, \boldsymbol{\nu}_2, \ldots, \boldsymbol{\nu}_p\}$ be a set of vectors of length $d > 1$ describing a Pareto front. Let $\mathcal{K}=\{\boldsymbol{\kappa}\colon\boldsymbol{\kappa}\in\{1,2,\dots,d\}^{p} \}$ be the set of all vectors of length $p$ having elements in $\{1,2,\dots,d\}$.
For a given $\boldsymbol{\kappa} \in \mathcal{K}$, define the vector $\boldsymbol{\nu}^{\mathrm{bf}}(\boldsymbol{\kappa})= \left( \nu^{\mathrm{bf}}_{1}(\boldsymbol{\kappa}),\nu^{\mathrm{bf}}_{2}(\boldsymbol{\kappa}),\dots,\nu^{\mathrm{bf}}_{d}(\boldsymbol{\kappa}) \right)$ where 
\[ \nu_r^{\mathrm{bf}}(\boldsymbol{\kappa}) := \min_{\substack{\ell \in \{1, 2, \ldots, p\}\colon\\ \kappa_\ell = r}} \nu_{\ell r},\]
for $r = 1, 2, \ldots, d$
and bf stands for ``brute force" \citep{applegate2020multi}. 
In other words, for a given $\boldsymbol{\kappa}$ and $r$, we look at the indices of the components of $\boldsymbol{\kappa}$ equal to $r$ and take the minimum of the $r$th element over the  corresponding $\boldsymbol{\nu}$ vectors; if no components of $\boldsymbol{\kappa}$ are equal to $r$, then we set $\nu^{\mathrm{bf}}_{r}(\boldsymbol{\kappa})=\infty$.
The set of phantom Pareto systems is defined as
$$\left\{\boldsymbol{\nu}^{\mathrm{bf}}(\boldsymbol{\kappa})\colon \boldsymbol{\kappa}\in\mathcal{K} \text{ and } \nexists\boldsymbol{\kappa}'\in\mathcal{K} \ \mathrm{such\ that\ } \boldsymbol{\nu}^{\mathrm{bf}}(\boldsymbol{\kappa})\prec_{p}\boldsymbol{\nu}^{\mathrm{bf}}(\boldsymbol{\kappa}') \right\}.$$

\subsection{Measures of System Quality for Numerical Experiments}

In our numerical experiments, we adopt a metric that measures a system's distance to acceptability in terms of how much its responses would need to be adjusted to make the system acceptable. 

\paragraph{Optimization with Stochastic Constraints.}
For optimization with stochastic constraints, we employ a one-sided $L^{\infty}$ distance.
We define the quality of System $i$ to be 
$$q_{i}=\max\{\Tilde{\mu}_{i1}-\Tilde{\mu}_{1}^{*},0\}+\sum_{r = 2}^d \max\{\Tilde{\mu}_{ir}-\Tilde{\mu}^{\dagger}_r,0\},$$
where 
$$\Tilde{\mu}_{ir}=\frac{{\mu}_{ir}-\min\limits_{i}\mu_{ir}}{\max\limits_{i}\mu_{ir}-\min\limits_{i}\mu_{ir}}$$ for $r=1,2,\dots,d$ and
$$\Tilde{\mu}_{1}^{*}=\min\limits_{i}\Tilde{\mu}_{i1} \quad \text{ and } \quad \Tilde{\mu}^{\dagger}_r=\frac{{\mu}^{\dagger}_r -\min\limits_{i}\mu_{ir}}{\max\limits_{i}\mu_{ir}-\min\limits_{i}\mu_{ir}}$$
for $r = 2, 3, \ldots, d$.
This metric essentially captures the sum of the optimality gap and the sum of the violations of the individual stochastic constraints.
\paragraph{Pareto Optimality}
For multi-objective optimization, we employ a one-sided $L^2$ distance. The quality $q_{i}$ of System $i$ is defined as the Euclidean distance from the system's response vector to $\left(\bigcup_{j\in\mathcal{A}}\mathds{W}(\Tilde{\boldsymbol{\mu}}_{j})\right)^{c}$, the region that is not dominated by any systems. 
Here, $\Tilde{\boldsymbol{\mu}}_{i}$ is the normalized response vector, where 
$$\Tilde{\mu}_{ir}=\frac{{\mu}_{ir}-\min\limits_{i}\mu_{ir}}{\max\limits_{i}\mu_{ir}-\min\limits_{i}\mu_{ir}}$$ for $r=1,2,\dots,d$.
This measure quantifies how close a system's response vector is to being non-dominated by any other system.
Because $\left(\bigcup_{j\in\mathcal{A}}\mathds{W}(\Tilde{\boldsymbol{\mu}}_{j})\right)^{c}$ is a non-convex open set, we make use of phantom Pareto systems to decompose the set into a finite union of half-boxes. The minimum $L^2$ distance is then computed by solving a set of quadratic programs, one for each phantom Pareto system.
See Section~\ref{sec:pareto} for more details about how this approach was used for solving the optimization problems that arise in FOSSA Box and FOSSA Half-Box.
	\end{document}